\documentclass[aps,showpacs,amsmath,amssymb,twocolumn,pra,superscriptaddress,notitlepage]{revtex4-2}

\usepackage{qcircuit}
\usepackage{amsmath,bm}
\usepackage{graphicx}
\usepackage{amsmath,amssymb,amsthm,mathrsfs,amsfonts,dsfont}
\usepackage{subfigure, epsfig}
\usepackage{braket}
\usepackage{bm}
\usepackage{enumerate}
\usepackage{color}
\usepackage{graphicx}
\usepackage{algorithm}
\usepackage{algorithmic}
\usepackage{braket}
\usepackage{comment}
\usepackage{appendix}
\usepackage{here}
\usepackage{tabularx}
\usepackage{physics}
\usepackage{mathtools}
\usepackage{hyperref}

\theoremstyle{plain}
\newtheorem{thm}{Theorem}
\newtheorem{lemma}{Lemma}
\newtheorem*{thm*}{Theorem}
\newtheorem*{lemma*}{Lemma}
\newtheorem{cor}[thm]{Corollary}

\hypersetup{
           breaklinks=true,   
           colorlinks=true,   
           citecolor=blue, 
           urlcolor=blue
           }

\newcommand{\purple}[1]{\textcolor{black}{#1}}
\begin{document}


\title{Trade-offs between Quantum and Classical Resources in \\
the Linear Combination of Unitaries}


\author{Kaito Wada}
\email{wkai1013keio840@keio.jp}
\affiliation{Graduate School of Science and Technology, Keio University, 3-14-1 Hiyoshi, Kohoku, Yokohama, Kanagawa 223-8522, Japan}

\author{Hiroyuki Harada}
\affiliation{Graduate School of Science and Technology, Keio University, 3-14-1 Hiyoshi, Kohoku, Yokohama, Kanagawa 223-8522, Japan}

\author{Yasunari Suzuki}
\affiliation{NTT Computer and Data Science Laboratories, NTT Inc., Musashino 180-8585, Japan}

\author{Yuuki Tokunaga}
\affiliation{NTT Computer and Data Science Laboratories, NTT Inc., Musashino 180-8585, Japan}

\author{Naoki Yamamoto}
\affiliation{Department of Applied Physics and Physico-Informatics, Keio University, Hiyoshi 3-14-1, Kohoku, Yokohama 223-8522, Japan}

\author{Suguru Endo}
\email{suguru.endou@ntt.com}
\affiliation{NTT Computer and Data Science Laboratories, NTT Inc., Musashino 180-8585, Japan}
\affiliation{NTT Research Center for Theoretical Quantum Information, NTT Inc. 3-1 Morinosato Wakanomiya, Atsugi, Kanagawa, 243-0198, Japan}


\begin{abstract}
The randomized linear combination of unitaries (LCU) method with many applications to early fault-tolerant quantum computing algorithms has been proposed. 
This quantum algorithm computes the same expectation values as the original, fully coherent LCU algorithm using a shallower quantum circuit with a single ancilla qubit, at the cost of a quadratically larger sampling overhead. 
In this work, we propose a quantum algorithm intermediate between the original and randomized LCU that manages the trade-off between the sampling overhead and circuit complexity.
Our algorithm divides the set of unitary operators into several groups and then randomly samples LCU circuits from these groups to evaluate the target expectation value. 
Notably, we reveal that across all grouping strategies, the mechanism of the sampling overhead reduction can be solely characterized by a metric we call \textit{the reduction factor}.
Moreover, we analytically prove an underlying monotonicity of the reduction factor in the group size: larger group sizes entail smaller sampling overhead.
Finally, our framework enables a more flexible algorithmic design by systematically yielding intermediate implementations of LCU-based algorithms; we provide intermediate implementations of non-Hermitian dynamics simulation, ground-state property estimation, and quantum error detection. 
Besides, we demonstrate this principle by deriving intermediate trade-off scaling in sample complexity and ancillary space for quantum linear system solver.
\end{abstract}

\maketitle
\section{Introduction}

Quantum computers will play a crucial role in a broad range of tasks, from, e.g., the simulation of quantum systems such as chemical~\cite{aspuru2005simulated,cao2019quantum,mcardle2020quantum} and condensed matter systems~\cite{bauer2016hybrid,yoshioka2024hunting,vorwerk2022quantum} to machine learning~\cite{biamonte2017quantum,cerezo2022challenges,schuld2015introduction,mitarai2018quantum}. Many quantum algorithms have been proposed to simulate the real and imaginary time evolution~\cite{lloyd1996universal,Low2019hamiltonian,berry2015hamiltonian,motta2020determining}, general non-Hermitian dynamics~\cite{an2023linear,an2023quantum,endo2020variational}, obtain the spectrum of the Hamiltonian~\cite{kitaev1995quantum,mcclean2017hybrid,zeng2021universal}, and perform machine learning tasks~\cite{childs2017quantum,wang2024qubit,biamonte2017quantum,cerezo2022challenges}. Due to the emergence of large-scale quantum computers and early error-corrected quantum devices~\cite{kim2023evidence,google2025quantum}, many quantum algorithms tailored for the early fault-tolerant era~\cite{PRXQuantum.5.020101,campbell2021early,suzuki2022quantum} have also been proposed to address various problems of interest in addition to real-time simulation~\cite{Chakraborty2024implementingany,wang2024qubit,zeng2021universal,zeng2022simple,wan2022randomized,kato2024exponentially,huo2023error,faehrmann2022randomizing,yang2021accelerated}. 
Such algorithms \purple{involve coherent techniques for quantum speedups such as Hamiltonian simulation, assuming early fault tolerance of quantum devices.}

The typically employed operations in early fault-tolerant quantum computing (FTQC) algorithms, e.g., random sampling of controlled unitary operations and post-processing of measurement outcomes, are for realizing the linear combination of unitary operators~\cite{Chakraborty2024implementingany,wang2024qubit}. By appropriately choosing the sampling distribution and unitary operations, this algorithm performs various tasks such as Hamiltonian simulation~\cite{wan2022randomized,zeng2022simple}, quantum principal component analysis~\cite{wada2025state}, ground state property estimation~\cite{lin2022heisenberg,zeng2021universal,zhang2022computing,sun2024high}, quantum linear system solver~\cite{wang2024qubit,Chakraborty2024implementingany}, differential equation~\cite{an2023linear,an2023quantum}, and even error detection~\cite{tsubouchi2023virtual}.
For example, the random sampling of quantum circuits with the probability $p_i p_j$ with $\{p_i \}$ being the probability distribution in Fig.~\ref{fig:lcu_circuit} (b) is leveraged to compute the expectation value for the target state $K \rho_{\rm in} K^\dag/ \mathrm{tr}[K \rho_{\rm in} K^\dag ]$ with $K\propto \sum_{i=1}^m p_i U_i$ and the input state $\rho_{\rm in}$. The advantage of this method is that the depth of the quantum algorithm does not depend on the number of terms $m$, but incurs the unavoidable sampling overhead $\mathcal{O}(\mathcal{P}^{-2})$ for the projection probability $\mathcal{P}\propto \mathrm{tr}[K \rho_{\rm in} K^\dag ]$ due to delegation to classical processing.

Meanwhile, the conventional linear combination of unitaries (LCU) algorithm~\cite{childs2012hamiltonian,berry2015hamiltonian} is a fully quantum model for realizing the linear combination operator $K$; see Fig.~\ref{fig:lcu_circuit}(a). 
Although the LCU algorithm requires $\lceil \mathrm{log}_2 m \rceil$-qubit ancilla states and complicated controlled operations involving multiple-control qubits, the success probability $\mathcal{P}$ of obtaining the target state gives its sampling overhead {$\mathcal{O}(\mathcal{P}^{-1})$}, which is quadratically better than the early-FTQC implementation.

In this work, we propose a hybrid framework that balances 
the classical and quantum resources, i.e.,
the sample complexity and the circuit complexity of the early FTQC algorithms with the LCU operations. 
Crucially, for any implementation strategies within this framework, the sample complexity can be characterized by a universal metric we call \textit{the reduction factor} $R$.
In our hybrid framework, we first partition the set of unitary operators $\{U_i\}_{i=1}^m$ into subgroups. Then, we realize the linear combination of unitary operators in the subgroup via the conventional LCU algorithm. By performing the random LCU implementations for these operators, we can estimate the expectation values for the target state. We reveal analytically that the reduction in the sampling overhead enabled by the partial use of the coherent LCU is characterized by 
$R$, which gives the sampling overhead of the quantum algorithm $\mathcal{O}(R \mathcal{P}^{-2})$. 
The reduction factor $R$ increases monotonically as the set of unitary operators is divided into finer partitions, and continuously bridges the full coherent LCU $R=\mathcal{P}$ and the random LCU  $R=1$.




We then demonstrate that the hybrid framework provides a new principle for the design of early FTQC algorithms. 
Rather than merely illustrating the framework, we provide new intermediate analytical scalings for the quantum linear system solver~\cite{childs2017quantum,wang2024qubit,Chakraborty2024implementingany}, and new hybrid protocols with intermediate implementations for non-Hermitian dynamics simulation~\cite{an2023linear}, ground-state property estimation~\cite{ge2019faster,zeng2021universal}, and quantum error detection~\cite{tsubouchi2023virtual,endo2024projective,endo2025quantum,anai2024unitary}.
In all cases, these results arise systematically from the reduction factor analysis for application-specific groupings in our framework.

First, while the linear combination of Hamiltonian simulation (LCHS) realizes the non-Hermitian dynamics via a generalized Fourier transform of time-evolution operators with some appropriate filter function~\cite{an2023linear,an2023quantum}, we show that applying the random LCU at the tail of the filter function realizes $256$-fold depth and $7$-ancilla qubit reduction with the sampling overhead sustained.
Second, we analytically show that the appropriate grouping of unitary operators in the quantum linear system solver~\cite{childs2017quantum} realizes the intermediate scaling of the sampling overheads and the number of ancilla qubits with respect to the condition number $\kappa$ and the required accuracy $\varepsilon$.

Third, for the virtual ground-state preparation, we propose a new algorithm only necessitating repetitive use of one ancilla qubit and a controlled time evolution operator, which inherits the better scaling of virtual~\cite{zeng2021universal} and full coherent~\cite{ge2019faster} LCU implementation in a certain asymptotic regime.
Finally, because the quantum error detection can be considered as a special case of LCU, i.e., the projector to the code space is expanded as a linear combination of stabilizer unitary operators, we can use our hybrid LCU method. 
While the virtual quantum error detection (VQED)~\cite{tsubouchi2023virtual} can be performed with a constant depth but incurs the quadratically worse sampling cost, we show that applying VQED to errors that are rarely detected contributes negligibly to the sampling cost. 
Thus, it is preferable to employ the conventional quantum error detection for frequent errors, while VQED should be employed for low-rate errors.

\label{section: introduction}

\section{Problem setup and linear combination of unitaries (LCU)}\label{sec:pset_prevLCU}
In this work, we consider a problem to estimate the expectation value of an observable ${O}$ regarding a quantum state transformed by a completely positive map (CP map) ${\Lambda}(\bullet)$.
In particular, we focus on a CP map ${\Lambda}$ with a single linear operator $K$ acting on a target system ${\rm S}$, which can be written as
\begin{equation}\label{eq:target_cpmap}
    {\Lambda}(\bullet) = K\bullet K^\dagger.
\end{equation}
Then, our goal is to estimate the following value within an additive error $\varepsilon$ and with a high probability:
\begin{equation}\label{eq:target_ratio_est}
    \frac{{\rm tr}[{O}{\Lambda}(\rho)]}{{\rm tr}[{\Lambda}(\rho)]}=\frac{{\rm tr}[{O}K\rho K^\dagger]}{{\rm tr}[K\rho K^\dagger]},
\end{equation}
where $\rho$ is an input quantum state on ${\rm S}$.
Also, we focus on the estimation of the numerator 
\begin{equation}\label{eq:target_nume_est}
    {\rm tr}[{O}{\Lambda}(\rho)]={\rm tr}[{O}K\rho K^\dagger].
\end{equation}

The problem contains various quantum algorithms, especially tailored for early FTQC algorithms~\cite{PRXQuantum.5.020101,Chakraborty2024implementingany}.
The class of CP maps in Eq.~\eqref{eq:target_cpmap} can be simulated by the method of LCU~\cite{childs2012hamiltonian,berry2015hamiltonian} with the help of classical post-processing; this technique is a crucial component for the recent quantum algorithms. 
In the following, we review the LCU method and its variant with the clarification of the number of measurements required for {solving} our problem.

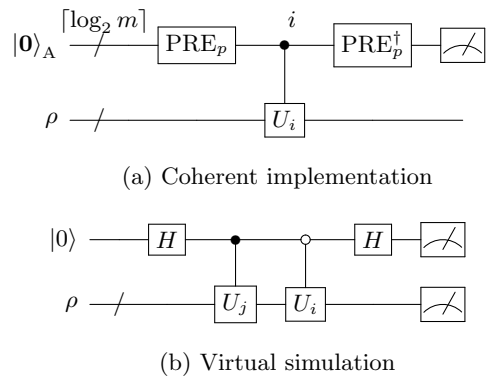
\begin{figure}[bt]
\centering
\begin{tabular}{c}
\\
\Qcircuit @C=1.2em @R=1.5em {
  \lstick{\ket{\bm{0}}_{\rm A}}&{/} \qw &\ustick{~~~~~~~~~~~~\lceil {\log_2 m\rceil}~~~~~~~~~~~~~~~~~i}\qw  & \gate{{\rm PRE}_{p}} & \ctrl{1} & \gate{{\rm PRE}_{p}^\dagger} &\meter \\
  \lstick{\rho}& {/} \qw & \qw & \qw & \gate{U_i} & \qw & \qw \\
  }
\\
\\
(a) Coherent implementation
\\
\\
\Qcircuit @C=1.2em @R=1.2em {
  \lstick{\ket{{0}}} & \qw & \gate{H} &\ctrl{1} & \ctrlo{1}&\gate{H}&\meter\\
  \lstick{\rho}& {/} \qw & \qw & \gate{U_j}&\gate{U_i} & \qw &\meter\\
  }
\\
\\
(b) Virtual simulation
\end{tabular}
\caption{Linear combination of unitaries (LCU) methods. The target operation is given by $K_{\rm LCU}=\sum_{i=1}^m p_i U_i$ for a probability distribution $\{p_i\}$ and unitary operators $\{U_i\}$.
The unitary ${\rm PRE}_p$ in (a) encodes $\{p_i\}$, and $K_{\rm LCU}\rho K^\dagger_{\rm LCU}$ is realized (up to normalization) when the final measurement after ${\rm PRE}^\dagger_p$ is all zero.
In (b), the indexes $i,j$ are independently sampled from the probability distribution $\{p_i\}$. By measuring $X$ in the top single qubit, we can recover the property of $K_{\rm LCU}\rho K_{\rm LCU}^\dagger$ without explicitly preparing it.}
\label{fig:lcu_circuit}
\end{figure}

\subsection{Coherent implementation}\label{sec:coherentlcu}

Let us consider a unitary decomposition of a linear operator $K$ given by
\begin{equation}\label{eq:lcu_before_normalize}
    K=\sum_{i=1}^m c_iU_i
\end{equation} 
for some non-negative coefficients 
$c:=\{c_i\}_{i=1}^m$ and unitary operators $\mathcal{U}:=\{U_i\}_{i=1}^m$.
Here, $m$ denotes the number of unitary operators.
By the singular value decomposition, we can always obtain a decomposition in Eq.~\eqref{eq:lcu_before_normalize} for known $K$~\cite{cui2012optimal}, while this procedure is computationally hard in large-size system.
\purple{Still, 
various problem-specific structures naturally suggest an efficient way to obtain such a decomposition of a target $K$; see Section~\ref{sec:application} or Ref.~\cite{Chakraborty2024implementingany}.}
We remark that when the coefficient $c_i$ takes a complex value, we can incorporate its phase into the unitary $U_i$. Thus, we can assume $c_i$ is non-negative without loss of generality.

To simulate the CP map $K\bullet K^\dagger$, the LCU method first implements the normalized version of Eq.~\eqref{eq:lcu_before_normalize}, 
that is, the CP map with the following operator $K_{\rm LCU}$:
\begin{equation}\label{eq:action_of_LCU}
    K_{\rm LCU}
    =
    \sum_{i=1}^m p_iU_i,~~~p_i:=\frac{c_i}{\|c\|_1}\geq 0,
\end{equation}
where $\|c\|_1$ denotes $L^1$ norm of the coefficient $c$ and $p:=\{p_i\}_{i=1}^m$ is a probability distribution satisfying $\sum_{i=1}^{m} p_i=1$.
For simplicity, we define $\tilde{\Lambda}_{p,\mathcal{U}}$ as the corresponding CP map 
\begin{align}\label{eq:CPmapByLCU}
    \tilde{\Lambda}_{p,\mathcal{U}}(\bullet)&:=K_{\rm LCU}(\bullet)K_{\rm LCU}^\dagger\notag\\
    &=\left(\sum_{i=1}^m p_i U_i\right)\bullet\left(\sum_{i=1}^m p_i U_i\right)^\dagger.
\end{align}
We remark that $\tilde{\Lambda}_{p,\mathcal{U}}$ recovers the CP map $K\bullet K^\dagger$ by multiplying $\|c\|^2_1$ as
\begin{equation}\label{eq:KandctildeLam}
    K\bullet K^\dagger = \left(\|c\|_1^2 \tilde{\Lambda}_{p,\mathcal{U}}\right)(\bullet).
\end{equation}


Hereafter, we focus on the coherent implementation of $\tilde{\Lambda}_{p,\mathcal{U}}$~\cite{childs2012hamiltonian,Low2019hamiltonian}.
To realize $\tilde{\Lambda}_{p,\mathcal{U}}$, two unitary gates, called {PREPARE} and {SELECT}, are usually employed.
The PREPARE operation is used to encode the coefficient $p$ and represented by a quantum circuit ${\rm PRE}_{p}$ that acts on an additionally introduced $\lceil {\log_2 m\rceil}$-qubit ancilla system ${\rm A}$ as
\begin{equation}
    {\rm PRE}_{p}:|\bm{0}\rangle \mapsto \sum_{i=1}^{m} \sqrt{p_i}|i\rangle,
\end{equation}
where $|\bm{0}\rangle$ and $|i\rangle$ denote an initial state and the computational basis in the ancilla system, respectively.
The circuit construction for the PREPARE operation requires $\mathcal{O}(m)$ elementary gates~\cite{Low2019hamiltonian}.
On the other hand, the SELECT operation acts on the whole system ${\rm A+S}$ between the ancilla and target qubits to encode the unitary operations $\mathcal{U}$ as follows.
\begin{equation}
    {\rm SEL}_{\mathcal{U}}=\sum_{i=1}^m|i\rangle\langle i|\otimes U_i.
\end{equation}
If we have access to each $U_i$ controlled by a single qubit, this SELECT operation can be constructed by $\mathcal{O}(m\log m)$ elementary gates (and a single-ancilla qubit), beyond these controlled-$U_i$ gates~\cite{zindorf2025efficient}. 

Now, we define a quantum circuit $L_{p,\mathcal{U}}$:
\begin{equation}\label{eq:LCUop}
    L_{p,\mathcal{U}}:=({\rm PRE}_p^\dagger\otimes \bm{1})\cdot {\rm SEL}_{\mathcal{U}}\cdot({\rm PRE}_p\otimes \bm{1}),
\end{equation}
where $\bm{1}$ denotes the identity operator.
Then, $K_{\rm LCU}$ can be implemented by $L_{p,\mathcal{U}}$ via the following relation: for any $\ket{\psi}$,
\begin{equation}\label{eq:LCUlemma}
    L_{p,\mathcal{U}}|\bm{0}\rangle|\psi\rangle=
    \ket{\bm{0}}\left(\sum_{i=1}^m p_iU_i\right)|\psi\rangle+\ket{\perp_{\psi}},
\end{equation}
where the unnormalized state $\ket{\perp_\psi}$ depends on $\ket{\psi}$ and is orthogonal to $|\bm{0}\rangle$.
This relation~\eqref{eq:LCUlemma}, which is known as the LCU lemma~\cite{berry2015hamiltonian}, indicates that we can realize the CP map in Eq.~\eqref{eq:CPmapByLCU} by using the system-ancilla interaction $L_{p,\mathcal{U}}$ followed by classical post-processing.
Using the circuit $L_{p,\mathcal{U}}$, we have 
\begin{equation}\label{eq:physical_pic_lcumap}
    \tilde{\Lambda}_{p,\mathcal{U}}(\rho) = {\rm tr}_{\rm A}\left[\Pi_{\rm A}\otimes \bm{1}\cdot L_{p,\mathcal{U}} (\Pi_{\rm A}\otimes \rho) L_{p,\mathcal{U}}^\dagger\right],
\end{equation}
where ${\rm tr}_{\rm A}$ denotes the partial trace over the ancilla system and $\Pi_{\rm A}:=\ket{\bm{0}}\bra{\bm{0}}$; see Fig.~\ref{fig:lcu_circuit}(a).
Therefore, we can effectively simulate the CP map $K\bullet K^\dagger$ with help of classical post-processing as
\begin{align}\label{eq:estimation_method_coherentlcu}
    {\rm tr}[OK\rho K^\dagger]
    &=\|c\|_1^2{\rm tr}\left[\Pi_{\rm A}\otimes O\cdot L_{p,\mathcal{U}} (\Pi_{\rm A}\otimes \rho) L_{p,\mathcal{U}}^\dagger\right],
\end{align}
for any observable $O$.
This provides an estimation scheme for the target values Eqs.~\eqref{eq:target_ratio_est} and~\eqref{eq:target_nume_est}.
We note that when we post-select $\Pi_{\rm A}$ in the ancilla system of Fig.~\ref{fig:lcu_circuit}(a), we can obtain the quantum state $\tilde{\Lambda}_{p,\mathcal{U}}(\rho)/{\rm tr}[\tilde{\Lambda}_{p,\mathcal{U}}(\rho)]$.

In terms of the measurement process, the success probability $\mathcal{P}$ for implementing $\tilde{\Lambda}_{p,\mathcal{U}}$ is given by 
\begin{equation}\label{eq:cLCU_successprob}
    \mathcal{P}=\left\|K_{\rm LCU}\ket{\psi}\right\|^2=\left\|\sum_i p_iU_i\ket{\psi}\right\|^2
\end{equation}
for any input state $\ket{\psi}$.
The success probability $\mathcal{P}$ can also be written as $\mathcal{P}={\rm tr}[{\tilde{\Lambda}_{p,\mathcal{U}}(\rho)}]$ for any input state $\rho$.
Importantly, the probability can be less than one because the operator $K_{\rm LCU}$ is not necessarily unitary.
In the following, we show that this probability affects the number of measurements required to estimate the target quantities.


A basic scheme to estimate the target value Eq.~\eqref{eq:target_ratio_est} is as follows~\cite{Chakraborty2024implementingany,wang2024qubit}.
Let $\varepsilon$ be an additive error, and let $\delta$ be a failure probability.
If we can perform the projective measurements on the target observable $O$ to $N'$ copies of $\tilde{\Lambda}_{p,\mathcal{U}}(\rho)/\mathcal{P}$, then the arithmetic mean of the measurement results gives an estimate.
\purple{By the Hoeffding's inequality, this estimate has the $\varepsilon$ error with probability at least $1-\delta$, given that $N'=\mathcal{O}(\|O\|^2\log(1/\delta)/\varepsilon^2)$.}
However, to prepare $N'$ copies of the state via the LCU method, we need $N=N'/\mathcal{P}$ copies of the input state $\rho$ in average.
Thus, using the LCU method to estimate Eq.~\eqref{eq:target_ratio_est}, the required number of measurements is given by
\begin{equation}\label{eq:sample_complexity_coherentLCU}
    \mathcal{O}\left(\frac{1}{\mathcal{P}}\times\frac{\|O\|^2\log(1/\delta)}{\varepsilon^2}\right).
\end{equation}
Also, we can estimate Eq.~\eqref{eq:target_ratio_est} by first separably measuring the numerator and denominator via Eq.~\eqref{eq:estimation_method_coherentlcu} and then calculating their ratio.
In this approach, we will prove the required number of measurements becomes the same scaling as Eq.~\eqref{eq:sample_complexity_coherentLCU} in the next section.

Finally, we here summarize the resources required for the LCU method with coherent implementation: for each circuit, $\mathcal{O}(m\log m)$ gates and $\mathcal{O}( \log m)$ ancilla qubits are required, beyond $U_i$ gates controlled by a single qubit.
As for the number of measurements, there is a linear dependence of $1/\mathcal{P}$ as in Eq.~\eqref{eq:sample_complexity_coherentLCU}, to estimate the target quantity~\eqref{eq:target_ratio_est}.
For the target quantity~\eqref{eq:target_nume_est}, please see the next section~\ref{sec:main_exp_val_est}.


\subsection{Virtual simulation}\label{sec:virtuallcu}

The original LCU method in the previous subsection can prepare the quantum state $\tilde{\Lambda}_{p,\mathcal{U}}(\rho)/\mathcal{P}$ in a coherent manner, but in estimating the target values, this direct state preparation is not necessarily needed.
In fact, the following \textit{virtual} LCU method~\cite{Chakraborty2024implementingany,wang2024qubit} can solve our problem with a reduction of the quantum resources in the original LCU, \purple{which uses} $\mathcal{O}(\log m)$ ancilla qubits and $\mathcal{O}(m\log m)$ elementary gates.

In the virtual LCU method, we first expand the sum of the CP map $\tilde{\Lambda}_{p,\mathcal{U}}$ as
\begin{equation}\label{eq:lcu_virtual}
    {\rm tr}\left[O\tilde{\Lambda}_{p,\mathcal{U}}(\rho)\right]=\sum_{i,j=1}^m p_ip_j{\rm Re}\left({\rm tr}\left[OU_i\rho U_j^\dagger\right]\right).
\end{equation}
This equality suggests that the value ${\rm tr}[O\tilde{\Lambda}_{p,\mathcal{U}}(\rho)]$ can be estimated via Monte-Carlo sampling combined with the Hadamard test, which requires only a constant number of ancilla qubits and elementary gates; see Fig.~\ref{fig:lcu_circuit}(b).
Specifically, we first randomly choose $U_i$ and $U_j$ with the probability distribution $p_ip_j$.
Then, we perform the projective measurement for $X\otimes O$ on the output state of Fig.~\ref{fig:lcu_circuit}(b) and obtain an output $\pm o_k$ for the eigenvalue $o_k$ of $O$.
By repeating this procedure several times independently and taking the arithmetic mean of these outputs, 
we can obtain an estimate of ${\rm tr}[O\tilde{\Lambda}_{p,\mathcal{U}}(\rho)]$.
Furthermore, using the estimates for ${\rm tr}[O\tilde{\Lambda}_{p,\mathcal{U}}(\rho)]$ and ${\rm tr}[\tilde{\Lambda}_{p,\mathcal{U}}(\rho)]$ (obtained by setting $O=\bm{1}$), we can estimate Eq.~\eqref{eq:target_ratio_est} by the ratio of these estimates.

In the final step, the division by $\mathcal{P}={\rm tr}[\tilde{\Lambda}_{p,\mathcal{U}}(\rho)]$ amplifies the error. 
Thus, for the value Eq.~\eqref{eq:target_ratio_est}, the required number of measurements for the procedure is given by 
\begin{equation}\label{eq:sample_complexity_virtualLCU}
    \mathcal{O}\left(\frac{1}{\mathcal{P}^2}\times \frac{\|O\|^2\log(1/\delta)}{\varepsilon^2 }\right),
\end{equation}
which is quadratically worse than the sample complexity Eq.~\eqref{eq:sample_complexity_coherentLCU} of the original LCU method with respect to the success probability $\mathcal{P}$.
Although this result is already proved in e.g., Ref.~\cite{Chakraborty2024implementingany}, our general results in the next section also contain the proof of Eq.~\eqref{eq:sample_complexity_virtualLCU} in addition to the proof of Eq.~\eqref{eq:sample_complexity_coherentLCU}.


In the next section, we provide a systematic way to \purple{manage} the trade-off between the quantum resources \purple{per circuit} and the sample complexity Eqs.~\eqref{eq:sample_complexity_coherentLCU} and \eqref{eq:sample_complexity_virtualLCU}.




\section{Main results}
\subsection{Hybrid simulation of LCU-based CP maps}

\subsubsection{Decomposition of $\tilde{\Lambda}_{p,\mathcal{U}}$}

\begin{figure}[tb]
 \centering
 \includegraphics[scale=1.1]{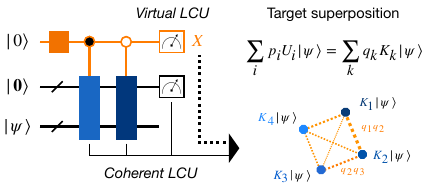}
 \caption{Idea of our decomposition for LCU-based CP map $\tilde{\Lambda}_{p,\mathcal{U}}$ or $\ket{\psi}\mapsto \sum_ip_i U_i\ket{\psi}$.
 To simulate the target state $\sum_{i} p_i U_i\ket{\psi}$, our method randomly uses the controlled version of the coherent LCU method for the partial superpositions $\{K_k|\psi\rangle\}$ and recovers the full state by measuring Pauli X in the top register as well as the virtual LCU method.
 The probability $q_kq_{k'}$ for virtually mixing 
 $K_k\ket{\psi}$ with $K_{k'}|\psi\rangle$ is determined by $K_k$.}
 \label{fig:virtual_quantum_superpos}
\end{figure}


We here introduce new methods to simulate LCU-based CP maps $\tilde{\Lambda}_{p,\mathcal{U}}$.
The core idea of our methods is to simulate the superposition over states $\{p_iU_i\ket{\psi}\}_{i=1}^m$ in the physical interpretation of Eq.~\eqref{eq:LCUlemma} using a \textit{hybrid} approach; also see Fig.~\ref{fig:virtual_quantum_superpos}.

To illustrate this idea, we first consider a simple case that we virtually simulate the superposition between three unnormalized states:
$$p_1 U_1\ket{\psi},~~\sum_{i=2}^{3} p_i U_i\ket{\psi},~~\mbox{and}~~\sum_{i=4}^m p_i U_i\ket{\psi}.$$
The superposition can be simulated with use of the following decomposition of $K_{\rm LCU}$:
\begin{equation}\label{eq:decomposition_example}
    K_{\rm LCU}=\sum_{i=1}^m p_iU_i = q_1K_{1}+q_2K_{2}+q_3K_{3},
\end{equation}
where $K_{k}$ and $q_k$ for $k=1,2,3$ are defined as 
{\color{black}
\begin{equation}\label{eq:weight_qk_def_Kk}
    K_k:=\sum_{i\in S_k}\frac{p_i}{q_k}U_i,~~~q_k = \sum_{i\in S_k} p_i.
\end{equation}
Here, $S_1=\{1\}$, $S_2=\{2,3\}$, and $S_3=\{4,5,...,m\}$; thus their union recovers the original set $[m]:=\{1,2,...,m\}$.
In the following observation}
\begin{align}\label{eq:decomp_example}
    \tilde{\Lambda}_{p,\mathcal{U}}(\rho)&=\left(\sum_{k=1}^{3} q_k K_{k}\right)\rho \left(\sum_{k'=1}^3 q_{k'}K_{k'}^\dagger\right)\notag\\
    &=\sum_{k,k'}q_kq_{k'} \frac{1}{2}\left(K_k\rho K_{k'}^\dagger+K_{k'}\rho K_{k}^\dagger\right),
\end{align}
we find that the unnormalized state $K_{k}\rho K_{k'}^\dag$ can be simulated by the coherent LCU method if $k=k'$.
Even in the case of $k\neq k'$, the combination of the virtual and coherent LCU methods allows us to effectively simulate $K_{k}\rho K_{k'}^\dagger+K_{k'}\rho K_{k}^\dagger$; see Fig.~\ref{fig:components_decomp_LCU}. Therefore, by randomly selecting these operations with probability $q_kq_{k'}$, we can recover the full superposition $K_{\rm LCU}\ket{\psi}=\sum_i p_i U_i\ket{\psi}$ with only use of (controlled version of) the coherent LCU for partial superpositions $\{K_{k}\ket{\psi}\}$.

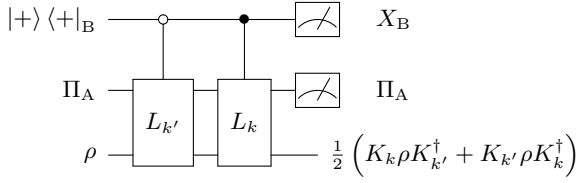
\begin{figure}[tb]
\begin{tabular}{l}
\Qcircuit @C=1em @R=1.5em {
  \lstick{\ket{+}\bra{+}_{\rm B}}     & \ctrlo{1}              & \ctrl{1}             &\meter &\rstick{X_{\rm B}}\\
  \lstick{\Pi_{\rm A}}      & \multigate{1}{L_{k'}}  & \multigate{1}{L_{k}}   &\meter &\rstick{\Pi_{\rm A}}\\
  \lstick{\rho}&  \ghost{L_{k'}}         & \ghost{L_{k}}       &\rstick{\frac{1}2\left(K_k\rho K_{k'}^\dagger +K_{k'}\rho K_{k}^\dagger\right)}  \qw}
  ~~~~~~~~~~~~~~~~
\\
\\
\end{tabular}
\caption{Quantum circuit in our decomposition of $\tilde{\Lambda}_{p,\mathcal{U}}$.
\purple{Similar to the virtual LCU, the final effective state in the third line is simulated by taking the expectation of $X_{\rm B}\otimes \Pi_{\rm A}$; see Eq.~\eqref{eq:main_equality}.}
The unitary $L_k$ is defined as Eq.~\eqref{eq:LCUop} for the LCU operator \purple{$K_k=\sum_{i\in S_k}(p_i/q_k)U_i$}.
\purple{Although $\Pi_{\rm A}$ denotes the initial state of the $\lceil\log_2 m\rceil$-qubit system A, $L_k$ acts only on a subset of the qubits in A.}
Note that if $k=k'$, then we can omit the 1-qubit control for $L_k$.}
\label{fig:components_decomp_LCU}
\end{figure}

While the above procedure focuses on the superposition of the three states, it can be generalized to the simulation of superposition among multiple (unnormalized) states.
This can be summarized as follows.
\begin{thm}[Decomposition of LCU-based CP maps]\label{thm:1}
    Suppose we have a probability distribution with $m$ outcomes $\{p_i\}$ and a set of $m$ unitary gates $\{U_i\}$.
    We split the index set $[m]:=\{1,2,\cdots,m\}$ into $G$ (nonempty) subsets $\{S_k\}$ satisfying
    \begin{equation}\label{eq:condition4partition}
        [m]=\bigcup_{k=1}^G S_k,~S_k\cap S_l=\emptyset,~(k\neq l).
    \end{equation} 
    For any partition $\{S_k\}$ of $[m]$, 
    there exists a mixed unitary channel $\Gamma$ \purple{that randomly realizes a quantum circuit in the form of Fig.~\ref{fig:components_decomp_LCU} (without measurements) and satisfies}
    \begin{align}\label{eq:main_equality}
        &\tilde{\Lambda}_{p,\mathcal{U}}(\bullet):=\left(\sum_{i=1}^m p_i U_i\right)\bullet\left(\sum_{i=1}^m p_i U_i\right)^\dagger\notag\\[6pt]
        &={\rm tr}_{{\rm AB}}\left[\left( \Pi_{\rm A}\otimes X_{\rm B}\otimes \bm{1}\right) \Gamma[\Pi_{\rm A}\otimes \ket{+}\bra{+}_{\rm B}\otimes (\bullet)]\right].
    \end{align}
\end{thm}

\noindent
The mixed unitary channel $\Gamma$ on A, B, and the target system S can be explicitly constructed, as follows; for a given partition $\{S_k\}$, 
\begin{equation}
        \Gamma[\star]:=\sum_{k,k'=1}^G{q_kq_{k'}}L^{(\rm c)}_{kk'}\left( \star\right) (L^{\rm (c)}_{kk'})^\dagger,
    \end{equation}
\purple{where $L_{kk'}^{(\rm c)}$ is the unitary operator in Fig.~\ref{fig:components_decomp_LCU} defined by
\begin{equation}\label{eq:Lkk'_def}
    L^{(\rm c)}_{kk'} := \ket{0}\bra{0}_{\rm B}\otimes L_{k'} +\ket{1}\bra{1}_{\rm B}\otimes L_{k}.
\end{equation}
The unitary $L_k$ is defined as Eq.~\eqref{eq:LCUop} for $K_k$.
The definition of $q_k$ and $K_k$ follows from Eq.~\eqref{eq:weight_qk_def_Kk}.}
\purple{We remark that the $L_k$ can be taken as any circuit such that
$\Pi_{\rm A}L_k\Pi_{\rm A}=\Pi_{\rm A}\otimes K_k$
holds, that is, any (error-free) block-encoding~\cite{Low2019hamiltonian,gilyen2019quantum} of $K_k$.
In the resource analysis below, we focus on the circuit structure of Eq.~\eqref{eq:LCUop}; finding a more efficient block-encoding of $K_k$ is an important direction for future work.}
We provide a sketch of the proof for Theorem~\ref{thm:1}; see Appendix~\ref{apdx:prf_thm1} for the full version of the proof.
\begin{proof}[Sketch of the proof of Theorem~\ref{thm:1}]
    For any input state $\rho$, which can also be taken as any operator on S in the following, we can show that
    \begin{align}\label{eq:projected_state_thm1}
        &\Pi_{\rm A}
        \Gamma[\Pi_{\rm A}\otimes \ket{+}\bra{+}_{\rm B}\otimes \rho]
        \Pi_{\rm A}\notag\\
        &=\Pi_{\rm A}\otimes \left[\frac{I_{\rm B}}{2}\otimes \left(\sum_{k=1}^{G} q_k K_k \rho K_k^\dagger\right)+\frac{X_{\rm B}}{2}\otimes \tilde{\Lambda}_{p,\mathcal{U}}(\rho)\right],
    \end{align}
    where we write $\Pi_{\rm A}\otimes \bm{1}_{\rm BS}$ as $\Pi_{\rm A}$ for simplicity.
    Thus, using the relation Eq.~\eqref{eq:projected_state_thm1}, we directly prove that
    \begin{align*}
        &\tilde{\Lambda}_{p,\mathcal{U}}(\rho)\notag\\
        &~~~={\rm tr}_{{\rm AB}}\left[\left( \Pi_{\rm A}\otimes X_{\rm B}\otimes \bm{1}\right) \cdot \Gamma[\Pi_{\rm A}\otimes \ket{+}\bra{+}_{\rm B}\otimes \rho]\right].
    \end{align*}    
\end{proof}

We here clarify the quantum resources required for simulating  our decomposition of $\tilde{\Lambda}_{p,\mathcal{U}}$.
For some partition $\{S_k\}$, the number of qubits in A can be reduced.
From the construction of $L_k$ in Eq.~\eqref{eq:LCUop}, the number of the qubits required for the PREPARE operation is at most $\lceil \log_2 |S_k|\rceil$ for the number $|S_k|$ of elements in $S_k$.
Thus, without loss of generality, we can replace the system A in Eq.~\eqref{eq:main_equality} with its subsystem ${\rm A}'$ that is sufficiently large for the PREPARE operation in $L_k$ of the largest subset among $\{S_k\}$.
As a result, quantum circuits for the decomposition in Theorem~\ref{thm:1} require
\begin{equation}\label{eq:additional_qubits_main}
    \mathcal{O}\left(\max_{k} \log |S_k|\right)~\mbox{ancilla~qubits}.
\end{equation}
Accordingly, we replace the initial state $\Pi_{\rm A}$ in Eq.~\eqref{eq:main_equality} with the $(\max_k\lceil \log_2 |S_k|\rceil)$-qubit initial state.
As for the gates in each circuit, it comes from the use of the coherent LCU methods for $K_k,K_{k'}$, leading to at most
\begin{equation}\label{eq:additional_gates_main}
    \mathcal{O}\left(\max_{k\neq l}\log\left( |S_k|^{|S_k|}|S_{l}|^{|S_{l}|}\right)\right)~\mbox{gates},
\end{equation}
\purple{except for controlled-$U_i$ gates.}
Here, $\max_{k\neq l}$ comes from the fact that when $k=l$, $L_{kl}^{(\rm c)}$ in $\Gamma$ is equivalent to $L_{k}$ (which acts on B trivially).


Theorem~\ref{thm:1} contains the two LCU methods described in Section~\ref{sec:pset_prevLCU}.
The coherent implementation of the LCU method in Section~\ref{sec:coherentlcu} corresponds to the choice of 
\begin{equation}\label{eq:no_partition}
    S_1:=[m]=\{1,2,...,m\}    
\end{equation} 
in Theorem~\ref{thm:1}.
In this partition, the mixed unitary channel $\Gamma$ becomes the unitary channel $L_{11}^{(\rm c)}(\star)L_{11}^{(\rm{c}  ) \dagger}$, and this has no interaction between AS and B due to the same index.
On the other hand, the fully stochastic LCU method in Section~\ref{sec:virtuallcu} corresponds to the choice of 
\begin{equation}\label{eq:complete_fragment}
    S_k:=\{k\}~~\mbox{for~all}~~k=1,2,...,m.
\end{equation}
Then, the quantum gate $L_k$ acts non-trivially on only S, thereby decoupling the mixed unitary channel $\Gamma$ completely from the system A.
A more detailed explanation is provided in Appendix~\ref{apdx:prf_thm1}.
While these two methods completely separate one system A or B from the other systems, our decomposition (for typical partitions) employs a 3-body interaction on ABS.
This 3-body interaction by $\Gamma$ enhances sample efficiency in the expectation value estimation by leveraging available quantum resources, as discussed in the next subsection.


\subsubsection{Expectation value estimation}\label{sec:main_exp_val_est}

The decomposition of Eq.~\eqref{eq:main_equality} naturally sets an estimator for ${\rm tr}[O\tilde{\Lambda}_{p,\mathcal{U}}(\rho)]$ as follows.
Note that the following procedure is valid for any partition $\{S_k\}_{k=1}^G$ satisfying the condition Eq.~\eqref{eq:condition4partition} of Theorem~\ref{thm:1}.

\begin{itemize}
    \item[(i)] Sample $(k,k')$ from the distribution $\{q_kq_{k'}\}$.

    \item[(ii)] Run the circuit in Fig.~\ref{fig:components_decomp_LCU} with the measurement of $O$ at the end of the bottom wire. The corresponding POVM is $\{\Pi_{\rm A}^{(z)} \otimes \ket{b}\bra{b}_{\rm B}\otimes \ket{j}\bra{j}_{\rm S}\}$,
    where $\ket{b}$ ($b\in \{0,1\}$) is the eigenstate of $X$ with eigenvalue $(-1)^b$. For $z\in\{0,1\}$, we define $\Pi_{\rm A}^{(0)}:=\Pi_{\rm A}$ and $\Pi_{\rm A}^{(1)}:=\bm{1}-\Pi_{\rm A}$.
    Also, $\ket{j}\bra{j}$ denotes the eigenstate of $O$ with eigenvalue $o_j$.
    
    
    \item[(iii)] Calculate $g_{O}:=(-1)^b \chi_{0}(z) o_j$, where $\chi_{0}(z)$ is the delta (or indicator) function for $\{z=0\}$.

\end{itemize}
It should be noted that while a part of qubits in A may be omitted as discussed in the previous subsection, we write all $\lceil \log_2 m\rceil$ qubits in A explicitly for conciseness.

First, we focus on the mean of the random variable $g_O$.
From the following observation
\begin{align}
    &\Pi_{\rm A}\otimes X_{\rm B}\otimes O_{\rm S}\notag\\
    &~~~=\sum_{z,b,j} (-1)^{b}\chi_0(z)o_j \cdot \Pi_{\rm A}^{(z)} \otimes \ket{b}\bra{b}_{\rm B}\otimes \ket{j}\bra{j}_{\rm S}
\end{align}
and the decomposition Eq.~\eqref{eq:main_equality}, we obtain
\begin{equation}
    \mathbb{E}\left[g_O\right] = {\rm tr}[O\tilde{\Lambda}_{p,\mathcal{U}}[\rho]].
\end{equation}
Importantly, the first moment does not depend on the choice of the partition $\{S_k\}$, but the above procedure (i)--(iii) itself depends on the partition.

Next, we evaluate the second moment $\mathbb{E}\left[g_O^2\right]$.
Using the expression of Eq.~\eqref{eq:projected_state_thm1}, we can directly obtain
\begin{align}\label{eq:second_mom_upb}
    &\mathbb{E}\left[g_O^2\right]={\rm tr}[(\Pi_{\rm A}\otimes I_{\rm B}\otimes O_{\rm S}^2) \Gamma[\Pi_{\rm A}\otimes \ket{+}\bra{+}_{\rm B}\otimes \rho]]\notag\\[6pt]
    &=
    {\rm tr}\left[O^2_{\rm S}\left\{\frac{I_{\rm B}}{2}\otimes \left(\sum_{k=1}^{G} q_k K_k \rho K_k^\dagger\right)+\frac{X_{\rm B}}{2}\otimes \tilde{\Lambda}_{p,\mathcal{U}}(\rho)\right\}\right]\notag\\[6pt]
    &\leq R\left[\{S_k\};\rho\right]\|O\|^2, 
\end{align}
where $R$, depends on the partition $\{S_k\}$ and the input state $\rho$, is defined as
\begin{equation}
\label{eq:def_R_factor}
    R\left[\{S_k\};\rho\right]:=\sum_{k=1}^{G} q_k {\rm tr}\left[ K_k^\dagger K_k \rho\right].
\end{equation}
We call $R$ the {\it reduction factor}.
For typical observables satisfying $O^2=\bm{1}$ (e.g., Pauli operators), the final equality in Eq.~\eqref{eq:second_mom_upb} holds.
Using Eq.~\eqref{eq:second_mom_upb}, we evaluate the variance of the estimator ${g}_O$ as
\begin{equation}\label{eq:var_upb_general}
    {\rm Var}[{g}_O]\leq {R\left[\{S_k\};\rho\right]\|O\|^2-{\rm tr}[O\tilde{\Lambda}_{p,\mathcal{U}}[\rho]]^2},
\end{equation}
where the equality holds when $O^2=\bm{1}$.
Thus, the variance depends on the partition $\{S_k\}$ in contrast to the first moment $\mathbb{E}\left[{g}_O\right]={\rm tr}[O\tilde{\Lambda}_{p,\mathcal{U}}[\rho]]$.

\purple{By using the $g_{O}$ and $g_{\bm{1}}$ (obtained by setting $O=\bm{1}$), we can estimate Eqs.~\eqref{eq:target_ratio_est} and~\eqref{eq:target_nume_est}.
Specifically, for the averaged value $\bar{g}_{{O},N}$ for identical $N$ samples of $\hat{g}_{O}$, we rescale $\bar{g}_{{O},N}$ for Eq.~\eqref{eq:target_nume_est} and take the ratio of $\bar{g}_{{O},N}$ and $\bar{g}_{{\bm{1}},N}$ for Eq.~\eqref{eq:target_ratio_est}.
The performance of this estimation is summarized as follows.}
%
\begin{thm}\label{thm:2_samplecomplexity}
    
    Let $O,\rho$ be an observable and initial quantum state, and let $K\bullet K^\dagger=\|c\|_1^2\tilde{\Lambda}_{p,\mathcal{U}}(\bullet)$ as in Eq.~\eqref{eq:KandctildeLam}.
    Using a decomposition of $\tilde{\Lambda}_{p,\mathcal{U}}$ in Theorem~\ref{thm:1} based on a partition $\{S_k\}$, we can obtain 
    the mean $\overline{g}_{O,N}$ of $N$ independent samples $g_{O}$ via (i)--(iii).
    For an additive error $\varepsilon~(\ll 1)$ and a failure probability $\delta>0$, 
    this estimator $\overline{g}_{O,N}$ satisfies 
    \begin{equation}
    {\rm Pr}\left(\left|\|c\|^2_1\overline{g}_{O,N}-{\rm tr}[OK\rho K^\dagger]\right|\leq \varepsilon\right)\geq 1-\delta
    \end{equation}
    for the number of samples
    \begin{equation}\label{eq:main_sample_numeest}
    N=\mathcal{O}\left(\|c\|^4_1 \cdot R[\{S_k\};\rho]\times \frac{\|O\|^2\log(1/\delta)}{\varepsilon^2}\right).
    \end{equation}
    Also,
    \begin{equation}
    {\rm Pr}\left(\left|\frac{\overline{g}_{O,N}}{\overline{g}_{\bm{1},N}}-\frac{{\rm tr}[OK\rho K^\dagger]}{{\rm tr}[K\rho K^\dagger]}\right|\leq \varepsilon\right)\geq 1-\delta
    \end{equation}
    holds for the number of samples
    \begin{equation}\label{eq:main_sample_ratioest}
    N=\mathcal{O}\left(\frac{R[\{S_k\};\rho]}{\mathcal{P}^2}\times \frac{\|O\|^2\log(1/\delta)}{\varepsilon^2}\right).
    \end{equation}
\end{thm}
\noindent
The number of samples Eqs.~\eqref{eq:main_sample_numeest} and~\eqref{eq:main_sample_ratioest} are proved in Appendix~\ref{apdx:sample_complexity}.
We remark that while several components in Eqs.~\eqref{eq:main_sample_numeest} and~\eqref{eq:main_sample_ratioest} (e.g., $R$ or $\mathcal{P}={\rm tr}[\tilde{\Lambda}_{p,\mathcal{U}}(\rho)]$) may be unknown in advance, we can overcome this by estimating unknown factors together with the target values when the number of samples $N$ is sufficiently large; see Appendix~\ref{apdx:sample_complexity} for the asymptotic analysis of the estimation schemes.

From Theorem~\ref{thm:2_samplecomplexity}, the partition $\{S_k\}$ for the decomposition of $\tilde{\Lambda}_{p,\mathcal{U}}$ varies the required number of samples to obtain the target values Eqs.~\eqref{eq:target_ratio_est} and~\eqref{eq:target_nume_est} via the change of the reduction factor $R$.
In the next section, we show that there is a clear relationship between the value of $R$ and the size of $S_k$, which determines the quantum resources for the procedure (ii); see Eqs.~\eqref{eq:additional_qubits_main} and~\eqref{eq:additional_gates_main}.


\subsubsection{Multi-round case}\label{sec:multi-round_case}
Although the above discussion focuses on the case of a \textit{single} LCU-based CP map $\tilde{\Lambda}_{p,\mathcal{U}}$, this can be generalized to multi-round LCU-based CP maps.
Let $\{\tilde{\Lambda}^{(\mu)}\}_{\mu=1}^r$ be a sequence of LCU-based CP maps in Theorem~\ref{thm:1}, and we fix a partition for each $\tilde{\Lambda}^{(\mu)}$ such that 
\begin{equation}
    \tilde{\Lambda}^{(\mu)}(\bullet)=\left(\sum_{k}q_{k}^{(\mu)}K_{k}^{(\mu)}\right)\bullet \left(\sum_{k}q_{k}^{(\mu)}K_{k}^{(\mu)}\right)^\dagger
\end{equation}
holds.
Then, we can follow the same procedure (i)---(iii) and obtain an estimate $g'_{O}$ whose expectation matches 
\begin{equation}
    \mathbb{E}[g'_O]={\rm tr}\left[O\left(\tilde{\Lambda}^{(r)}\cdots\circ\tilde{\Lambda}^{(2)}\circ\tilde{\Lambda}^{(1)}\right)(\rho)\right].
\end{equation}
Using this, we can estimate various properties of the following quantum state
\begin{equation}
    \frac{\left(\tilde{\Lambda}^{(r)}\cdots\circ\tilde{\Lambda}^{(2)}\circ\tilde{\Lambda}^{(1)}\right)(\rho)}{{\rm tr}\left[\left(\tilde{\Lambda}^{(r)}\cdots\circ\tilde{\Lambda}^{(2)}\circ\tilde{\Lambda}^{(1)}\right)(\rho)\right]}.
\end{equation}
In this application, the sample complexity similar to Theorem~\ref{thm:2_samplecomplexity} can be derived for the generalized $R$:
\begin{equation}\label{eq:multi_round_R}
    \prod_{\mu=1}^{r}\left(\sum_{k}q^{(\mu)}_{k}{\rm tr}\left[K_{k}^{(\mu)\dagger}K_{k}^{(\mu)}\rho_{\mu}\right]\right),
\end{equation}
where the quantum state $\rho_{\mu}$ is given by 
\begin{equation}
    \rho_{\mu+1}=\frac{\sum_{k}q^{(\mu)}_{k}K^{(\mu)}_{k}\rho_{\mu}K_{k}^{(\mu)\dagger}}{{\rm tr}\left[\sum_{k} q^{(\mu)}_{k}K_{k}^{(\mu)}\rho_{\mu}K_{k}^{(\mu)\dagger}\right]}
\end{equation}
and $\rho_1=\rho$.
From Eq.~\eqref{eq:multi_round_R}, we observe that the generalized $R$ is the product of the original $R$ Eq.~\eqref{eq:def_R_factor} for the input state $\rho_{\mu}$.

\subsection{Property of the reduction factor $R$}
\subsubsection{Monotonicity under additional quantum resources}\label{sec:tradeoff}


\begin{figure}[tb]
 \centering
 \includegraphics[scale=0.55]{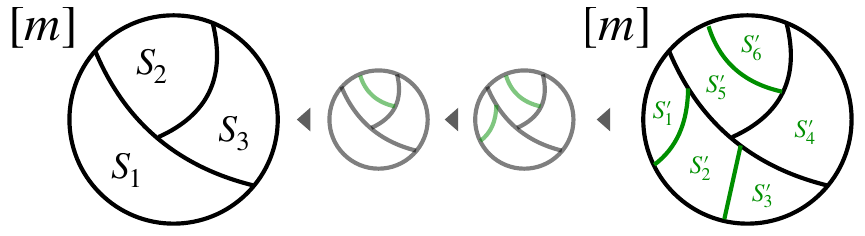}
 \caption{\purple{A series of partitions of $[m]$ satisfying the assumption of Theorem~\ref{thm:2}.}}
 \label{fig:thm3}
\end{figure}

The larger number of elements in $S_k$, the more quantum resources are required for the implementation of $L_k$, which prepares the quantum superposition among $\{p_iU_i\ket{\psi}\}_{i\in S_k}$.
The following theorem indicates that this additional usage of quantum resources enhances the efficiency of the estimation scheme, that is, the variance of the estimator ${g}_{O}$ can be reduced. We provide its proof in Appendix~\ref{apdx:prf_thm2} \purple{and Fig.~\ref{fig:thm3} shows an illustrative example.}
\begin{thm}\label{thm:2}
    Let us consider the decomposition methods for $\tilde{\Lambda}_{p,\mathcal{U}}(\bullet)$ in Theorem~\ref{thm:1}.
    Suppose we have two partitions $\{S_k\}_{k=1}^G$ and $\{S'_{l}\}_{l=1}^{G'}$ of $[m]$ satisfying the condition Eq.~\eqref{eq:condition4partition} in Theorem~\ref{thm:1}.
    If for any $S'_l$ ($l=1,2,...,G'$), there exists a single set $S_k$ that contains $S_l'$, then the corresponding reduction factors {in Eq.~\eqref{eq:def_R_factor}} satisfy the following inequalities for any input quantum state $\rho$:
    \begin{equation}\label{eq:inequalities4R}
        {\mathcal{P}}\leq R\left[\{S_k\}_{k=1}^G;\rho\right] \leq R\left[\{S'_{l}\}_{l=1}^{G'};\rho\right]\leq 1.
    \end{equation}
    The upper and lower bound can be saturated by the two extreme cases: the virtual method Eq.~\eqref{eq:complete_fragment} and coherent method Eq.~\eqref{eq:no_partition}, respectively.
\end{thm}

Since the $R$ linearly affects the required number of samples Eqs.~\eqref{eq:main_sample_numeest} and~\eqref{eq:main_sample_ratioest}, the monotone property of Eq.~\eqref{eq:inequalities4R} allows us to systematically improve the sample complexity by leveraging additional quantum resources.
Specifically, in the case of Eq.~\eqref{eq:main_sample_ratioest}, 
our method can connect the two scalings regarding the inverse of success probability $\mathcal{P}={\rm tr}{[\tilde{\Lambda}_{p,\mathcal{U}}(\rho)]}$; $\mathcal{P}^{-2}$ (achieved by the virtual LCU method, $R=1$) and $\mathcal{P}^{-1}$ (achieved by the coherent LCU method, $R=\mathcal{P}$).
This improvement is crucial when the success probability $\mathcal{P}={\rm tr}{[\tilde{\Lambda}_{p,\mathcal{U}}(\rho)]}$ is small.

\subsubsection{More detailed inequalities}

Here, we provide some useful inequalities for the reduction factors $R$ between distinct partitions.
The proof can be found in Appendix~\ref{sec:apdxD}.
\begin{lemma}\label{lem:R_factpr_ineq_prop_main}
    Suppose that the same assumption of Theorem~\ref{thm:1} and $|S_{G}|\geq 2$ hold. Then, the followings hold.
    \begin{itemize}
        \item[\rm A.] {\bf Split subset:}
        If we split $S_{G}$ into two subsets $S_{A},S_{B}$ such that $\{S_A,S_B\}$ is a partition of $S_{G}$, then for any $\rho$
        \begin{align}
            &R\left[\{S_k\}_{k=1}^{G-1}\cup \{S_A,S_B\}; \rho\right] \notag\\
            &~~~~~\leq R\left[\{S_k\}_{k=1}^G; \rho\right]+2H\left(q_A,q_B\right)
        \end{align} 
        holds. $H(a,b):=2ab/(a+b)$ denotes the harmonic mean, and $q_k=\sum_{i\in S_k} p_i$ for $k=A,B$.
        \item[\rm B.] {\bf Fully fragment subset:}
        If we fully fragment $S_{G}$ into $\{\{i\}:i\in S_G\}$, then for any $\rho$
        \begin{align}
            &R\left[\{S_k\}_{k=1}^{G-1}\cup\{\{i\}:i\in S_G\}; \rho\right] \notag\\
            &~~~~~\leq R\left[\{S_k\}_{k=1}^G; \rho\right]+q_G,
        \end{align}
        where $q_G :=\sum_{i\in S_G} p_i$.
    \end{itemize}
\end{lemma}

These properties indicate that if the probability distribution $\{p_i\}$ in $\tilde{\Lambda}_{p,\mathcal{U}}$ is sufficiently localized in a certain index subset, the other elements can be fully fragmented while keeping the $R$ small i.e., $R\sim \mathcal{P}$.
Specifically, for a partition $\{S_A,S_B\}$ of $[m]$, we have
\begin{align}\label{eq:R_factor_ineq}
    R[\{S_A\}\cup\{\{i\}:i\in S_B\};\rho]&\leq \mathcal{P}+q_B+2H(q_A,q_B)\notag\\
    &=\mathcal{P}+\mathcal{O}(q_B)
\end{align}
from Lemma~\ref{lem:R_factpr_ineq_prop_main}.
In such cases, our protocol shows the advantages of both the coherent and virtual simulation of LCU-based CP maps.
That is, in comparison with the virtual simulation, our protocol achieves a nearly quadratic improvement in the sample complexity~\eqref{eq:main_sample_ratioest} regarding the projection probability $\mathcal{P}={\rm tr}[\tilde{\Lambda}_{p,\mathcal{U}}(\rho)]$ at the cost of a (slightly) increased circuit complexity, which is also smaller than that of the full coherent implementation.


\section{Application}
\label{sec:application}

Here, we deploy our hybrid strategies in several representative applications.
For each application, the reduction factor $R$ is tuned by an application-specific grouping of unitary operators, which allows us to appropriately trade off the sampling overhead against the quantum circuit depth {and ancilla qubits}.

\begin{figure*}[tb]
 \begin{minipage}[b]{0.495\linewidth}
    \centering
    \includegraphics[scale=0.85]{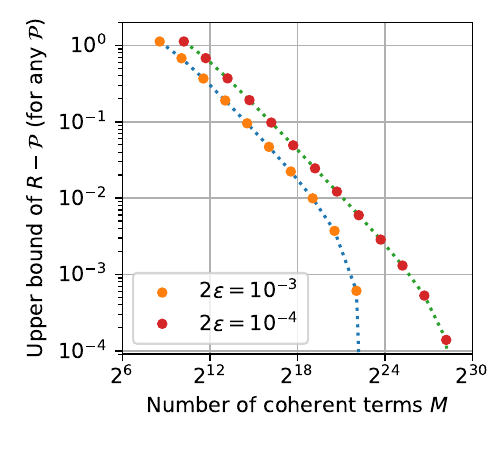}
  \end{minipage}
  \begin{minipage}[b]{0.495\linewidth}
    \centering
    \includegraphics[scale=0.85]{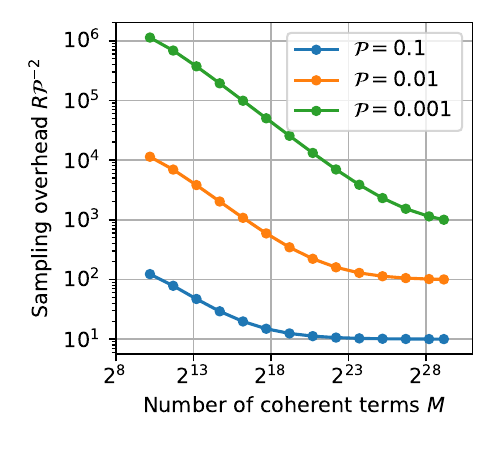}
  \end{minipage}
  \caption{The $M$-dependence of the upper bound of (left) the $R-\mathcal{P}$ for Eq.~\eqref{eq:lchs_target_lcu} and (right) the sampling overhead $R\mathcal{P}^{-2}$ in the LCHS example.
  The corresponding partition is described below Eq.~\eqref{eq:lchs_target_lcu}.
  The parameters are chosen as $\|L\|=2$ and $T=3$.
  We plot the final expression of Eq.~\eqref{eq:lchs_rp_eval} as the dotted lines in the left figure. We set $2\varepsilon=10^{-4}$ in the right figure.}
  \label{fig:lchs_res}
\end{figure*}

\subsection{Linear combination of Hamiltonian simulation}

We here consider a specific example for solving differential equations ${\rm d}u(t)/{\rm d}t = -Au(t)$.
Here, $A$ is a general time-independent matrix. 
A generalization to time-dependent cases or inhomogeneous cases of the following analysis would be possible, but we leave it for future work.
Assuming the Hermitian part $L=(A+A^\dagger)/2$ is positive semi-definite, the time propagator from $u(0) $ to $u(T)$ can be written as LCU~\cite{an2023linear}: 
\begin{equation}\label{eq:lchs_propagator}
    e^{-AT}=\int_{\mathbb{R}} \frac{1}{\pi(1+k^2)} e^{-iT(H+kL)} {\rm d}k,
\end{equation}
where $H=(A-A^\dagger)/2i$.
If the condition $L\geq 0$ is false, then we consider the equation of $v(t)=e^{-ct}u(t)$ with a positive value $c$ such that $L+c\bm{1}\geq 0$.
Then, $v(t)$ is propagated by Eq.~\eqref{eq:lchs_propagator} with $L+c\bm{1}$ instead of $L$, and we can recover $u(T)=e^{cT}v(T)$.

Since the propagator Eq.~\eqref{eq:lchs_propagator} is in the form of Eq.~\eqref{eq:lcu_before_normalize}, we can use our framework described in the previous section.
To be specific, we first truncates the integral over $\mathbb{R}$ to $[-\mathcal{K}_1,\mathcal{K}_1]$ for some $\mathcal{K}_1$ and then divide the truncated integral into two parts:
\begin{equation}\label{eq:lchs_paritioned}
    \int_{-\mathcal{K}_2}^{\mathcal{K}_2} \frac{e^{-iT(H+kL)}}{\pi(1+k^2)}{\rm d}k+\int_{\mathcal{K}_2\leq|k|\leq \mathcal{K}_1} \frac{e^{-iT(H+kL)}}{\pi(1+k^2)}{\rm d}k
\end{equation}
where $\mathcal{K}_2\in [0,\mathcal{K}_1]$ is a parameter for specifying a series of partitions later.
From the analysis of Ref.~\cite{an2023linear}, the choice of $\mathcal{K}_1=\tan(\pi(1-\varepsilon)/2)=\mathcal{O}(1/\varepsilon)$ is sufficient to ensure the truncation error between Eq.~\eqref{eq:lchs_propagator} and~\eqref{eq:lchs_paritioned} below $\varepsilon$. 
Then, we further discretize the first integral as $\sum_{j=0}^{M} s_j e^{-iT(H+k_jL)}$ via the trapezoidal rule with
\begin{equation}
    s_j=\frac{1}{\pi(1+k_j^2)}\frac{(2-\bm{1}_{j=0,M})\mathcal{K}_2}{M},~~k_j=-\mathcal{K}_2+\frac{2j\mathcal{K}_2}{M},
\end{equation}
where $\bm{1}_{j=0,M}$ is 1 for $j=0,M$ and 0 for $j\neq 0,M$.
The number $M$ is taken as $\mathcal{O}(\|L\|T\sqrt{\mathcal{K}_2^3/\varepsilon})$ to suppress the discretization error below $\varepsilon$.
Overall, the target propagator $e^{-AT}$ is well approximated by 
\begin{equation}\label{eq:lchs_target_lcu}
    \sum_{j=0}^M s_j e^{-iT(H+k_jL)}+\int_{\mathcal{K}_2\leq|k|\leq \mathcal{K}_1} \frac{e^{-iT(H+kL)}}{\pi(1+k^2)}dk.
\end{equation}

In the following, we focus on a partition that the terms in the first summation are contained in $S_1$ and all the other terms in the second integral are separably contained in each size-1 subset. 
This partition is controlled by the parameter $\mathcal{K}_2$; in particular, $\mathcal{K}_2=0$ and $\mathcal{K}_2=\mathcal{K}_1$ correspond to the full virtual and coherent LCU, respectively.
Using Eq.~\eqref{eq:R_factor_ineq}, the reduction factor $R$ for this partition can be evaluated as
\begin{align}\label{eq:lchs_rp_eval}
    R-\mathcal{P}
    &\leq \frac{2\alpha(\mathcal{K}_1,\mathcal{K}_2)}{\pi}\frac{5\|s\|_1+(2/\pi)\alpha(\mathcal{K}_1,\mathcal{K}_2)}{(\|s\|_1+(2/\pi)\alpha(\mathcal{K}_1,\mathcal{K}_2))^2}\notag\\
    &\simeq \left(1-\frac{\arctan{\mathcal{K}_2}}{\arctan \mathcal{K}_1}\right)\left(1+4\frac{\arctan{\mathcal{K}_2}}{\arctan \mathcal{K}_1}\right),
\end{align}
where $\alpha(\mathcal{K}_1,\mathcal{K}_2):=\arctan(\mathcal{K}_1)-\arctan(\mathcal{K}_2)$, and we use $\|s\|_1\simeq (2/\pi)\arctan{\mathcal{K}_2}$ for the second line.
Here, we substituted $q_{\rm A}\propto \|s\|_1$, $q_{\rm B}\propto (2/\pi)\alpha(\mathcal{K}_1,\mathcal{K}_2)$, and $q_{\rm A}+q_{\rm B}=1$ into Eq.~\eqref{eq:R_factor_ineq}.

We show the relationship between the number of the coherent terms $M$ and the upper bound of $R-\mathcal{P}$ (left) and the sampling overhead $R\mathcal{P}^{-2}$ (right) in Fig.~\ref{fig:lchs_res} by sweeping the parameter $\mathcal{K}_2$ in $[0,\mathcal{K}_1]$ for the case of $\|L\|=2$ and time $T=3$.
The presented result in the left figure holds for any $\mathcal{P}\gtrsim e^{-2\|L\|T}$, which is determined by specifying an input state and $A$.
Although we here do not specify the input (except for the norm $\|L\|$ of $L$), for instance,
we focus on the case $\mathcal{P}\sim 10^{-2}$.
In this case, $R=\mathcal{P}\sim 10^{-2}$ holds for the full coherent LCU.
From Fig.~\ref{fig:lchs_res}, when $2\varepsilon=10^{-4}$, it is sufficient to choose $M\sim 2^{21}$ in our method to bound ${R}-\mathcal{P}\lesssim 10^{-2}$, i.e., ${R}\lesssim \mathcal{P}+10^{-2}\sim 2\times 10^{-2}$, yielding
a similar number of samples~\eqref{eq:main_sample_ratioest} to the full coherent LCU which requires $\sim 2^{29}$ terms in a single LCU circuit.
As a result, in this case, we obtain $7$ ancilla qubits reduction and approximately $\times 256$ depth reduction while keeping $\mathcal{O}(1/\mathcal{P})$ sampling overhead for the ratio estimation, in comparison with the full coherent LCU method.
We note that the depth reduction may be more modest when leveraging the structure of PREPARE and SELECT for this task.


\subsection{\purple{Quantum} linear system solver}
The \purple{quantum linear system solver} is for solving a linear equation $M \ket{x} = \ket{b}$ with $M$ being a Hermitian matrix and $\ket{x}$ and $\ket{b}$ are vectors~\cite{childs2017quantum}.
\purple{Specifically, we aim at obtaining the expectation value of an observable for the normalized solution vector~\cite{wang2024qubit,Chakraborty2024implementingany}.
Now we assume that the singular values of $M$ are in $[1/\kappa,1]$ for some $\kappa>1$.
To prepare the solution vector, Ref.~\cite{childs2017quantum} uses the following decomposition of $M^{-1}$ via the Fourier transform}
\begin{equation}
M^{-1} =\frac{i}{\sqrt{2\pi}}   \int_{0}^{\infty} {\rm d}y\int_{-\infty}^{\infty} {\rm d}z~ze^{-z^2/2}e^{-iyz M},
\end{equation}
and its discretized version we employed here
\begin{equation}\label{Eq: linearsystemmain}
M^{-1} \sim \frac{i}{\sqrt{2 \pi}} \sum_{j=0}^{J-1} \Delta_y \sum_{k=-K}^K \Delta_z z_k e^{-\frac{z_k^2}{2}} e^{-i M y_jz_k},
\end{equation}
where $z_k=k\Delta_z$, $y_j=j\Delta_y$, $J= \Theta (\frac{\kappa}{\varepsilon} \mathrm{log}(\frac{\kappa}{\varepsilon}) )$, $K= \Theta ({\kappa} \mathrm{log}(\frac{\kappa}{\varepsilon}) )$, $\Delta_y= \Theta (\varepsilon/ \sqrt{\mathrm{log}(\frac{\kappa}{\varepsilon})})$, and $\Delta_z= \Theta (\kappa ^{-1}/\sqrt{\mathrm{log}(\frac{\kappa}{\varepsilon})})$ for the approximation error $\varepsilon$.
For implementing \purple{the LCU Eq.~\eqref{Eq: linearsystemmain} in the form of Eq.~\eqref{eq:lcu_before_normalize}}, we need to evaluate the normalization factor \purple{corresponding to $\|c\|_1$ in Eq.~\eqref{eq:action_of_LCU}}, which can be described as~\cite{childs2017quantum} 
\begin{equation}
\begin{aligned}
\|c\|_1 &= \frac{1}{\sqrt{2\pi}} \sum_{j=0}^{J-1} \Delta_y \sum_{k=-K}^K \Delta_z |z_k|e^{-z_k^2/2} \\
&= \Theta\left(\kappa \sqrt{\mathrm{log}(\kappa/\varepsilon)}\right).
\end{aligned}
\end{equation}

Then, we evaluate the reduction factor for the conventional LCU (Ref.~\cite{childs2017quantum} without amplitude amplification), randomized LCU~\cite{wang2024qubit,Chakraborty2024implementingany}, and {our hybrid LCU}.
For {our hybrid} LCU, we partition the target normalized LCU operator as
\begin{equation}
\begin{aligned}
&\frac{i}{\sqrt{2 \pi}} \sum_{j=0}^{J-1} \frac{\Delta_y}{ \|c\|_1} \sum_{k=-K}^K \Delta_z z_k e^{-\frac{z_k^2}{2}} e^{-i M y_jz_k}= \sum_{j=0}^{J-1} q_j \purple{K_j},
\end{aligned}
\end{equation}
where $q_j \propto \frac{\Delta_y}{\|c\|_1}$ and $K_j \propto \sum_{k=-K}^K \Delta_z z_k e^{-\frac{z_k^2}{2}} e^{-i M y_j z_k}$. 
Then, we perform our {hybrid} LCU by randomly sampling the operator $K_j$ with the probability $q_j$.
The corresponding $R$ becomes
\begin{equation}
\begin{aligned}
R_{\rm int} &= \sum_j q_j \bra{b} \purple{K_j^\dag K_j} \ket{b}= \mathcal{O}(\mathrm{log}^{-1/2}(\kappa/\varepsilon)).
\end{aligned}
\end{equation}
The derivation of $R_{\rm int}$ is provided in Appendix~\ref{apdx:QLSS}.
Note that the $R$ with the conventional LCU implementation is
\begin{equation}
R_{\rm conv}= \mathcal{P}= \mathcal{O}(\mathrm{log}^{-1} (\kappa/\varepsilon)),
\end{equation}
while the randomized LCU implementation leads to $R_{\rm rand}=1$.
Refer to the Appendix~\ref{apdx:QLSS} for detailed derivations.
Thus, it is clear that {our hybrid} LCU interpolates between the randomized and the conventional LCU with respect to the sampling cost, while the ancilla count is given by $\mathcal{O}(1)$, $\mathcal{O}(\log K)$, and $\mathcal{O}(\log(JK))$ for the randomized, hybrid, and conventional LCU, respectively; see the summary in Table~\ref{tab:QLSS_comparison}.

\begin{table}[t]
  \centering
  \renewcommand{\arraystretch}{1.5}
  \begin{tabular}{lccc}
    \hline
    & Random & Hybrid & Coherent \\
    \hline
    Reduction factor $R$ & $1$ & $\mathcal{O}(\log^{-1/2}(\frac{\kappa}{\varepsilon}))$ & $\mathcal{O}(\log^{-1}(\frac{\kappa}{\varepsilon}))$ \\
    Ancilla count    & $\mathcal{O}(1)$   & $\tilde{\mathcal{O}}(\log(\kappa))$  & $\tilde{\mathcal{O}}(\log(\frac{\kappa}\varepsilon))$  \\
    \hline
  \end{tabular}
  \renewcommand{\arraystretch}{1.0}
  \caption{Comparison across randomized LCU, hybrid LCU, and coherent LCU methods for the quantum linear system solver. 
  Here, $\kappa,\varepsilon$ are the condition number and the required accuracy, respectively.
  $\tilde{\mathcal{O}}$ hides $\log\log$ factors.}
  \label{tab:QLSS_comparison}
\end{table}

\subsection{\purple{Ground state property estimation}}

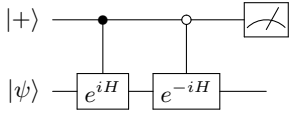
\begin{figure}[tb]
\begin{tabular}{l}
\\
\Qcircuit @C=1em @R=1.5em {
  \lstick{\ket{+}} & \ctrl{1}& \ctrlo{1}&\meter\\
  \lstick{\ket{\psi}}&\gate{e^{iH}} &\gate{e^{-iH}}&\qw\\
  }
\\
\end{tabular}
\caption{\purple{Quantum circuit for the cosine filter $\cos^T(H)$ ($T=1$).
The measurement is performed by the Pauli $X$ basis.
By repeating this circuit $T$ times, we obtain the filtered state $\cos^T(H)\ket{\psi}/\|\cos^T(H)\ket{\psi}\|$ if the measurement outcomes are all $+1$.}}
\label{fig:circuit_gprep}
\end{figure}

Here, we consider the combination of the coherent ground state preparation method~\cite{ge2019faster} using the cosine filter and the universal algorithmic cooling method~\cite{zeng2021universal}. 
Our method corresponds to the decomposition of the target LCU operator as
\begin{equation}\label{eq:cos_hybrid_filter}
K_{\rm LCU}= \sum_k q_k U_k K_{\rm cos},
\end{equation}
where $\sum q_k U_k$ for a unitary operator $U_k$ {and a probability $q_k$} is a LCU operator for virtually projecting states onto the ground state of a Hamiltonian constituted by the universal algorithmic cooling method.
Also, $K_{\rm cos}$ is the LCU operator of a cosine filter we will mention later. Note that the universal algorithmic cooling method allows for the continuous limit $\int {\rm d}q~q(U) U$ {for some continuous probability distribution $q(U)$}. Then, the reduction factor $R$ for Eq.~\eqref{eq:cos_hybrid_filter} reads 
\begin{equation}
\begin{aligned}
R &= \sum_k q_k \tr[(U_k K_{\rm cos})^\dag U_k K_{\rm cos}\rho] \\
&= \tr[K_{\rm cos}^\dag K_{\rm cos} \rho],
\end{aligned}
\end{equation}
which only depends on the LCU operator in the cosine filter method {and an initial state $\rho$}.
As detailed below, while the reduction factor can be enhanced by the cosine filter method, other parameter dependence (e.g., accuracy and spectral gap) in the gate complexity can be further improved by the universal algorithmic cooling method. 
This indicates that the combination of these two methods allows for balancing both the sampling overhead and the advantageous scaling realized by the universal algorithmic cooling method. Note that our proposed method is hardware-efficient as we only use {Hamiltonian evolution unitaries with a single-ancilla control qubit.}

We first briefly review the cosine filter method proposed in Ref.~\cite{ge2019faster}.
Let $\tilde{H}$ be a target Hamiltonian whose spectrum in $[0,1]$ by some normalization.
We assume that $\tilde{H}$ has a unique ground state with the eigenvalue $\lambda_0$, which is separated from the other eigenvalues by at least a known value $\Delta$.
Also, we employ an initial state $\ket{\psi}$ that is ensured to be $|\langle \lambda_0\ket{\psi}|^2\geq {p_0}$ for the ground state $\ket{\lambda_0}$ and a known lower bound $p_0> 0$.
The goal is to prepare $\ket{\lambda_0}$ with accuracy $\varepsilon$, when an estimate $E$ for $\lambda_0$ is given.
As well as Ref.~\cite{ge2019faster}, we assume that $\lambda_0-E\leq \mathcal{O}(\Delta/\log((p_0\varepsilon)^{-1})$.
In the following, we refer to the total evolution time of the Hamiltonian $\tilde{H}$ as the \textit{total time complexity}.

To construct the cosine filter for a positive integer $T$ and a Hermite operator $H$, we can employ the quantum circuit in Fig.~\ref{fig:circuit_gprep}.
That is, by repetitive application of controlled $e^{\pm iH}$ and post-selection of $\ket{+}$ for the single-ancilla qubit $T$ times, we can project the input state to $\mathrm{cos}^T(H)\ket{\psi}/\| \mathrm{cos}^T(H)\ket{\psi}\|$. 
The projection probability is given by $\bra{\psi} \mathrm{cos}^{2T}(H) \ket{\psi}$.
Note that this circuit can be considered as a hardware-efficient \purple{coherent LCU implementation of $K_{\rm LCU}=\cos^T(H)$} only with a repetitive use of an ancilla qubit and controlled time evolutions.
\purple{Now, we take $H:=\tilde{H}-(E-\tau)\bm{1}$, where $\tau$ is a small value specified below.
According to the analysis of Ref.~\cite{ge2019faster},
to certify}
\begin{equation}
\begin{aligned}
\left\| \frac{\mathrm{cos}^T({H})\ket{\psi}}{\|\mathrm{cos}^T({H})\ket{\psi} \|} -\ket{\lambda_0}\right\|=\mathcal{O}(\varepsilon)
\end{aligned}
\end{equation}
and $\|\mathrm{cos}^T({H})\ket{\psi} \|= \Omega(\sqrt{p_0})$, we need 
\begin{equation}
\begin{aligned}
T= \Theta \bigg(\Delta^{-2} {\mathrm{log}^2((p_0 \varepsilon)^{-1})}\bigg)
\end{aligned}
\label{Eq: M}
\end{equation}
and 
\begin{equation}
    \tau = \Theta \left( \frac{\Delta}{\mathrm{log}((p_0 \varepsilon)^{-1})}\right).
\end{equation}
Although this choice is sufficient, Ref.~\cite{ge2019faster} further improves the maximal evolution time per circuit \purple{(without amplitude amplification)} to
$\Theta ( \Delta^{-1} \mathrm{log}^{3/2}((p_0 \varepsilon)^{-1}))$ by using a suitable truncation of the Fourier expansion of $\mathrm{cos}^T({H})$.
This improved complexity is quadratically better than Eq.~\eqref{Eq: M} with respect to $\Delta$, but this improved method requires the costly PREPARE circuit with additional multiple ancilla qubits.

\purple{Our hybrid method approaches the scaling of $\Delta^{-1}$ in the limit $\varepsilon\to 0$, while maintaining the hardware-efficient implementation and the $p_0^{-1}$ scaling of the coherent LCU.}
Now, we proceed to the explanation of our hybrid method {with Eq.~\eqref{eq:cos_hybrid_filter}}. 
We first apply the cosine filter method {with $K_{\rm cos}=\cos^{T'}(H)$} by choosing $T'$ such that
\begin{equation}
\begin{aligned}
\left\| \frac{\mathrm{cos}^{T'}({H})\ket{\psi}}{\|\mathrm{cos}^{T'}({H})\ket{\psi} \|} -\ket{\lambda_0}\right\|=\mathcal{O}(p_0)
\end{aligned}
\end{equation}
and $\|\mathrm{cos}^{T'}({H})\ket{\psi} \|= \Omega(\sqrt{p_0})$ hold.
From Eq.~\eqref{Eq: M}, we can take $T'= \Theta(\Delta^{-2} \mathrm{log}^2(p_0^{-1}) )$.
In addition, we admit a more rough estimate of the ground energy with accuracy $\mathcal{O}(\Delta/\mathrm{log}(p_0^{-1}))$.
Next, we apply the universal algorithmic cooling method~\cite{zeng2021universal} for further improving the ground state preparation accuracy. 
We denote the projected state as $\ket{\psi'}= \mathrm{cos}^{T'}({H})\ket{\psi}/\|\mathrm{cos}^{T'}({H})\ket{\psi} \|$ in the previous step. \purple{Note that $1-|\langle \lambda_0\ket{\psi'}|=\mathcal{O}(p_0^2)$ holds.} 
We then require
\begin{equation}
\begin{aligned}
\bigg\| \frac{e^{-\frac{1}{2} \sigma^2 ({H}')^2}\ket{\psi'}}{\|e^{-\frac{1}{2} \sigma^2 ({H}')^2}\ket{\psi'} \|} -\ket{\lambda_0} \bigg\| =\mathcal{O}(\varepsilon)
\end{aligned}
\end{equation}
and \purple{$\| e^{-\frac{1}{2} \sigma^2 ({H}')^2}\ket{\psi'} \| =\Omega(|\langle \lambda_0\ket{\psi'}|)$},
where ${H}'= \tilde{H}- (E'-\tau')\bm{1}$ with $\tau'$ being a small real constant and $\sigma^2$ is the variance of the Gaussian filter function in the universal algorithmic cooling method. 
From Ref.~\cite{zeng2021universal}, we need
\begin{equation}
\begin{aligned}
\tau'&= \mathcal{O} \bigg( \frac{\Delta}{\sqrt{\mathrm{log}({p_0}/{\varepsilon})}} \bigg), \\
\delta'' &= \mathcal{O}\bigg( \frac{\Delta}{\sqrt{\mathrm{log}({p_0}/{\varepsilon})}} \bigg)
\end{aligned}
\end{equation}
and
\begin{equation}
\sigma^2 = \mathcal{O}\big(\Delta^{-2}\mathrm{log}({p_0/\varepsilon}) \big),
\end{equation}
where $\delta''$ is the required accuracy for the ground energy estimation $E'$. Notice that the required maximum time evolution is $\mathcal{O}(\sigma\sqrt{\log(1/\varepsilon)})$, \purple{where the log factor $\sqrt{\log(1/\varepsilon)}$ comes from the truncation of the Fourier transformation for the Gaussian filter~\cite{zeng2021universal}.
From the discussion of Section~\ref{sec:multi-round_case}, the $R$ becomes $R=\bra{\psi} \mathrm{cos}^{2T'}(H) \ket{\psi}$ due to the coherent projection effect in the first cosine filter.}

\purple{The total time complexity for the expectation value estimation with a unit-norm observable reads
\begin{align}
&\mathcal{O}\bigg(p_0^{-1} \varepsilon^{-2}\left(\Delta^{-2}\mathrm{log}^2(p_0^{-1}) +  \Delta^{-1}\sqrt{\mathrm{log}(1/\varepsilon) \mathrm{log}(p_0/\varepsilon)}\right)\bigg)
\end{align}
due to the $R$ evaluation and Eq.~\eqref{eq:main_sample_ratioest} with fixed $\delta$.
Now, denoting $\varepsilon= p_0^{\alpha}$, we obtain the total time complexity as
\begin{equation}
\mathcal{O}\left(p_0^{-1} \Delta^{-1}\varepsilon^{-2} \mathrm{log}(1/\varepsilon) \left(\frac{\Delta^{-1} \mathrm{log} (1/\varepsilon)}{\alpha^2}+ \sqrt{1-\alpha^{-1}}\right) \right)
\end{equation}
Therefore, for $\varepsilon \ll p_0$, i.e., $\alpha \rightarrow \infty$, 
the total complexity asymptotically approaches 
\begin{equation}\label{eq:asym_complexity}
\mathcal{O}\bigg(p_0^{-1} \Delta^{-1}\varepsilon^{-2}{\mathrm{log}(1/\varepsilon)} \bigg),
\end{equation}
which inherits both advantageous scalings of the physical and virtual ground-state preparation methods.
In addition, the required accuracy for the energy estimation is relaxed to $\mathcal{O}(\Delta/\sqrt{\mathrm{log}(p_0/\varepsilon)})$ from $\mathcal{O}(\Delta/\mathrm{log}((p_0\varepsilon)^{-1}))$.
}
\purple{We note that the scaling Eq.~\eqref{eq:asym_complexity} is essentially achieved by the method in Ref.~\cite{dong2022ground}, while it additionally needs to perform a modulation of the ancilla state during the controlled time evolution.}

\subsection{Quantum error detection}

Here, we discuss the application of our protocol to quantum error detection~\cite{knill2000theory}. Recently, the virtual quantum error detection protocol has been proposed, which virtually projects a quantum state onto the code subspace in a similar vein to the random LCU algorithm~\cite{tsubouchi2023virtual,mcclean2020decoding,cai2021quantum}. Suppose that we aim to apply the code projector $P={|\mathbb{S}|^{-1}} \sum_{S \in \mathbb{S}} S$, where $\mathbb{S}$ is the set of stabilizers and $S \in \mathbb{S}$. By substituting the unitary operators constituting the target operator $K_{\rm LCU}=\sum_i p_i U_i$ with the stabilizer operators for the projector, we can similarly perform the projection onto the code space, allowing for computing the expectation values in a hardware-efficient manner with a constant depth, regardless of the code structures. We can also apply this method to the symmetries in bosonic codes, e.g., rotation- and translation-symmetric bosonic codes~\cite{endo2025quantum,endo2024projective,anai2024unitary}.

Now, we introduce the intermediate implementation of quantum error detection realized by our protocol. Suppose the target code projector can be written as $P_{C}= P_{C}^{(1)} P_{C}^{(2)}$ by using the projectors $P_{C}^{(1)}$ and $P_{C}^{(2)}$. For example, the projector onto the Steane code $P_{\rm ST}$ can be expanded by $P_{\rm ST}=P_{Z} P_{X}$~\cite{steane1996simple}. Here, $P_{Z}$ is expanded by Z Pauli stabilizers, with $P_{X}$ being expanded by $X$ Pauli stabilizers. Then, we expand $P_{C}= {|\mathbb{S}_1 |^{-1}}\sum_{S \in \mathbb{S}_1} S P_{C}^{(2)}$, where $\mathbb{S}_1$ is the set of stabilizers constituting the projector $P_C^{(1)}$. We then perform our protocol by setting $q_S =|\mathbb{S}_1|^{-1}$ and $K_S = S P_C^{(2)}$ for the intermediate decomposition $K_{\rm LCU}=\sum_S q_S K_S$. Note that this implementation corresponds to the case where we perform the conventional quantum error detection for the projector $P_C^{(2)}$ and subsequently perform the virtual quantum error detection for $P_C^{(1)}$. Then, the reduction factor reads
\begin{equation}
\begin{aligned}
R_{\rm QED}&= \frac{1}{|\mathbb{S}_1|}\sum_{S \in \mathbb{S}_1} \mathrm{Tr}[(S P_C^{(2)})^\dag S P_C^{(2)} \rho] \\
&= \mathrm{Tr}[P_C^{(2)} \rho]
\end{aligned}
\end{equation}
for the input state $\rho$.

This result indicates that the reduction factor is only determined by the projector $P_C^{(2)}$. Therefore, when errors are likely to be detected $P_C^{(2)}$, by applying the conventional quantum error detection for $P_C^{(2)}$, we can minimize the sampling overhead. On the other hand, we can ease the hardware overhead by performing virtual quantum error detection for $P_C^{(1)}$. For example, when we use the Steane code and the dominant errors are Pauli Z-type errors, we can perform the conventional error detection for $P_X$ and virtual quantum error detection for $P_Z$, balancing the sampling and hardware overheads. 
Note that such a situation is anticipated when the biased cat qubits described as $\ket{\pm}_C \propto \ket{\alpha} \pm \ket{-\alpha}$ are used as a physical qubit, where $\ket{\alpha}$ is a coherent state with $\alpha >0$. As the amplitude $\alpha$ increases, the phase flip error rate linearly increases with $\alpha^2$ while the effect of bit-flip errors is exponentially suppressed with $\alpha^2$~\cite{leghtas2015confining,touzard2018coherent}.

We numerically simulate the reduction factor $R$ for our protocol and the conventional quantum error detection by using the example of the Steane code~\cite{steane1996simple}. For a randomly generated state in the code space $\ket{\psi}_C$, we apply the noise described by $\mathcal{E}_X^{\otimes 7} \circ \mathcal{E}_Z^{\otimes 7}$ with $\mathcal{E}_X(\rho)= (1-p_X)\rho + p_X X \rho X$ and $\mathcal{E}_Z(\rho)= (1-p_Z)\rho +p_Z Z \rho Z$ for the bit flip error rate $p_X$ and phase flip error rate $p_Z$. For describing the biased error model, we set $p_X = r p_Z$ for $r<1$. We then compute the projection probability $\mathcal{P}=\tr[\mathcal{E}_X^{\otimes 7} \circ \mathcal{E}_Z^{\otimes 7} (\ket{\psi_C}\bra{\psi_C}) P_C]$ when we apply the conventional quantum error detection and our hybrid strategy by setting $P_C^{(2)}=P_X$ and $P_C^{(1)}=P_Z$. Then, we plot the reduction factor $R$ for our hybrid strategy and the projection probabilities depending on the error rate $p_X$ for $r=0.1, 0.2, 0.3$ in Fig. \ref{fig:cqed}. We find that $R-\mathcal{P}$ decreases with smaller $r$ in the simulated parameter regime and is much smaller than the projection probabilities, especially at small error rates, due to bias in the noise model. Note that the reduction factor is equivalent for all the values of $r$ because the reduction factor $R$ is only determined by the phase flip errors and the projector $P^{(2)}_{C}$. 

\begin{figure}[tb]
 \centering
 \includegraphics[width=\columnwidth]{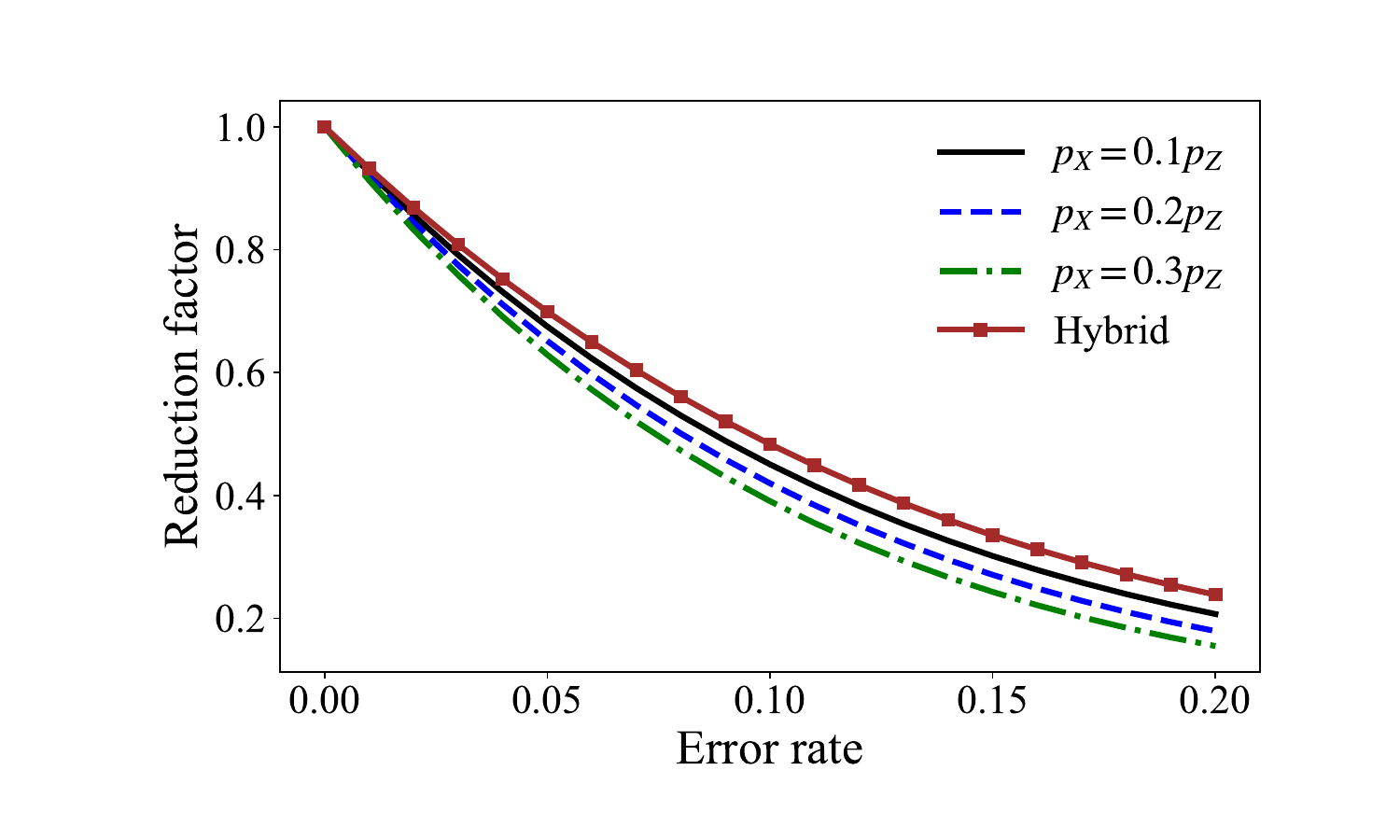}
 \caption{The projection probabilities for $r=0.1$ (black lined), $r=0.2$ (blue dashed), and $r=0.3$ (green dash-dotted), and reduction factor $R$ in our hybrid scheme (brown square) for quantum error detection in accordance with the phase-flip probability $p_Z$ for each physical qubit.}
 \label{fig:cqed}
\end{figure}

\section{Conclusion and discussion}

In this work, we clarify the trade-off relation between the circuit complexity and the number of measurements in quantum algorithms based on CP maps with LCU operations.
{This finding comes from the development and thorough analysis of a new quantum algorithm that provides an efficient building block for these high-level quantum algorithms.}
Specifically, the proposed quantum algorithm smoothly connects the coherent and virtual simulation of a given LCU operation. 
In addition, we prove that it has the monotonic reduction in the number of measurements by increasing the number of unitaries that are coherently implemented in a single circuit.

Toward the early FTQC era, our general algorithm provides a systematic way to manage the trade-off between the circuit complexity and the number of measurements.
We confirm this by providing the theoretical analysis for our algorithm's performance in the application to simulation of non-Hermitian processes, quantum linear system solver, ground state property estimation, and quantum error detection.
From the aspect of sampling overhead, the coherent LCU method corresponds to the extreme case of our algorithm with the minimal overhead.
However, our important finding is that, by making modest use of classical randomization, we can reduce the circuit complexity only with a small sampling overhead comparable to the extreme coherent case.

Here, we discuss some potential future directions. The {query} complexity of quantum algorithms has a clear hierarchy depending on the allowed quantum resources {per circuit}; we can improve the scaling of {query} complexity {from} $1/\mathcal{P}^2$ (virtual LCU) to $1/\mathcal{P}$ (coherent LCU). 
Moreover, if we are allowed to have further quantum resources, we can realize the \purple{query} complexity $1/\sqrt{\mathcal{P}}$ by combining the coherent LCU and the amplitude amplification techniques~\cite{brassard2000quantum}. 
{As such,} while we showed that we can naturally connect the first two methods in our work, the connection to the amplitude amplification is worth exploring.

Another promising direction is an information-theoretic analysis of our protocol.
Although the acquisition of the ancilla-qubit measurement outcome clearly reduces the sample complexity, elucidating the information-theoretic quantities—for instance, entropic~\cite{landi2021irreversible} or correlation-based measures~\cite{peres1996separability}—that capture this reduction may pave the way for a more systematic optimization of both quantum and classical resources in early-fault-tolerant quantum algorithms.

Finally, the quantum singular value transformation (QSVT) or its variants~\cite{Low2019hamiltonian,gilyen2019quantum,dong2022ground} are alternative methods to construct a filter. 
The combination of such a filter design and our framework may provide a more resource-efficient implementation of filters. 
For example, if the constructed filter in QSVT is not perfect, our method may be used to compensate for the algorithmic error derived from incomplete filters by partially using the virtual simulation of the recovery map.
\purple{Indeed, several recent investigations~\cite{zeng2022simple,kato2024exponentially,chakraborty2025quantum,wang2025randomized} highlight the importance of devising such algorithmic error-recovery operations.}

\section*{Acknowledgments}

K.W. was supported by JSPS KAKENHI Grant Number JP24KJ1963. This work was supported by JST [Moonshot R\&D] Grant No.~JPMJMS2061; MEXT Q-LEAP, Grant No.~JPMXS0120319794 and No.~JPMXS0118067285; JST CREST Grant No.~JPMJCR23I4 and No.~JPMJCR25I4.

\bibliography{bib}

\begin{thebibliography}{61}%
\makeatletter
\providecommand \@ifxundefined [1]{%
 \@ifx{#1\undefined}
}%
\providecommand \@ifnum [1]{%
 \ifnum #1\expandafter \@firstoftwo
 \else \expandafter \@secondoftwo
 \fi
}%
\providecommand \@ifx [1]{%
 \ifx #1\expandafter \@firstoftwo
 \else \expandafter \@secondoftwo
 \fi
}%
\providecommand \natexlab [1]{#1}%
\providecommand \enquote  [1]{``#1''}%
\providecommand \bibnamefont  [1]{#1}%
\providecommand \bibfnamefont [1]{#1}%
\providecommand \citenamefont [1]{#1}%
\providecommand \href@noop [0]{\@secondoftwo}%
\providecommand \href [0]{\begingroup \@sanitize@url \@href}%
\providecommand \@href[1]{\@@startlink{#1}\@@href}%
\providecommand \@@href[1]{\endgroup#1\@@endlink}%
\providecommand \@sanitize@url [0]{\catcode `\\12\catcode `\$12\catcode
  `\&12\catcode `\#12\catcode `\^12\catcode `\_12\catcode `\%12\relax}%
\providecommand \@@startlink[1]{}%
\providecommand \@@endlink[0]{}%
\providecommand \url  [0]{\begingroup\@sanitize@url \@url }%
\providecommand \@url [1]{\endgroup\@href {#1}{\urlprefix }}%
\providecommand \urlprefix  [0]{URL }%
\providecommand \Eprint [0]{\href }%
\providecommand \doibase [0]{https://doi.org/}%
\providecommand \selectlanguage [0]{\@gobble}%
\providecommand \bibinfo  [0]{\@secondoftwo}%
\providecommand \bibfield  [0]{\@secondoftwo}%
\providecommand \translation [1]{[#1]}%
\providecommand \BibitemOpen [0]{}%
\providecommand \bibitemStop [0]{}%
\providecommand \bibitemNoStop [0]{.\EOS\space}%
\providecommand \EOS [0]{\spacefactor3000\relax}%
\providecommand \BibitemShut  [1]{\csname bibitem#1\endcsname}%
\let\auto@bib@innerbib\@empty
\bibitem [{\citenamefont {Aspuru-Guzik}\ \emph {et~al.}(2005)\citenamefont
  {Aspuru-Guzik}, \citenamefont {Dutoi}, \citenamefont {Love},\ and\
  \citenamefont {Head-Gordon}}]{aspuru2005simulated}%
  \BibitemOpen
  \bibfield  {author} {\bibinfo {author} {\bibfnamefont {A.}~\bibnamefont
  {Aspuru-Guzik}}, \bibinfo {author} {\bibfnamefont {A.~D.}\ \bibnamefont
  {Dutoi}}, \bibinfo {author} {\bibfnamefont {P.~J.}\ \bibnamefont {Love}},\
  and\ \bibinfo {author} {\bibfnamefont {M.}~\bibnamefont {Head-Gordon}},\
  }\bibfield  {title} {\bibinfo {title} {Simulated quantum computation of
  molecular energies},\ }\href@noop {} {\bibfield  {journal} {\bibinfo
  {journal} {Science}\ }\textbf {\bibinfo {volume} {309}},\ \bibinfo {pages}
  {1704} (\bibinfo {year} {2005})}\BibitemShut {NoStop}%
\bibitem [{\citenamefont {Cao}\ \emph {et~al.}(2019)\citenamefont {Cao},
  \citenamefont {Romero}, \citenamefont {Olson}, \citenamefont {Degroote},
  \citenamefont {Johnson}, \citenamefont {Kieferov{\'a}}, \citenamefont
  {Kivlichan}, \citenamefont {Menke}, \citenamefont {Peropadre}, \citenamefont
  {Sawaya} \emph {et~al.}}]{cao2019quantum}%
  \BibitemOpen
  \bibfield  {author} {\bibinfo {author} {\bibfnamefont {Y.}~\bibnamefont
  {Cao}}, \bibinfo {author} {\bibfnamefont {J.}~\bibnamefont {Romero}},
  \bibinfo {author} {\bibfnamefont {J.~P.}\ \bibnamefont {Olson}}, \bibinfo
  {author} {\bibfnamefont {M.}~\bibnamefont {Degroote}}, \bibinfo {author}
  {\bibfnamefont {P.~D.}\ \bibnamefont {Johnson}}, \bibinfo {author}
  {\bibfnamefont {M.}~\bibnamefont {Kieferov{\'a}}}, \bibinfo {author}
  {\bibfnamefont {I.~D.}\ \bibnamefont {Kivlichan}}, \bibinfo {author}
  {\bibfnamefont {T.}~\bibnamefont {Menke}}, \bibinfo {author} {\bibfnamefont
  {B.}~\bibnamefont {Peropadre}}, \bibinfo {author} {\bibfnamefont {N.~P.}\
  \bibnamefont {Sawaya}}, \emph {et~al.},\ }\bibfield  {title} {\bibinfo
  {title} {Quantum chemistry in the age of quantum computing},\ }\href@noop {}
  {\bibfield  {journal} {\bibinfo  {journal} {Chemical reviews}\ }\textbf
  {\bibinfo {volume} {119}},\ \bibinfo {pages} {10856} (\bibinfo {year}
  {2019})}\BibitemShut {NoStop}%
\bibitem [{\citenamefont {McArdle}\ \emph {et~al.}(2020)\citenamefont
  {McArdle}, \citenamefont {Endo}, \citenamefont {Aspuru-Guzik}, \citenamefont
  {Benjamin},\ and\ \citenamefont {Yuan}}]{mcardle2020quantum}%
  \BibitemOpen
  \bibfield  {author} {\bibinfo {author} {\bibfnamefont {S.}~\bibnamefont
  {McArdle}}, \bibinfo {author} {\bibfnamefont {S.}~\bibnamefont {Endo}},
  \bibinfo {author} {\bibfnamefont {A.}~\bibnamefont {Aspuru-Guzik}}, \bibinfo
  {author} {\bibfnamefont {S.~C.}\ \bibnamefont {Benjamin}},\ and\ \bibinfo
  {author} {\bibfnamefont {X.}~\bibnamefont {Yuan}},\ }\bibfield  {title}
  {\bibinfo {title} {Quantum computational chemistry},\ }\href@noop {}
  {\bibfield  {journal} {\bibinfo  {journal} {Reviews of Modern Physics}\
  }\textbf {\bibinfo {volume} {92}},\ \bibinfo {pages} {015003} (\bibinfo
  {year} {2020})}\BibitemShut {NoStop}%
\bibitem [{\citenamefont {Bauer}\ \emph {et~al.}(2016)\citenamefont {Bauer},
  \citenamefont {Wecker}, \citenamefont {Millis}, \citenamefont {Hastings},\
  and\ \citenamefont {Troyer}}]{bauer2016hybrid}%
  \BibitemOpen
  \bibfield  {author} {\bibinfo {author} {\bibfnamefont {B.}~\bibnamefont
  {Bauer}}, \bibinfo {author} {\bibfnamefont {D.}~\bibnamefont {Wecker}},
  \bibinfo {author} {\bibfnamefont {A.~J.}\ \bibnamefont {Millis}}, \bibinfo
  {author} {\bibfnamefont {M.~B.}\ \bibnamefont {Hastings}},\ and\ \bibinfo
  {author} {\bibfnamefont {M.}~\bibnamefont {Troyer}},\ }\bibfield  {title}
  {\bibinfo {title} {Hybrid quantum-classical approach to correlated
  materials},\ }\href@noop {} {\bibfield  {journal} {\bibinfo  {journal}
  {Physical Review X}\ }\textbf {\bibinfo {volume} {6}},\ \bibinfo {pages}
  {031045} (\bibinfo {year} {2016})}\BibitemShut {NoStop}%
\bibitem [{\citenamefont {Yoshioka}\ \emph {et~al.}(2024)\citenamefont
  {Yoshioka}, \citenamefont {Okubo}, \citenamefont {Suzuki}, \citenamefont
  {Koizumi},\ and\ \citenamefont {Mizukami}}]{yoshioka2024hunting}%
  \BibitemOpen
  \bibfield  {author} {\bibinfo {author} {\bibfnamefont {N.}~\bibnamefont
  {Yoshioka}}, \bibinfo {author} {\bibfnamefont {T.}~\bibnamefont {Okubo}},
  \bibinfo {author} {\bibfnamefont {Y.}~\bibnamefont {Suzuki}}, \bibinfo
  {author} {\bibfnamefont {Y.}~\bibnamefont {Koizumi}},\ and\ \bibinfo {author}
  {\bibfnamefont {W.}~\bibnamefont {Mizukami}},\ }\bibfield  {title} {\bibinfo
  {title} {Hunting for quantum-classical crossover in condensed matter
  problems},\ }\href@noop {} {\bibfield  {journal} {\bibinfo  {journal} {npj
  Quantum Information}\ }\textbf {\bibinfo {volume} {10}},\ \bibinfo {pages}
  {45} (\bibinfo {year} {2024})}\BibitemShut {NoStop}%
\bibitem [{\citenamefont {Vorwerk}\ \emph {et~al.}(2022)\citenamefont
  {Vorwerk}, \citenamefont {Sheng}, \citenamefont {Govoni}, \citenamefont
  {Huang},\ and\ \citenamefont {Galli}}]{vorwerk2022quantum}%
  \BibitemOpen
  \bibfield  {author} {\bibinfo {author} {\bibfnamefont {C.}~\bibnamefont
  {Vorwerk}}, \bibinfo {author} {\bibfnamefont {N.}~\bibnamefont {Sheng}},
  \bibinfo {author} {\bibfnamefont {M.}~\bibnamefont {Govoni}}, \bibinfo
  {author} {\bibfnamefont {B.}~\bibnamefont {Huang}},\ and\ \bibinfo {author}
  {\bibfnamefont {G.}~\bibnamefont {Galli}},\ }\bibfield  {title} {\bibinfo
  {title} {Quantum embedding theories to simulate condensed systems on quantum
  computers},\ }\href@noop {} {\bibfield  {journal} {\bibinfo  {journal}
  {Nature Computational Science}\ }\textbf {\bibinfo {volume} {2}},\ \bibinfo
  {pages} {424} (\bibinfo {year} {2022})}\BibitemShut {NoStop}%
\bibitem [{\citenamefont {Biamonte}\ \emph {et~al.}(2017)\citenamefont
  {Biamonte}, \citenamefont {Wittek}, \citenamefont {Pancotti}, \citenamefont
  {Rebentrost}, \citenamefont {Wiebe},\ and\ \citenamefont
  {Lloyd}}]{biamonte2017quantum}%
  \BibitemOpen
  \bibfield  {author} {\bibinfo {author} {\bibfnamefont {J.}~\bibnamefont
  {Biamonte}}, \bibinfo {author} {\bibfnamefont {P.}~\bibnamefont {Wittek}},
  \bibinfo {author} {\bibfnamefont {N.}~\bibnamefont {Pancotti}}, \bibinfo
  {author} {\bibfnamefont {P.}~\bibnamefont {Rebentrost}}, \bibinfo {author}
  {\bibfnamefont {N.}~\bibnamefont {Wiebe}},\ and\ \bibinfo {author}
  {\bibfnamefont {S.}~\bibnamefont {Lloyd}},\ }\bibfield  {title} {\bibinfo
  {title} {Quantum machine learning},\ }\href@noop {} {\bibfield  {journal}
  {\bibinfo  {journal} {Nature}\ }\textbf {\bibinfo {volume} {549}},\ \bibinfo
  {pages} {195} (\bibinfo {year} {2017})}\BibitemShut {NoStop}%
\bibitem [{\citenamefont {Cerezo}\ \emph {et~al.}(2022)\citenamefont {Cerezo},
  \citenamefont {Verdon}, \citenamefont {Huang}, \citenamefont {Cincio},\ and\
  \citenamefont {Coles}}]{cerezo2022challenges}%
  \BibitemOpen
  \bibfield  {author} {\bibinfo {author} {\bibfnamefont {M.}~\bibnamefont
  {Cerezo}}, \bibinfo {author} {\bibfnamefont {G.}~\bibnamefont {Verdon}},
  \bibinfo {author} {\bibfnamefont {H.-Y.}\ \bibnamefont {Huang}}, \bibinfo
  {author} {\bibfnamefont {L.}~\bibnamefont {Cincio}},\ and\ \bibinfo {author}
  {\bibfnamefont {P.~J.}\ \bibnamefont {Coles}},\ }\bibfield  {title} {\bibinfo
  {title} {Challenges and opportunities in quantum machine learning},\
  }\href@noop {} {\bibfield  {journal} {\bibinfo  {journal} {Nature
  computational science}\ }\textbf {\bibinfo {volume} {2}},\ \bibinfo {pages}
  {567} (\bibinfo {year} {2022})}\BibitemShut {NoStop}%
\bibitem [{\citenamefont {Schuld}\ \emph {et~al.}(2015)\citenamefont {Schuld},
  \citenamefont {Sinayskiy},\ and\ \citenamefont
  {Petruccione}}]{schuld2015introduction}%
  \BibitemOpen
  \bibfield  {author} {\bibinfo {author} {\bibfnamefont {M.}~\bibnamefont
  {Schuld}}, \bibinfo {author} {\bibfnamefont {I.}~\bibnamefont {Sinayskiy}},\
  and\ \bibinfo {author} {\bibfnamefont {F.}~\bibnamefont {Petruccione}},\
  }\bibfield  {title} {\bibinfo {title} {An introduction to quantum machine
  learning},\ }\href@noop {} {\bibfield  {journal} {\bibinfo  {journal}
  {Contemporary Physics}\ }\textbf {\bibinfo {volume} {56}},\ \bibinfo {pages}
  {172} (\bibinfo {year} {2015})}\BibitemShut {NoStop}%
\bibitem [{\citenamefont {Mitarai}\ \emph {et~al.}(2018)\citenamefont
  {Mitarai}, \citenamefont {Negoro}, \citenamefont {Kitagawa},\ and\
  \citenamefont {Fujii}}]{mitarai2018quantum}%
  \BibitemOpen
  \bibfield  {author} {\bibinfo {author} {\bibfnamefont {K.}~\bibnamefont
  {Mitarai}}, \bibinfo {author} {\bibfnamefont {M.}~\bibnamefont {Negoro}},
  \bibinfo {author} {\bibfnamefont {M.}~\bibnamefont {Kitagawa}},\ and\
  \bibinfo {author} {\bibfnamefont {K.}~\bibnamefont {Fujii}},\ }\bibfield
  {title} {\bibinfo {title} {Quantum circuit learning},\ }\href@noop {}
  {\bibfield  {journal} {\bibinfo  {journal} {Physical Review A}\ }\textbf
  {\bibinfo {volume} {98}},\ \bibinfo {pages} {032309} (\bibinfo {year}
  {2018})}\BibitemShut {NoStop}%
\bibitem [{\citenamefont {Lloyd}(1996)}]{lloyd1996universal}%
  \BibitemOpen
  \bibfield  {author} {\bibinfo {author} {\bibfnamefont {S.}~\bibnamefont
  {Lloyd}},\ }\bibfield  {title} {\bibinfo {title} {Universal quantum
  simulators},\ }\href@noop {} {\bibfield  {journal} {\bibinfo  {journal}
  {Science}\ }\textbf {\bibinfo {volume} {273}},\ \bibinfo {pages} {1073}
  (\bibinfo {year} {1996})}\BibitemShut {NoStop}%
\bibitem [{\citenamefont {Low}\ and\ \citenamefont
  {Chuang}(2019)}]{Low2019hamiltonian}%
  \BibitemOpen
  \bibfield  {author} {\bibinfo {author} {\bibfnamefont {G.~H.}\ \bibnamefont
  {Low}}\ and\ \bibinfo {author} {\bibfnamefont {I.~L.}\ \bibnamefont
  {Chuang}},\ }\bibfield  {title} {\bibinfo {title} {Hamiltonian {S}imulation
  by {Q}ubitization},\ }\href {https://doi.org/10.22331/q-2019-07-12-163}
  {\bibfield  {journal} {\bibinfo  {journal} {{Quantum}}\ }\textbf {\bibinfo
  {volume} {3}},\ \bibinfo {pages} {163} (\bibinfo {year} {2019})}\BibitemShut
  {NoStop}%
\bibitem [{\citenamefont {Berry}\ \emph {et~al.}(2015)\citenamefont {Berry},
  \citenamefont {Childs},\ and\ \citenamefont
  {Kothari}}]{berry2015hamiltonian}%
  \BibitemOpen
  \bibfield  {author} {\bibinfo {author} {\bibfnamefont {D.~W.}\ \bibnamefont
  {Berry}}, \bibinfo {author} {\bibfnamefont {A.~M.}\ \bibnamefont {Childs}},\
  and\ \bibinfo {author} {\bibfnamefont {R.}~\bibnamefont {Kothari}},\
  }\bibfield  {title} {\bibinfo {title} {Hamiltonian simulation with nearly
  optimal dependence on all parameters},\ }in\ \href
  {https://doi.org/10.1109/FOCS.2015.54} {\emph {\bibinfo {booktitle} {2015
  IEEE 56th annual symposium on foundations of computer science}}}\ (\bibinfo
  {organization} {IEEE},\ \bibinfo {year} {2015})\ pp.\ \bibinfo {pages}
  {792--809}\BibitemShut {NoStop}%
\bibitem [{\citenamefont {Motta}\ \emph {et~al.}(2020)\citenamefont {Motta},
  \citenamefont {Sun}, \citenamefont {Tan}, \citenamefont {O’Rourke},
  \citenamefont {Ye}, \citenamefont {Minnich}, \citenamefont {Brandao},\ and\
  \citenamefont {Chan}}]{motta2020determining}%
  \BibitemOpen
  \bibfield  {author} {\bibinfo {author} {\bibfnamefont {M.}~\bibnamefont
  {Motta}}, \bibinfo {author} {\bibfnamefont {C.}~\bibnamefont {Sun}}, \bibinfo
  {author} {\bibfnamefont {A.~T.}\ \bibnamefont {Tan}}, \bibinfo {author}
  {\bibfnamefont {M.~J.}\ \bibnamefont {O’Rourke}}, \bibinfo {author}
  {\bibfnamefont {E.}~\bibnamefont {Ye}}, \bibinfo {author} {\bibfnamefont
  {A.~J.}\ \bibnamefont {Minnich}}, \bibinfo {author} {\bibfnamefont {F.~G.}\
  \bibnamefont {Brandao}},\ and\ \bibinfo {author} {\bibfnamefont {G.~K.-L.}\
  \bibnamefont {Chan}},\ }\bibfield  {title} {\bibinfo {title} {Determining
  eigenstates and thermal states on a quantum computer using quantum imaginary
  time evolution},\ }\href@noop {} {\bibfield  {journal} {\bibinfo  {journal}
  {Nature Physics}\ }\textbf {\bibinfo {volume} {16}},\ \bibinfo {pages} {205}
  (\bibinfo {year} {2020})}\BibitemShut {NoStop}%
\bibitem [{\citenamefont {An}\ \emph {et~al.}(2023{\natexlab{a}})\citenamefont
  {An}, \citenamefont {Liu},\ and\ \citenamefont {Lin}}]{an2023linear}%
  \BibitemOpen
  \bibfield  {author} {\bibinfo {author} {\bibfnamefont {D.}~\bibnamefont
  {An}}, \bibinfo {author} {\bibfnamefont {J.-P.}\ \bibnamefont {Liu}},\ and\
  \bibinfo {author} {\bibfnamefont {L.}~\bibnamefont {Lin}},\ }\bibfield
  {title} {\bibinfo {title} {Linear combination of hamiltonian simulation for
  nonunitary dynamics with optimal state preparation cost},\ }\href
  {https://doi.org/10.1103/PhysRevLett.131.150603} {\bibfield  {journal}
  {\bibinfo  {journal} {Phys. Rev. Lett.}\ }\textbf {\bibinfo {volume} {131}},\
  \bibinfo {pages} {150603} (\bibinfo {year} {2023}{\natexlab{a}})}\BibitemShut
  {NoStop}%
\bibitem [{\citenamefont {An}\ \emph {et~al.}(2023{\natexlab{b}})\citenamefont
  {An}, \citenamefont {Childs},\ and\ \citenamefont {Lin}}]{an2023quantum}%
  \BibitemOpen
  \bibfield  {author} {\bibinfo {author} {\bibfnamefont {D.}~\bibnamefont
  {An}}, \bibinfo {author} {\bibfnamefont {A.~M.}\ \bibnamefont {Childs}},\
  and\ \bibinfo {author} {\bibfnamefont {L.}~\bibnamefont {Lin}},\ }\bibfield
  {title} {\bibinfo {title} {Quantum algorithm for linear non-unitary dynamics
  with near-optimal dependence on all parameters},\ }\href@noop {} {\bibfield
  {journal} {\bibinfo  {journal} {arXiv preprint arXiv:2312.03916}\ } (\bibinfo
  {year} {2023}{\natexlab{b}})}\BibitemShut {NoStop}%
\bibitem [{\citenamefont {Endo}\ \emph {et~al.}(2020)\citenamefont {Endo},
  \citenamefont {Sun}, \citenamefont {Li}, \citenamefont {Benjamin},\ and\
  \citenamefont {Yuan}}]{endo2020variational}%
  \BibitemOpen
  \bibfield  {author} {\bibinfo {author} {\bibfnamefont {S.}~\bibnamefont
  {Endo}}, \bibinfo {author} {\bibfnamefont {J.}~\bibnamefont {Sun}}, \bibinfo
  {author} {\bibfnamefont {Y.}~\bibnamefont {Li}}, \bibinfo {author}
  {\bibfnamefont {S.~C.}\ \bibnamefont {Benjamin}},\ and\ \bibinfo {author}
  {\bibfnamefont {X.}~\bibnamefont {Yuan}},\ }\bibfield  {title} {\bibinfo
  {title} {Variational quantum simulation of general processes},\ }\href@noop
  {} {\bibfield  {journal} {\bibinfo  {journal} {Physical Review Letters}\
  }\textbf {\bibinfo {volume} {125}},\ \bibinfo {pages} {010501} (\bibinfo
  {year} {2020})}\BibitemShut {NoStop}%
\bibitem [{\citenamefont {Kitaev}(1995)}]{kitaev1995quantum}%
  \BibitemOpen
  \bibfield  {author} {\bibinfo {author} {\bibfnamefont {A.~Y.}\ \bibnamefont
  {Kitaev}},\ }\bibfield  {title} {\bibinfo {title} {Quantum measurements and
  the abelian stabilizer problem},\ }\href@noop {} {\bibfield  {journal}
  {\bibinfo  {journal} {arXiv preprint quant-ph/9511026}\ } (\bibinfo {year}
  {1995})}\BibitemShut {NoStop}%
\bibitem [{\citenamefont {McClean}\ \emph {et~al.}(2017)\citenamefont
  {McClean}, \citenamefont {Kimchi-Schwartz}, \citenamefont {Carter},\ and\
  \citenamefont {De~Jong}}]{mcclean2017hybrid}%
  \BibitemOpen
  \bibfield  {author} {\bibinfo {author} {\bibfnamefont {J.~R.}\ \bibnamefont
  {McClean}}, \bibinfo {author} {\bibfnamefont {M.~E.}\ \bibnamefont
  {Kimchi-Schwartz}}, \bibinfo {author} {\bibfnamefont {J.}~\bibnamefont
  {Carter}},\ and\ \bibinfo {author} {\bibfnamefont {W.~A.}\ \bibnamefont
  {De~Jong}},\ }\bibfield  {title} {\bibinfo {title} {Hybrid quantum-classical
  hierarchy for mitigation of decoherence and determination of excited
  states},\ }\href@noop {} {\bibfield  {journal} {\bibinfo  {journal} {Physical
  Review A}\ }\textbf {\bibinfo {volume} {95}},\ \bibinfo {pages} {042308}
  (\bibinfo {year} {2017})}\BibitemShut {NoStop}%
\bibitem [{\citenamefont {Zeng}\ \emph {et~al.}(2021)\citenamefont {Zeng},
  \citenamefont {Sun},\ and\ \citenamefont {Yuan}}]{zeng2021universal}%
  \BibitemOpen
  \bibfield  {author} {\bibinfo {author} {\bibfnamefont {P.}~\bibnamefont
  {Zeng}}, \bibinfo {author} {\bibfnamefont {J.}~\bibnamefont {Sun}},\ and\
  \bibinfo {author} {\bibfnamefont {X.}~\bibnamefont {Yuan}},\ }\bibfield
  {title} {\bibinfo {title} {Universal quantum algorithmic cooling on a quantum
  computer},\ }\href@noop {} {\bibfield  {journal} {\bibinfo  {journal} {arXiv
  preprint arXiv:2109.15304}\ } (\bibinfo {year} {2021})}\BibitemShut {NoStop}%
\bibitem [{\citenamefont {Childs}\ \emph {et~al.}(2017)\citenamefont {Childs},
  \citenamefont {Kothari},\ and\ \citenamefont {Somma}}]{childs2017quantum}%
  \BibitemOpen
  \bibfield  {author} {\bibinfo {author} {\bibfnamefont {A.~M.}\ \bibnamefont
  {Childs}}, \bibinfo {author} {\bibfnamefont {R.}~\bibnamefont {Kothari}},\
  and\ \bibinfo {author} {\bibfnamefont {R.~D.}\ \bibnamefont {Somma}},\
  }\bibfield  {title} {\bibinfo {title} {Quantum algorithm for systems of
  linear equations with exponentially improved dependence on precision},\
  }\href@noop {} {\bibfield  {journal} {\bibinfo  {journal} {SIAM Journal on
  Computing}\ }\textbf {\bibinfo {volume} {46}},\ \bibinfo {pages} {1920}
  (\bibinfo {year} {2017})}\BibitemShut {NoStop}%
\bibitem [{\citenamefont {Wang}\ \emph {et~al.}(2024)\citenamefont {Wang},
  \citenamefont {McArdle},\ and\ \citenamefont {Berta}}]{wang2024qubit}%
  \BibitemOpen
  \bibfield  {author} {\bibinfo {author} {\bibfnamefont {S.}~\bibnamefont
  {Wang}}, \bibinfo {author} {\bibfnamefont {S.}~\bibnamefont {McArdle}},\ and\
  \bibinfo {author} {\bibfnamefont {M.}~\bibnamefont {Berta}},\ }\bibfield
  {title} {\bibinfo {title} {Qubit-efficient randomized quantum algorithms for
  linear algebra},\ }\href {https://doi.org/10.1103/PRXQuantum.5.020324}
  {\bibfield  {journal} {\bibinfo  {journal} {PRX Quantum}\ }\textbf {\bibinfo
  {volume} {5}},\ \bibinfo {pages} {020324} (\bibinfo {year}
  {2024})}\BibitemShut {NoStop}%
\bibitem [{\citenamefont {Kim}\ \emph {et~al.}(2023)\citenamefont {Kim},
  \citenamefont {Eddins}, \citenamefont {Anand}, \citenamefont {Wei},
  \citenamefont {Van Den~Berg}, \citenamefont {Rosenblatt}, \citenamefont
  {Nayfeh}, \citenamefont {Wu}, \citenamefont {Zaletel}, \citenamefont {Temme}
  \emph {et~al.}}]{kim2023evidence}%
  \BibitemOpen
  \bibfield  {author} {\bibinfo {author} {\bibfnamefont {Y.}~\bibnamefont
  {Kim}}, \bibinfo {author} {\bibfnamefont {A.}~\bibnamefont {Eddins}},
  \bibinfo {author} {\bibfnamefont {S.}~\bibnamefont {Anand}}, \bibinfo
  {author} {\bibfnamefont {K.~X.}\ \bibnamefont {Wei}}, \bibinfo {author}
  {\bibfnamefont {E.}~\bibnamefont {Van Den~Berg}}, \bibinfo {author}
  {\bibfnamefont {S.}~\bibnamefont {Rosenblatt}}, \bibinfo {author}
  {\bibfnamefont {H.}~\bibnamefont {Nayfeh}}, \bibinfo {author} {\bibfnamefont
  {Y.}~\bibnamefont {Wu}}, \bibinfo {author} {\bibfnamefont {M.}~\bibnamefont
  {Zaletel}}, \bibinfo {author} {\bibfnamefont {K.}~\bibnamefont {Temme}},
  \emph {et~al.},\ }\bibfield  {title} {\bibinfo {title} {Evidence for the
  utility of quantum computing before fault tolerance},\ }\href@noop {}
  {\bibfield  {journal} {\bibinfo  {journal} {Nature}\ }\textbf {\bibinfo
  {volume} {618}},\ \bibinfo {pages} {500} (\bibinfo {year}
  {2023})}\BibitemShut {NoStop}%
\bibitem [{\citenamefont {AI}\ and\ \citenamefont
  {Collaborators}(2025)}]{google2025quantum}%
  \BibitemOpen
  \bibfield  {author} {\bibinfo {author} {\bibfnamefont {G.~Q.}\ \bibnamefont
  {AI}}\ and\ \bibinfo {author} {\bibnamefont {Collaborators}},\ }\bibfield
  {title} {\bibinfo {title} {Quantum error correction below the surface code
  threshold},\ }\href@noop {} {\bibfield  {journal} {\bibinfo  {journal}
  {Nature}\ }\textbf {\bibinfo {volume} {638}},\ \bibinfo {pages} {920}
  (\bibinfo {year} {2025})}\BibitemShut {NoStop}%
\bibitem [{\citenamefont {Katabarwa}\ \emph {et~al.}(2024)\citenamefont
  {Katabarwa}, \citenamefont {Gratsea}, \citenamefont {Caesura},\ and\
  \citenamefont {Johnson}}]{PRXQuantum.5.020101}%
  \BibitemOpen
  \bibfield  {author} {\bibinfo {author} {\bibfnamefont {A.}~\bibnamefont
  {Katabarwa}}, \bibinfo {author} {\bibfnamefont {K.}~\bibnamefont {Gratsea}},
  \bibinfo {author} {\bibfnamefont {A.}~\bibnamefont {Caesura}},\ and\ \bibinfo
  {author} {\bibfnamefont {P.~D.}\ \bibnamefont {Johnson}},\ }\bibfield
  {title} {\bibinfo {title} {Early fault-tolerant quantum computing},\ }\href
  {https://doi.org/10.1103/PRXQuantum.5.020101} {\bibfield  {journal} {\bibinfo
   {journal} {PRX Quantum}\ }\textbf {\bibinfo {volume} {5}},\ \bibinfo {pages}
  {020101} (\bibinfo {year} {2024})}\BibitemShut {NoStop}%
\bibitem [{\citenamefont {Campbell}(2021)}]{campbell2021early}%
  \BibitemOpen
  \bibfield  {author} {\bibinfo {author} {\bibfnamefont {E.~T.}\ \bibnamefont
  {Campbell}},\ }\bibfield  {title} {\bibinfo {title} {Early fault-tolerant
  simulations of the hubbard model},\ }\href@noop {} {\bibfield  {journal}
  {\bibinfo  {journal} {Quantum Science and Technology}\ }\textbf {\bibinfo
  {volume} {7}},\ \bibinfo {pages} {015007} (\bibinfo {year}
  {2021})}\BibitemShut {NoStop}%
\bibitem [{\citenamefont {Suzuki}\ \emph {et~al.}(2022)\citenamefont {Suzuki},
  \citenamefont {Endo}, \citenamefont {Fujii},\ and\ \citenamefont
  {Tokunaga}}]{suzuki2022quantum}%
  \BibitemOpen
  \bibfield  {author} {\bibinfo {author} {\bibfnamefont {Y.}~\bibnamefont
  {Suzuki}}, \bibinfo {author} {\bibfnamefont {S.}~\bibnamefont {Endo}},
  \bibinfo {author} {\bibfnamefont {K.}~\bibnamefont {Fujii}},\ and\ \bibinfo
  {author} {\bibfnamefont {Y.}~\bibnamefont {Tokunaga}},\ }\bibfield  {title}
  {\bibinfo {title} {Quantum error mitigation as a universal error reduction
  technique: Applications from the nisq to the fault-tolerant quantum computing
  eras},\ }\href@noop {} {\bibfield  {journal} {\bibinfo  {journal} {PRX
  Quantum}\ }\textbf {\bibinfo {volume} {3}},\ \bibinfo {pages} {010345}
  (\bibinfo {year} {2022})}\BibitemShut {NoStop}%
\bibitem [{\citenamefont {Chakraborty}(2024)}]{Chakraborty2024implementingany}%
  \BibitemOpen
  \bibfield  {author} {\bibinfo {author} {\bibfnamefont {S.}~\bibnamefont
  {Chakraborty}},\ }\bibfield  {title} {\bibinfo {title} {Implementing any
  {L}inear {C}ombination of {U}nitaries on {I}ntermediate-term {Q}uantum
  {C}omputers},\ }\href {https://doi.org/10.22331/q-2024-10-10-1496} {\bibfield
   {journal} {\bibinfo  {journal} {{Quantum}}\ }\textbf {\bibinfo {volume}
  {8}},\ \bibinfo {pages} {1496} (\bibinfo {year} {2024})}\BibitemShut
  {NoStop}%
\bibitem [{\citenamefont {Zeng}\ \emph {et~al.}(2022)\citenamefont {Zeng},
  \citenamefont {Sun}, \citenamefont {Jiang},\ and\ \citenamefont
  {Zhao}}]{zeng2022simple}%
  \BibitemOpen
  \bibfield  {author} {\bibinfo {author} {\bibfnamefont {P.}~\bibnamefont
  {Zeng}}, \bibinfo {author} {\bibfnamefont {J.}~\bibnamefont {Sun}}, \bibinfo
  {author} {\bibfnamefont {L.}~\bibnamefont {Jiang}},\ and\ \bibinfo {author}
  {\bibfnamefont {Q.}~\bibnamefont {Zhao}},\ }\bibfield  {title} {\bibinfo
  {title} {Simple and high-precision hamiltonian simulation by compensating
  trotter error with linear combination of unitary operations},\ }\href@noop {}
  {\bibfield  {journal} {\bibinfo  {journal} {arXiv preprint arXiv:2212.04566}\
  } (\bibinfo {year} {2022})}\BibitemShut {NoStop}%
\bibitem [{\citenamefont {Wan}\ \emph {et~al.}(2022)\citenamefont {Wan},
  \citenamefont {Berta},\ and\ \citenamefont {Campbell}}]{wan2022randomized}%
  \BibitemOpen
  \bibfield  {author} {\bibinfo {author} {\bibfnamefont {K.}~\bibnamefont
  {Wan}}, \bibinfo {author} {\bibfnamefont {M.}~\bibnamefont {Berta}},\ and\
  \bibinfo {author} {\bibfnamefont {E.~T.}\ \bibnamefont {Campbell}},\
  }\bibfield  {title} {\bibinfo {title} {Randomized quantum algorithm for
  statistical phase estimation},\ }\href@noop {} {\bibfield  {journal}
  {\bibinfo  {journal} {Physical Review Letters}\ }\textbf {\bibinfo {volume}
  {129}},\ \bibinfo {pages} {030503} (\bibinfo {year} {2022})}\BibitemShut
  {NoStop}%
\bibitem [{\citenamefont {Kato}\ \emph {et~al.}(2024)\citenamefont {Kato},
  \citenamefont {Wada}, \citenamefont {Ito},\ and\ \citenamefont
  {Yamamoto}}]{kato2024exponentially}%
  \BibitemOpen
  \bibfield  {author} {\bibinfo {author} {\bibfnamefont {J.}~\bibnamefont
  {Kato}}, \bibinfo {author} {\bibfnamefont {K.}~\bibnamefont {Wada}}, \bibinfo
  {author} {\bibfnamefont {K.}~\bibnamefont {Ito}},\ and\ \bibinfo {author}
  {\bibfnamefont {N.}~\bibnamefont {Yamamoto}},\ }\bibfield  {title} {\bibinfo
  {title} {Exponentially accurate open quantum simulation via randomized
  dissipation with minimal ancilla},\ }\href@noop {} {\bibfield  {journal}
  {\bibinfo  {journal} {arXiv preprint arXiv:2412.19453}\ } (\bibinfo {year}
  {2024})}\BibitemShut {NoStop}%
\bibitem [{\citenamefont {Huo}\ and\ \citenamefont {Li}(2023)}]{huo2023error}%
  \BibitemOpen
  \bibfield  {author} {\bibinfo {author} {\bibfnamefont {M.}~\bibnamefont
  {Huo}}\ and\ \bibinfo {author} {\bibfnamefont {Y.}~\bibnamefont {Li}},\
  }\bibfield  {title} {\bibinfo {title} {Error-resilient {M}onte {C}arlo
  quantum simulation of imaginary time},\ }\href
  {https://doi.org/10.22331/q-2023-02-09-916} {\bibfield  {journal} {\bibinfo
  {journal} {{Quantum}}\ }\textbf {\bibinfo {volume} {7}},\ \bibinfo {pages}
  {916} (\bibinfo {year} {2023})}\BibitemShut {NoStop}%
\bibitem [{\citenamefont {Faehrmann}\ \emph {et~al.}(2022)\citenamefont
  {Faehrmann}, \citenamefont {Steudtner}, \citenamefont {Kueng}, \citenamefont
  {Kieferov{\'a}},\ and\ \citenamefont {Eisert}}]{faehrmann2022randomizing}%
  \BibitemOpen
  \bibfield  {author} {\bibinfo {author} {\bibfnamefont {P.~K.}\ \bibnamefont
  {Faehrmann}}, \bibinfo {author} {\bibfnamefont {M.}~\bibnamefont
  {Steudtner}}, \bibinfo {author} {\bibfnamefont {R.}~\bibnamefont {Kueng}},
  \bibinfo {author} {\bibfnamefont {M.}~\bibnamefont {Kieferov{\'a}}},\ and\
  \bibinfo {author} {\bibfnamefont {J.}~\bibnamefont {Eisert}},\ }\bibfield
  {title} {\bibinfo {title} {Randomizing multi-product formulas for hamiltonian
  simulation},\ }\href@noop {} {\bibfield  {journal} {\bibinfo  {journal}
  {Quantum}\ }\textbf {\bibinfo {volume} {6}},\ \bibinfo {pages} {806}
  (\bibinfo {year} {2022})}\BibitemShut {NoStop}%
\bibitem [{\citenamefont {Yang}\ \emph {et~al.}(2021)\citenamefont {Yang},
  \citenamefont {Lu},\ and\ \citenamefont {Li}}]{yang2021accelerated}%
  \BibitemOpen
  \bibfield  {author} {\bibinfo {author} {\bibfnamefont {Y.}~\bibnamefont
  {Yang}}, \bibinfo {author} {\bibfnamefont {B.-N.}\ \bibnamefont {Lu}},\ and\
  \bibinfo {author} {\bibfnamefont {Y.}~\bibnamefont {Li}},\ }\bibfield
  {title} {\bibinfo {title} {Accelerated quantum monte carlo with mitigated
  error on noisy quantum computer},\ }\href@noop {} {\bibfield  {journal}
  {\bibinfo  {journal} {PRX Quantum}\ }\textbf {\bibinfo {volume} {2}},\
  \bibinfo {pages} {040361} (\bibinfo {year} {2021})}\BibitemShut {NoStop}%
\bibitem [{\citenamefont {Wada}\ \emph {et~al.}(2025)\citenamefont {Wada},
  \citenamefont {Kato}, \citenamefont {Harada},\ and\ \citenamefont
  {Yamamoto}}]{wada2025state}%
  \BibitemOpen
  \bibfield  {author} {\bibinfo {author} {\bibfnamefont {K.}~\bibnamefont
  {Wada}}, \bibinfo {author} {\bibfnamefont {J.}~\bibnamefont {Kato}}, \bibinfo
  {author} {\bibfnamefont {H.}~\bibnamefont {Harada}},\ and\ \bibinfo {author}
  {\bibfnamefont {N.}~\bibnamefont {Yamamoto}},\ }\bibfield  {title} {\bibinfo
  {title} {State-to-hamiltonian conversion with a few copies},\ }\href@noop {}
  {\bibfield  {journal} {\bibinfo  {journal} {arXiv preprint arXiv:2509.14791}\
  } (\bibinfo {year} {2025})}\BibitemShut {NoStop}%
\bibitem [{\citenamefont {Lin}\ and\ \citenamefont
  {Tong}(2022)}]{lin2022heisenberg}%
  \BibitemOpen
  \bibfield  {author} {\bibinfo {author} {\bibfnamefont {L.}~\bibnamefont
  {Lin}}\ and\ \bibinfo {author} {\bibfnamefont {Y.}~\bibnamefont {Tong}},\
  }\bibfield  {title} {\bibinfo {title} {Heisenberg-limited ground-state energy
  estimation for early fault-tolerant quantum computers},\ }\href@noop {}
  {\bibfield  {journal} {\bibinfo  {journal} {PRX quantum}\ }\textbf {\bibinfo
  {volume} {3}},\ \bibinfo {pages} {010318} (\bibinfo {year}
  {2022})}\BibitemShut {NoStop}%
\bibitem [{\citenamefont {Zhang}\ \emph {et~al.}(2022)\citenamefont {Zhang},
  \citenamefont {Wang},\ and\ \citenamefont {Johnson}}]{zhang2022computing}%
  \BibitemOpen
  \bibfield  {author} {\bibinfo {author} {\bibfnamefont {R.}~\bibnamefont
  {Zhang}}, \bibinfo {author} {\bibfnamefont {G.}~\bibnamefont {Wang}},\ and\
  \bibinfo {author} {\bibfnamefont {P.}~\bibnamefont {Johnson}},\ }\bibfield
  {title} {\bibinfo {title} {Computing ground state properties with early
  fault-tolerant quantum computers},\ }\href@noop {} {\bibfield  {journal}
  {\bibinfo  {journal} {Quantum}\ }\textbf {\bibinfo {volume} {6}},\ \bibinfo
  {pages} {761} (\bibinfo {year} {2022})}\BibitemShut {NoStop}%
\bibitem [{\citenamefont {Sun}\ \emph {et~al.}(2024)\citenamefont {Sun},
  \citenamefont {Zeng}, \citenamefont {Gur},\ and\ \citenamefont
  {Kim}}]{sun2024high}%
  \BibitemOpen
  \bibfield  {author} {\bibinfo {author} {\bibfnamefont {J.}~\bibnamefont
  {Sun}}, \bibinfo {author} {\bibfnamefont {P.}~\bibnamefont {Zeng}}, \bibinfo
  {author} {\bibfnamefont {T.}~\bibnamefont {Gur}},\ and\ \bibinfo {author}
  {\bibfnamefont {M.}~\bibnamefont {Kim}},\ }\bibfield  {title} {\bibinfo
  {title} {High-precision and low-depth eigenstate property estimation: theory
  and resource estimation},\ }\href@noop {} {\bibfield  {journal} {\bibinfo
  {journal} {arXiv preprint arXiv:2406.04307}\ } (\bibinfo {year}
  {2024})}\BibitemShut {NoStop}%
\bibitem [{\citenamefont {Tsubouchi}\ \emph {et~al.}(2023)\citenamefont
  {Tsubouchi}, \citenamefont {Suzuki}, \citenamefont {Tokunaga}, \citenamefont
  {Yoshioka},\ and\ \citenamefont {Endo}}]{tsubouchi2023virtual}%
  \BibitemOpen
  \bibfield  {author} {\bibinfo {author} {\bibfnamefont {K.}~\bibnamefont
  {Tsubouchi}}, \bibinfo {author} {\bibfnamefont {Y.}~\bibnamefont {Suzuki}},
  \bibinfo {author} {\bibfnamefont {Y.}~\bibnamefont {Tokunaga}}, \bibinfo
  {author} {\bibfnamefont {N.}~\bibnamefont {Yoshioka}},\ and\ \bibinfo
  {author} {\bibfnamefont {S.}~\bibnamefont {Endo}},\ }\bibfield  {title}
  {\bibinfo {title} {Virtual quantum error detection},\ }\href@noop {}
  {\bibfield  {journal} {\bibinfo  {journal} {Physical Review A}\ }\textbf
  {\bibinfo {volume} {108}},\ \bibinfo {pages} {042426} (\bibinfo {year}
  {2023})}\BibitemShut {NoStop}%
\bibitem [{\citenamefont {Childs}\ and\ \citenamefont
  {Wiebe}(2012)}]{childs2012hamiltonian}%
  \BibitemOpen
  \bibfield  {author} {\bibinfo {author} {\bibfnamefont {A.~M.}\ \bibnamefont
  {Childs}}\ and\ \bibinfo {author} {\bibfnamefont {N.}~\bibnamefont {Wiebe}},\
  }\bibfield  {title} {\bibinfo {title} {Hamiltonian simulation using linear
  combinations of unitary operations},\ }\bibfield  {journal} {\bibinfo
  {journal} {arXiv preprint arXiv:1202.5822}\ }\href
  {https://doi.org/https://doi.org/10.26421/QIC12.11-12}
  {https://doi.org/10.26421/QIC12.11-12} (\bibinfo {year} {2012})\BibitemShut
  {NoStop}%
\bibitem [{\citenamefont {Ge}\ \emph {et~al.}(2019)\citenamefont {Ge},
  \citenamefont {Tura},\ and\ \citenamefont {Cirac}}]{ge2019faster}%
  \BibitemOpen
  \bibfield  {author} {\bibinfo {author} {\bibfnamefont {Y.}~\bibnamefont
  {Ge}}, \bibinfo {author} {\bibfnamefont {J.}~\bibnamefont {Tura}},\ and\
  \bibinfo {author} {\bibfnamefont {J.~I.}\ \bibnamefont {Cirac}},\ }\bibfield
  {title} {\bibinfo {title} {Faster ground state preparation and high-precision
  ground energy estimation with fewer qubits},\ }\href@noop {} {\bibfield
  {journal} {\bibinfo  {journal} {Journal of Mathematical Physics}\ }\textbf
  {\bibinfo {volume} {60}} (\bibinfo {year} {2019})}\BibitemShut {NoStop}%
\bibitem [{\citenamefont {Endo}\ \emph {et~al.}(2024)\citenamefont {Endo},
  \citenamefont {Anai}, \citenamefont {Matsuzaki}, \citenamefont {Tokunaga},\
  and\ \citenamefont {Suzuki}}]{endo2024projective}%
  \BibitemOpen
  \bibfield  {author} {\bibinfo {author} {\bibfnamefont {S.}~\bibnamefont
  {Endo}}, \bibinfo {author} {\bibfnamefont {K.}~\bibnamefont {Anai}}, \bibinfo
  {author} {\bibfnamefont {Y.}~\bibnamefont {Matsuzaki}}, \bibinfo {author}
  {\bibfnamefont {Y.}~\bibnamefont {Tokunaga}},\ and\ \bibinfo {author}
  {\bibfnamefont {Y.}~\bibnamefont {Suzuki}},\ }\bibfield  {title} {\bibinfo
  {title} {Projective squeezing for translation symmetric bosonic codes},\
  }\href@noop {} {\bibfield  {journal} {\bibinfo  {journal} {arXiv preprint
  arXiv:2403.14218}\ } (\bibinfo {year} {2024})}\BibitemShut {NoStop}%
\bibitem [{\citenamefont {Endo}\ \emph {et~al.}(2025)\citenamefont {Endo},
  \citenamefont {Suzuki}, \citenamefont {Tsubouchi}, \citenamefont {Asaoka},
  \citenamefont {Yamamoto}, \citenamefont {Matsuzaki},\ and\ \citenamefont
  {Tokunaga}}]{endo2025quantum}%
  \BibitemOpen
  \bibfield  {author} {\bibinfo {author} {\bibfnamefont {S.}~\bibnamefont
  {Endo}}, \bibinfo {author} {\bibfnamefont {Y.}~\bibnamefont {Suzuki}},
  \bibinfo {author} {\bibfnamefont {K.}~\bibnamefont {Tsubouchi}}, \bibinfo
  {author} {\bibfnamefont {R.}~\bibnamefont {Asaoka}}, \bibinfo {author}
  {\bibfnamefont {K.}~\bibnamefont {Yamamoto}}, \bibinfo {author}
  {\bibfnamefont {Y.}~\bibnamefont {Matsuzaki}},\ and\ \bibinfo {author}
  {\bibfnamefont {Y.}~\bibnamefont {Tokunaga}},\ }\bibfield  {title} {\bibinfo
  {title} {Quantum error mitigation for rotation-symmetric bosonic codes with
  symmetry expansion},\ }\href@noop {} {\bibfield  {journal} {\bibinfo
  {journal} {Physical Review A}\ }\textbf {\bibinfo {volume} {111}},\ \bibinfo
  {pages} {062402} (\bibinfo {year} {2025})}\BibitemShut {NoStop}%
\bibitem [{\citenamefont {Anai}\ \emph {et~al.}(2024)\citenamefont {Anai},
  \citenamefont {Suzuki}, \citenamefont {Tokunaga}, \citenamefont {Matsuzaki},
  \citenamefont {Takeda},\ and\ \citenamefont {Endo}}]{anai2024unitary}%
  \BibitemOpen
  \bibfield  {author} {\bibinfo {author} {\bibfnamefont {K.}~\bibnamefont
  {Anai}}, \bibinfo {author} {\bibfnamefont {Y.}~\bibnamefont {Suzuki}},
  \bibinfo {author} {\bibfnamefont {Y.}~\bibnamefont {Tokunaga}}, \bibinfo
  {author} {\bibfnamefont {Y.}~\bibnamefont {Matsuzaki}}, \bibinfo {author}
  {\bibfnamefont {S.}~\bibnamefont {Takeda}},\ and\ \bibinfo {author}
  {\bibfnamefont {S.}~\bibnamefont {Endo}},\ }\bibfield  {title} {\bibinfo
  {title} {Unitary-transformed projective squeezing: applications for
  circuit-knitting and state-preparation of non-gaussian states},\ }\href@noop
  {} {\bibfield  {journal} {\bibinfo  {journal} {arXiv preprint
  arXiv:2411.19505}\ } (\bibinfo {year} {2024})}\BibitemShut {NoStop}%
\bibitem [{\citenamefont {Cui}\ \emph {et~al.}(2012)\citenamefont {Cui},
  \citenamefont {Zhou},\ and\ \citenamefont {Long}}]{cui2012optimal}%
  \BibitemOpen
  \bibfield  {author} {\bibinfo {author} {\bibfnamefont {J.~X.}\ \bibnamefont
  {Cui}}, \bibinfo {author} {\bibfnamefont {T.}~\bibnamefont {Zhou}},\ and\
  \bibinfo {author} {\bibfnamefont {G.~L.}\ \bibnamefont {Long}},\ }\bibfield
  {title} {\bibinfo {title} {An optimal expression of a kraus operator as a
  linear combination of unitary matrices},\ }\href
  {https://doi.org/10.1088/1751-8113/45/44/444011} {\bibfield  {journal}
  {\bibinfo  {journal} {Journal of Physics A: Mathematical and Theoretical}\
  }\textbf {\bibinfo {volume} {45}},\ \bibinfo {pages} {444011} (\bibinfo
  {year} {2012})}\BibitemShut {NoStop}%
\bibitem [{\citenamefont {Zindorf}\ and\ \citenamefont
  {Bose}(2025)}]{zindorf2025efficient}%
  \BibitemOpen
  \bibfield  {author} {\bibinfo {author} {\bibfnamefont {B.}~\bibnamefont
  {Zindorf}}\ and\ \bibinfo {author} {\bibfnamefont {S.}~\bibnamefont {Bose}},\
  }\bibfield  {title} {\bibinfo {title} {Efficient implementation of
  multicontrolled quantum gates},\ }\href@noop {} {\bibfield  {journal}
  {\bibinfo  {journal} {Physical Review Applied}\ }\textbf {\bibinfo {volume}
  {24}},\ \bibinfo {pages} {044030} (\bibinfo {year} {2025})}\BibitemShut
  {NoStop}%
\bibitem [{\citenamefont {Gily{\'e}n}\ \emph {et~al.}(2019)\citenamefont
  {Gily{\'e}n}, \citenamefont {Su}, \citenamefont {Low},\ and\ \citenamefont
  {Wiebe}}]{gilyen2019quantum}%
  \BibitemOpen
  \bibfield  {author} {\bibinfo {author} {\bibfnamefont {A.}~\bibnamefont
  {Gily{\'e}n}}, \bibinfo {author} {\bibfnamefont {Y.}~\bibnamefont {Su}},
  \bibinfo {author} {\bibfnamefont {G.~H.}\ \bibnamefont {Low}},\ and\ \bibinfo
  {author} {\bibfnamefont {N.}~\bibnamefont {Wiebe}},\ }\bibfield  {title}
  {\bibinfo {title} {Quantum singular value transformation and beyond:
  exponential improvements for quantum matrix arithmetics},\ }in\ \href@noop {}
  {\emph {\bibinfo {booktitle} {Proceedings of the 51st annual ACM SIGACT
  symposium on theory of computing}}}\ (\bibinfo {year} {2019})\ pp.\ \bibinfo
  {pages} {193--204}\BibitemShut {NoStop}%
\bibitem [{\citenamefont {Dong}\ \emph {et~al.}(2022)\citenamefont {Dong},
  \citenamefont {Lin},\ and\ \citenamefont {Tong}}]{dong2022ground}%
  \BibitemOpen
  \bibfield  {author} {\bibinfo {author} {\bibfnamefont {Y.}~\bibnamefont
  {Dong}}, \bibinfo {author} {\bibfnamefont {L.}~\bibnamefont {Lin}},\ and\
  \bibinfo {author} {\bibfnamefont {Y.}~\bibnamefont {Tong}},\ }\bibfield
  {title} {\bibinfo {title} {Ground-state preparation and energy estimation on
  early fault-tolerant quantum computers via quantum eigenvalue transformation
  of unitary matrices},\ }\href@noop {} {\bibfield  {journal} {\bibinfo
  {journal} {PRX quantum}\ }\textbf {\bibinfo {volume} {3}},\ \bibinfo {pages}
  {040305} (\bibinfo {year} {2022})}\BibitemShut {NoStop}%
\bibitem [{\citenamefont {Knill}\ \emph {et~al.}(2000)\citenamefont {Knill},
  \citenamefont {Laflamme},\ and\ \citenamefont {Viola}}]{knill2000theory}%
  \BibitemOpen
  \bibfield  {author} {\bibinfo {author} {\bibfnamefont {E.}~\bibnamefont
  {Knill}}, \bibinfo {author} {\bibfnamefont {R.}~\bibnamefont {Laflamme}},\
  and\ \bibinfo {author} {\bibfnamefont {L.}~\bibnamefont {Viola}},\ }\bibfield
   {title} {\bibinfo {title} {Theory of quantum error correction for general
  noise},\ }\href@noop {} {\bibfield  {journal} {\bibinfo  {journal} {Physical
  Review Letters}\ }\textbf {\bibinfo {volume} {84}},\ \bibinfo {pages} {2525}
  (\bibinfo {year} {2000})}\BibitemShut {NoStop}%
\bibitem [{\citenamefont {McClean}\ \emph {et~al.}(2020)\citenamefont
  {McClean}, \citenamefont {Jiang}, \citenamefont {Rubin}, \citenamefont
  {Babbush},\ and\ \citenamefont {Neven}}]{mcclean2020decoding}%
  \BibitemOpen
  \bibfield  {author} {\bibinfo {author} {\bibfnamefont {J.~R.}\ \bibnamefont
  {McClean}}, \bibinfo {author} {\bibfnamefont {Z.}~\bibnamefont {Jiang}},
  \bibinfo {author} {\bibfnamefont {N.~C.}\ \bibnamefont {Rubin}}, \bibinfo
  {author} {\bibfnamefont {R.}~\bibnamefont {Babbush}},\ and\ \bibinfo {author}
  {\bibfnamefont {H.}~\bibnamefont {Neven}},\ }\bibfield  {title} {\bibinfo
  {title} {Decoding quantum errors with subspace expansions},\ }\href@noop {}
  {\bibfield  {journal} {\bibinfo  {journal} {Nature communications}\ }\textbf
  {\bibinfo {volume} {11}},\ \bibinfo {pages} {636} (\bibinfo {year}
  {2020})}\BibitemShut {NoStop}%
\bibitem [{\citenamefont {Cai}(2021)}]{cai2021quantum}%
  \BibitemOpen
  \bibfield  {author} {\bibinfo {author} {\bibfnamefont {Z.}~\bibnamefont
  {Cai}},\ }\bibfield  {title} {\bibinfo {title} {Quantum error mitigation
  using symmetry expansion},\ }\href@noop {} {\bibfield  {journal} {\bibinfo
  {journal} {Quantum}\ }\textbf {\bibinfo {volume} {5}},\ \bibinfo {pages}
  {548} (\bibinfo {year} {2021})}\BibitemShut {NoStop}%
\bibitem [{\citenamefont {Steane}(1996)}]{steane1996simple}%
  \BibitemOpen
  \bibfield  {author} {\bibinfo {author} {\bibfnamefont {A.~M.}\ \bibnamefont
  {Steane}},\ }\bibfield  {title} {\bibinfo {title} {Simple quantum
  error-correcting codes},\ }\href@noop {} {\bibfield  {journal} {\bibinfo
  {journal} {Physical Review A}\ }\textbf {\bibinfo {volume} {54}},\ \bibinfo
  {pages} {4741} (\bibinfo {year} {1996})}\BibitemShut {NoStop}%
\bibitem [{\citenamefont {Leghtas}\ \emph {et~al.}(2015)\citenamefont
  {Leghtas}, \citenamefont {Touzard}, \citenamefont {Pop}, \citenamefont {Kou},
  \citenamefont {Vlastakis}, \citenamefont {Petrenko}, \citenamefont {Sliwa},
  \citenamefont {Narla}, \citenamefont {Shankar}, \citenamefont {Hatridge}
  \emph {et~al.}}]{leghtas2015confining}%
  \BibitemOpen
  \bibfield  {author} {\bibinfo {author} {\bibfnamefont {Z.}~\bibnamefont
  {Leghtas}}, \bibinfo {author} {\bibfnamefont {S.}~\bibnamefont {Touzard}},
  \bibinfo {author} {\bibfnamefont {I.~M.}\ \bibnamefont {Pop}}, \bibinfo
  {author} {\bibfnamefont {A.}~\bibnamefont {Kou}}, \bibinfo {author}
  {\bibfnamefont {B.}~\bibnamefont {Vlastakis}}, \bibinfo {author}
  {\bibfnamefont {A.}~\bibnamefont {Petrenko}}, \bibinfo {author}
  {\bibfnamefont {K.~M.}\ \bibnamefont {Sliwa}}, \bibinfo {author}
  {\bibfnamefont {A.}~\bibnamefont {Narla}}, \bibinfo {author} {\bibfnamefont
  {S.}~\bibnamefont {Shankar}}, \bibinfo {author} {\bibfnamefont {M.~J.}\
  \bibnamefont {Hatridge}}, \emph {et~al.},\ }\bibfield  {title} {\bibinfo
  {title} {Confining the state of light to a quantum manifold by engineered
  two-photon loss},\ }\href@noop {} {\bibfield  {journal} {\bibinfo  {journal}
  {Science}\ }\textbf {\bibinfo {volume} {347}},\ \bibinfo {pages} {853}
  (\bibinfo {year} {2015})}\BibitemShut {NoStop}%
\bibitem [{\citenamefont {Touzard}\ \emph {et~al.}(2018)\citenamefont
  {Touzard}, \citenamefont {Grimm}, \citenamefont {Leghtas}, \citenamefont
  {Mundhada}, \citenamefont {Reinhold}, \citenamefont {Axline}, \citenamefont
  {Reagor}, \citenamefont {Chou}, \citenamefont {Blumoff}, \citenamefont
  {Sliwa} \emph {et~al.}}]{touzard2018coherent}%
  \BibitemOpen
  \bibfield  {author} {\bibinfo {author} {\bibfnamefont {S.}~\bibnamefont
  {Touzard}}, \bibinfo {author} {\bibfnamefont {A.}~\bibnamefont {Grimm}},
  \bibinfo {author} {\bibfnamefont {Z.}~\bibnamefont {Leghtas}}, \bibinfo
  {author} {\bibfnamefont {S.~O.}\ \bibnamefont {Mundhada}}, \bibinfo {author}
  {\bibfnamefont {P.}~\bibnamefont {Reinhold}}, \bibinfo {author}
  {\bibfnamefont {C.}~\bibnamefont {Axline}}, \bibinfo {author} {\bibfnamefont
  {M.}~\bibnamefont {Reagor}}, \bibinfo {author} {\bibfnamefont
  {K.}~\bibnamefont {Chou}}, \bibinfo {author} {\bibfnamefont {J.}~\bibnamefont
  {Blumoff}}, \bibinfo {author} {\bibfnamefont {K.~M.}\ \bibnamefont {Sliwa}},
  \emph {et~al.},\ }\bibfield  {title} {\bibinfo {title} {Coherent oscillations
  inside a quantum manifold stabilized by dissipation},\ }\href@noop {}
  {\bibfield  {journal} {\bibinfo  {journal} {Physical Review X}\ }\textbf
  {\bibinfo {volume} {8}},\ \bibinfo {pages} {021005} (\bibinfo {year}
  {2018})}\BibitemShut {NoStop}%
\bibitem [{\citenamefont {Brassard}\ \emph {et~al.}(2000)\citenamefont
  {Brassard}, \citenamefont {Hoyer}, \citenamefont {Mosca},\ and\ \citenamefont
  {Tapp}}]{brassard2000quantum}%
  \BibitemOpen
  \bibfield  {author} {\bibinfo {author} {\bibfnamefont {G.}~\bibnamefont
  {Brassard}}, \bibinfo {author} {\bibfnamefont {P.}~\bibnamefont {Hoyer}},
  \bibinfo {author} {\bibfnamefont {M.}~\bibnamefont {Mosca}},\ and\ \bibinfo
  {author} {\bibfnamefont {A.}~\bibnamefont {Tapp}},\ }\bibfield  {title}
  {\bibinfo {title} {Quantum amplitude amplification and estimation},\
  }\href@noop {} {\bibfield  {journal} {\bibinfo  {journal} {arXiv preprint
  quant-ph/0005055}\ } (\bibinfo {year} {2000})}\BibitemShut {NoStop}%
\bibitem [{\citenamefont {Landi}\ and\ \citenamefont
  {Paternostro}(2021)}]{landi2021irreversible}%
  \BibitemOpen
  \bibfield  {author} {\bibinfo {author} {\bibfnamefont {G.~T.}\ \bibnamefont
  {Landi}}\ and\ \bibinfo {author} {\bibfnamefont {M.}~\bibnamefont
  {Paternostro}},\ }\bibfield  {title} {\bibinfo {title} {Irreversible entropy
  production: From classical to quantum},\ }\href@noop {} {\bibfield  {journal}
  {\bibinfo  {journal} {Reviews of Modern Physics}\ }\textbf {\bibinfo {volume}
  {93}},\ \bibinfo {pages} {035008} (\bibinfo {year} {2021})}\BibitemShut
  {NoStop}%
\bibitem [{\citenamefont {Peres}(1996)}]{peres1996separability}%
  \BibitemOpen
  \bibfield  {author} {\bibinfo {author} {\bibfnamefont {A.}~\bibnamefont
  {Peres}},\ }\bibfield  {title} {\bibinfo {title} {Separability criterion for
  density matrices},\ }\href@noop {} {\bibfield  {journal} {\bibinfo  {journal}
  {Physical Review Letters}\ }\textbf {\bibinfo {volume} {77}},\ \bibinfo
  {pages} {1413} (\bibinfo {year} {1996})}\BibitemShut {NoStop}%
\bibitem [{\citenamefont {Chakraborty}\ \emph {et~al.}(2025)\citenamefont
  {Chakraborty}, \citenamefont {Hazra}, \citenamefont {Li}, \citenamefont
  {Shao}, \citenamefont {Wang},\ and\ \citenamefont
  {Zhang}}]{chakraborty2025quantum}%
  \BibitemOpen
  \bibfield  {author} {\bibinfo {author} {\bibfnamefont {S.}~\bibnamefont
  {Chakraborty}}, \bibinfo {author} {\bibfnamefont {S.}~\bibnamefont {Hazra}},
  \bibinfo {author} {\bibfnamefont {T.}~\bibnamefont {Li}}, \bibinfo {author}
  {\bibfnamefont {C.}~\bibnamefont {Shao}}, \bibinfo {author} {\bibfnamefont
  {X.}~\bibnamefont {Wang}},\ and\ \bibinfo {author} {\bibfnamefont
  {Y.}~\bibnamefont {Zhang}},\ }\bibfield  {title} {\bibinfo {title} {Quantum
  singular value transformation without block encodings: Near-optimal
  complexity with minimal ancilla},\ }\href@noop {} {\bibfield  {journal}
  {\bibinfo  {journal} {arXiv preprint arXiv:2504.02385}\ } (\bibinfo {year}
  {2025})}\BibitemShut {NoStop}%
\bibitem [{\citenamefont {Wang}\ \emph {et~al.}(2025)\citenamefont {Wang},
  \citenamefont {Zhang}, \citenamefont {Hazra}, \citenamefont {Li},
  \citenamefont {Shao},\ and\ \citenamefont
  {Chakraborty}}]{wang2025randomized}%
  \BibitemOpen
  \bibfield  {author} {\bibinfo {author} {\bibfnamefont {X.}~\bibnamefont
  {Wang}}, \bibinfo {author} {\bibfnamefont {Y.}~\bibnamefont {Zhang}},
  \bibinfo {author} {\bibfnamefont {S.}~\bibnamefont {Hazra}}, \bibinfo
  {author} {\bibfnamefont {T.}~\bibnamefont {Li}}, \bibinfo {author}
  {\bibfnamefont {C.}~\bibnamefont {Shao}},\ and\ \bibinfo {author}
  {\bibfnamefont {S.}~\bibnamefont {Chakraborty}},\ }\bibfield  {title}
  {\bibinfo {title} {Randomized quantum singular value transformation},\
  }\href@noop {} {\bibfield  {journal} {\bibinfo  {journal} {arXiv preprint
  arXiv:2510.06851}\ } (\bibinfo {year} {2025})}\BibitemShut {NoStop}%
\bibitem [{\citenamefont {Vershynin}(2018)}]{vershynin2018high}%
  \BibitemOpen
  \bibfield  {author} {\bibinfo {author} {\bibfnamefont {R.}~\bibnamefont
  {Vershynin}},\ }\href@noop {} {\emph {\bibinfo {title} {High-dimensional
  probability: An introduction with applications in data science}}},\
  Vol.~\bibinfo {volume} {47}\ (\bibinfo  {publisher} {Cambridge university
  press},\ \bibinfo {year} {2018})\BibitemShut {NoStop}%
\bibitem [{\citenamefont {Shao}(2008)}]{shao2008mathematical}%
  \BibitemOpen
  \bibfield  {author} {\bibinfo {author} {\bibfnamefont {J.}~\bibnamefont
  {Shao}},\ }\href@noop {} {\emph {\bibinfo {title} {Mathematical
  statistics}}}\ (\bibinfo  {publisher} {Springer Science \& Business Media},\
  \bibinfo {year} {2008})\BibitemShut {NoStop}%
\end{thebibliography}%

\onecolumngrid

\appendix
\newpage
\onecolumngrid

\section{Estimation of the ratio of expectations}\label{apdx:est_ratio_samplenum}

\subsection{Known upper bound of variance}
In this appendix, we clarify several properties of estimation schemes for target values which can be written in the form of the ratio of two expectations.
In particular, for evaluating the total number of samples required to achieve some target precision, we can employ the Bernstein's inequality~\cite{vershynin2018high}: let $X_1,X_2,...,X_N$ be independent and identically distributed (iid) random variables, with mean $\mu_{X}=\mathbb{E}[X_i]$ and upper bound of variance $\tilde{\sigma}^2_{X}\geq {\rm Var}[X_i]$, satisfying $|X_i|\leq c~(i=1,2,...,N)$ for some positive value $c$.
Then, for every $\epsilon> 0$, we have
\begin{equation}\label{eq:Bernstein}
        {\rm Pr}\left(\left|\overline{X}_N-\mu_{X}\right|\geq \epsilon\right)\leq 2\exp\left(-\frac{N\epsilon^2}{2\tilde{\sigma}^2_X  +4c\epsilon/3}\right),
\end{equation}
where we define $\overline{X}_N:=(1/N)\sum_{i=1}^N X_i$.
Note that the Bernstein's inequality takes care of the variance of random variable unlike the Hoeffding's inequality, which is often used in the analysis for quantum algorithms.

Now, we prove the following lemma, which can be applied to the estimation schemes considered in this work.
\begin{lemma}\label{lem:ratio_confidenceint}
    Let $(X,Y)$ be random variables with means $\mu_{X}$, $\mu_{Y}~(\mu_Y\neq 0)$ and upper bounds of variances $\tilde{\sigma}^2_X\geq {\rm Var}[X]$, $\tilde{\sigma}^2_Y\geq {\rm Var}[Y]$, respectively.
    Also, we assume that $|X|$ and $|Y|$ are upper bounded by some positive value $c$, \purple{and $|\mu_X/\mu_Y|\leq c'$.} 
    Then, for any additive error \purple{$\epsilon\in (0,2c']$} and failure probability $\delta>0$, 
    \begin{equation}\label{eq:samples_varknown}
        N=32\ln \left(\frac{4}{\delta}\right)\times 
        \max\left\{\frac{\tilde{\sigma}_X^2/|\mu_Y|^2}{\epsilon^2 }+\frac{c/(6|\mu_Y|)}{\epsilon},
        \frac{\tilde{\sigma}_Y^2/|\mu_Y|^2}{\epsilon^2}\purple{(c')}^2+\frac{c/(6|\mu_Y|)}{\epsilon}\purple{c'}\right\}
    \end{equation}
    independent and identically distributed samples from $(X,Y)$ suffice to obtain an $\epsilon$-close estimate for $\mu_{X}/\mu_{Y}$ with at least $1-\delta$ probability.
\end{lemma}

\begin{proof}
    For $N$ independent and identically distributed samples $\{(X_i,Y_i)\}$, we define
    \begin{equation}
        \overline{X}_N:=\frac{1}{N}\sum_{i=1}^N X_i,~~~\overline{Y}_N:=\frac{1}{N}\sum_{i=1}^N Y_i.
    \end{equation}
    From the assumption, the Bernstein's inequality holds for any error $\epsilon_X>0$:
    \begin{equation}
        {\rm Pr}\left(\left|\overline{X}_N-\mu_{X}\right|\geq \epsilon_X\right)\leq 2\exp\left(-\frac{N\epsilon_X^2}{2\tilde{\sigma}^2_X  +4c\epsilon_X/3}\right).
    \end{equation}
    The same inequality holds for $\overline{Y}_N$ and $\epsilon_Y>0$.
    To estimate the target value $\mu_{X}/\mu_Y$, we here consider an estimator $\overline{X}_N/\overline{Y}_N$.
    If a realized value $(\overline{X}_N,\overline{Y}_N)$ satisfies
    \begin{equation}
        \left|\overline{X}_N-\mu_{X}\right|< \epsilon_X~~~\mbox{and}~~~\left|\overline{Y}_N-\mu_{Y}\right|< \epsilon_Y,~~~\epsilon_Y\in (0,|\mu_Y|/2],
    \end{equation}
    then we have
    \begin{equation}
    \left|\frac{\overline{X}_N}{\overline{Y}_N}-\frac{\mu_X}{\mu_Y}\right|\leq 2\frac{\left|\mu_Y \overline{X}_N-\mu_X\overline{Y}_N\right|}{\left|\mu_Y\right|^2}
    \leq \purple{2\frac{|\mu_Y|\epsilon_X+\epsilon_Y|\mu_X|}{|\mu_Y|^2}\leq \frac{2}{|\mu_Y|}\left(\epsilon_X+c'\epsilon_Y\right)}
    \end{equation}
    where we used the assumption $\epsilon_Y\in (0,|\mu_Y|/2]$.
    Therefore, the Bernstein's inequality and the union bound yield
    \begin{equation}
        {\rm Pr}\left(\left|\frac{\overline{X}_N}{\overline{Y}_N}-\frac{\mu_X}{\mu_Y}\right|\leq \purple{\frac{2}{|\mu_Y|}\left(\epsilon_X+c'\epsilon_Y\right)}\right)\geq 1-2\exp\left(-\frac{N\epsilon_X^2}{2\tilde{\sigma}^2_X  +4c\epsilon_X/3}\right)-2\exp\left(-\frac{N\epsilon_Y^2}{2\tilde{\sigma}^2_Y  +4c\epsilon_Y/3}\right).
    \end{equation}
    This means that, for a given target precision $\epsilon$ and an error probability $\delta$, taking $\epsilon_X,\epsilon_Y$, and $N$ as
    \begin{equation}
        \epsilon_X = \frac{\epsilon |\mu_Y|}{4},~~~\purple{\epsilon_Y = \frac{\epsilon |\mu_Y|}{4}\frac{1}{c'}\in (0,|\mu_Y|/2]},
    \end{equation}
    \begin{align}
        N&= \ln \left(\frac{4}{\delta}\right)\times \max\left\{\frac{2\tilde{\sigma}_X^2+4c\epsilon_X/3}{\epsilon_X^2},\frac{2\tilde{\sigma}_Y^2+4c\epsilon_Y/3}{\epsilon_Y^2}\right\}\notag\\[6pt]
        &= 32\ln \left(\frac{4}{\delta}\right)\times 
        \max\left\{\frac{\tilde{\sigma}_X^2/|\mu_Y|^2}{\epsilon^2 }+\frac{c/(6|\mu_Y|)}{\epsilon},
        \frac{\tilde{\sigma}_Y^2/|\mu_Y|^2}{\epsilon^2}\purple{(c')}^2+\frac{c/(6|\mu_Y|)}{\epsilon}\purple{c'}\right\},
    \end{align}
    we can obtain an $\epsilon$-close estimate of the target value with probability at least $1-\delta$:
    \begin{equation}\label{eq:confidence_int}
        {\rm Pr}\left(\left|\frac{\overline{X}_N}{\overline{Y}_N}-\frac{\mu_X}{\mu_Y}\right|\leq \epsilon\right)\geq 1-\delta.
    \end{equation}
\end{proof}

\subsection{Unknown variance (asymptotic theory)}\label{apdx:est_ratio_samplenum_B}

In Lemma~\ref{lem:ratio_confidenceint}, we may not know the (upper bound of) variance and the expectation $\mu_Y$ in advance, which appear in Eq.~\eqref{eq:samples_varknown} to estimate the number of samples $N$ required for Eq.~\eqref{eq:confidence_int}.
However, in the asymptotic regime on $N$, we can derive a confidence interval similar to Eq.~\eqref{eq:confidence_int} using only known quantities.
To this end, we here review two types of convergence of random variables~\cite{shao2008mathematical}.
\begin{itemize}
    \item \textit{Convergence in probability}: for real-valued random variables $\{X_N\}$ and $X$, we say that $X_N$ converges to $X$ in \textit{probability}, denoted by $X_N\xrightarrow{p} X$, if and only if for any $\epsilon>0$, the sequence $\{X_N\}$ satisfies 
    \begin{equation}
    \mathrm{Pr}\left(\left|X_N-X\right|>\epsilon\right)\to 0~~\mbox{as}~~N\to \infty.
    \end{equation}

    \item \textit{Convergence in distribution}: Let $\{X_N\}$ and $X$ be real-valued random variables, and let $F_N(x)$ and $F(x)$ be the corresponding distribution function (i.e., $F(x):=\mathrm{Pr}(X\leq x)$).
    We say that $X_N$ converges to $X$ in \textit{distribution}, denoted by $X_N\xrightarrow{d} X$, if and only if for any continuity point $x$ of $F(x)$, the sequence $\{F_N(x)\}$ satisfyies
    \begin{equation}
        F_N(x)\to F(x)~~\mbox{as}~~N\to \infty.
    \end{equation}
\end{itemize}
In addition, we provide two important lemmas for the proof of Lemma~\ref{lem:asymnormal_ratio}. The proof of these lemmas can be found in a textbook for statistical estimation theory e.g.,~\cite{shao2008mathematical}.
\begin{lemma}
    [Slutsky's lemma~\cite{shao2008mathematical}]
    Let $\{X_N\}$, $\{Y_N\}$, $X$ be random variables such that $X_N\xrightarrow{d} X$ and $Y_N\xrightarrow{p} c$ for some constant $c\in \mathbb{R}$. 
    Then, $X_NY_N$ converges to $cX$ in distribution, i.e., $X_NY_N\xrightarrow{d} cX$.
\end{lemma}
\begin{lemma}
    [Delta method~\cite{shao2008mathematical}]
    Let $\{\bm{X}_N\}$ and $\bm{Y}$ be $d$-dimensional random variables satisfying
    \begin{equation}
        a_N (\bm{X}_N-\bm{c})\xrightarrow{d} \bm{Y},
    \end{equation}
    where $\bm{c}\in \mathbb{R}^d$ and $\{a_N\}$ is a sequence of positive numbers with $a_N\to \infty$ (as $N\to \infty$).
    If $\bm{Y}$ follows a $d$-dimensional centered normal distribution $\mathcal{N}_d(0,\Sigma)$ with a covariance matrix $\Sigma$, then
    \begin{equation}
        a_N\left[g(\bm{X}_N)-g(\bm{c})\right]\xrightarrow{d} \mathcal{N}_1\left(0,\left[\nabla g(\bm{c})\right]^T\Sigma \left[\nabla g(\bm{c})\right]\right)
    \end{equation}
    for any function $g:\mathbb{R}^d\to \mathbb{R}$ differentiable at $\bm{c}$.
    Here, $\nabla g$ is the gradient of $g$.
\end{lemma}

Combining the Delta method with the central limit theorem, we can show that the ratio estimator $\overline{X}_N/\overline{Y}_N$ (with an appropriate rescaling factor) asymptotically follows the standard normal distribution. For simplicity, we assume $X$ and $Y$ are independent.
\begin{lemma}\label{lem:asymnormal_ratio}
    Let $(X,Y)$ be independent random variables with means $\mu_X,\mu_Y(\neq 0)$ and standard deviations $\sigma_X:=\sqrt{{\rm Var}[X]},\sigma_Y:=\sqrt{{\rm Var}[Y]}$, respectively.
    For $N$ independent samples $\{(X_i,Y_i)\}_{i=1}^N$ of the random variable $(X,Y)$, we write sample means and sample variances of $X,Y$ as $\overline{X}_N,\overline{Y}_N, \overline{\sigma}^2_{X,N},\overline{\sigma}^2_{Y,N}$, respectively.
    Then,
    \begin{equation}\label{eq:dconv_ratio}
        \frac{\sqrt{N}}{\hat{\sigma}_{\rm ratio}}\left(\frac{\overline{X}_N}{\overline{Y}_N}-\frac{\mu_X}{\mu_Y}\right)\xrightarrow{d} \mathcal{N}(0,1)~~\mbox{as}~~N\to \infty
    \end{equation}
    holds, where $\mathcal{N}(0,1)$ is the standard normal distribution. $\hat{\sigma}_{\rm ratio}$ is defined as
    \begin{equation}\label{eq:hatvarratio}
        \hat{\sigma}_{\rm ratio}:= \sqrt{\frac{\overline{\sigma}^2_{X,N}}{(\overline{Y}_N)^2} +\frac{(\overline{X}_N \overline{\sigma}_{Y,N})^2}{(\overline{Y}_N)^4}}.
    \end{equation}
    Furthermore, $\hat{\sigma}_{\rm ratio}^2\xrightarrow{p}\sigma_{\rm ratio}^2$ as $N\to \infty$, and $\sigma_{\rm ratio}^2$ is defined as
    \begin{equation}
        \sigma^2_{\rm ratio} := \frac{\sigma_X^2}{\mu_Y^2} +\frac{\mu_X^2}{\mu_Y^4}\sigma_Y^2.
    \end{equation}
\end{lemma}
\noindent
We here make some remarks on Lemma~\ref{lem:asymnormal_ratio}.
The convergence of Eq.~\eqref{eq:dconv_ratio} means that the estimator $\frac{\overline{X}_N}{\overline{Y}_N}$ asymptotically follows the normal distribution
    \begin{equation}
    \frac{\overline{X}_N}{\overline{Y}_N}\sim \mathcal{N}\left(\frac{\mu_X}{\mu_Y},\frac{\hat{\sigma}^2_{\rm ratio}}{N}\right),~ ~ ~ ~\hat{\sigma}^2_{\rm ratio} = \frac{\overline{\sigma}^2_{X,N}}{(\overline{Y}_N)^2} +\frac{(\overline{X}_N \overline{\sigma}_{Y,N})^2}{(\overline{Y}_N)^4},
    \end{equation}
when $N$ is large.
Also, from the definition of $\xrightarrow{p}$, the convergence $\hat{\sigma}_{\rm ratio}^2\xrightarrow{p}\sigma_{\rm ratio}^2$ justifies the following approximation
\begin{equation}\label{eq:hat_ratio_approx}
    \hat{\sigma}^2_{\rm ratio}\approx \sigma^2_{\rm ratio} = \frac{\sigma_X^2}{\mu_Y^2} +\frac{\mu_X^2}{\mu_Y^4}\sigma_Y^2
\end{equation}
as $N$ increases.
    
From Eq.~\eqref{eq:dconv_ratio}, we can also show 
\begin{equation}
\mathrm{Pr}\left(\left|\frac{\sqrt{N}}{\hat{\sigma}_{\rm ratio}}\left(\frac{\overline{X}_N}{\overline{Y}_N}-\frac{\mu_X}{\mu_Y}\right)\right|\leq \epsilon\right) \to 1-\frac{2}{\sqrt{2\pi}}\int_{\epsilon}^{\infty} e^{-t^2/2} dt
\end{equation}
for any $\epsilon>0$.
Thus, taking $z_{\delta/2}>0$ satisfying 
\begin{equation}
    \frac{1}{\sqrt{2\pi}}\int_{z_{\delta/2}}^{\infty} e^{-t^2/2} dt = \delta/2,~~~z^2_{\delta/2}\leq 2\ln (1/\delta),
\end{equation}
for $\delta\in (0,1)$ and setting $\epsilon$ as $z_{\delta/2}$, we conclude that
\begin{equation}
\mathrm{Pr}\left(\left|\frac{\overline{X}_N}{\overline{Y}_N}-\frac{\mu_X}{\mu_Y}\right|\leq \sqrt{\frac{\hat{\sigma}^2_{\rm ratio}z^2_{\delta/2}}{{N}}}\right) \to 1-\delta.
\end{equation}
We remark that the factor $\hat{\sigma}_{\rm ratio}^2$ can be calculated from the obtained samples $\{X_i,Y_i\}$ as in Eq.~\eqref{eq:hatvarratio}.

\begin{proof}
    [Proof of Lemma~\ref{lem:asymnormal_ratio}]
    Let $\boldsymbol{Z}_i:=(X_i,Y_i)$ $(i=1,2,...,N)$ and let $\overline{\boldsymbol{Z}}_N$ be
    \begin{equation}
    \overline{\boldsymbol{Z}}_N := \frac{1}{N}\sum_{i=1}^N \boldsymbol{Z}_i = (\overline{X}_N,\overline{Y}_N).
    \end{equation}
    From the (multidimensional) central limit theorem, we have
    \begin{equation}
    \sqrt{N}(\overline{\boldsymbol{Z}}_N-\mathbb{E}[\boldsymbol{Z}_1]) \xrightarrow{d} \mathcal{N}_2(0,\Sigma),
    \end{equation}
    where $\mathcal{N}_2(0,\Sigma)$ denotes the two-dimensional normal distribution with zero-mean and covariance matrix $\Sigma$.
    The covariance matrix can be written as 
    \begin{equation}
    \Sigma:=
    \begin{pmatrix}
    \sigma^2_X&0\\
    0&\sigma^2_Y
    \end{pmatrix},
    \end{equation}
    due to the independence of $X$ and $Y$.
    Applying the Delta method for $g(x,y):=x/y$ and $\overline{\bm{Z}}_N$, we obtain
    \begin{equation}
    \sqrt{N}\left(\frac{\overline{X}_N}{\overline{Y}_N}-\frac{\mu_X}{\mu_Y}\right) \xrightarrow{d} \mathcal{N}(0,\sigma^2_{\rm ratio}),
    \end{equation}
    where we used the assumption that $\mu_Y\neq 0$, and $\sigma_{\rm ratio}^2$ is defined as
    \begin{equation}
    \sigma_{\rm ratio}^2 = \frac{\sigma_X^2}{\mu_Y^2} - 2\frac{\mu_X}{\mu_Y^3}{\rm Cov}(X,Y)+\frac{\mu_X^2}{\mu_Y^4}\sigma_Y^2 = \frac{\sigma_X^2}{\mu_Y^2} +\frac{\mu_X^2}{\mu_Y^4}\sigma_Y^2.
    \end{equation}
    The second equality holds due to the independence of $X$ and $Y$.
    To estimate $\sigma^2_{\rm ratio}$, we use
    $\hat{\sigma}_{\rm ratio}$ defined in Eq.~\eqref{eq:hatvarratio}, and we can show that $\hat{\sigma}_{\rm ratio}^2\xrightarrow{p} {\sigma}_{\rm ratio}^2$ and ${\sigma}_{\rm ratio}/\hat{\sigma}_{\rm ratio}\xrightarrow{p} 1$ hold from the law of large numbers and the continuous mapping theorem~\cite{shao2008mathematical}.
    Therefore, using the Slutsky's lemma, we conclude that
    \begin{equation}
    \frac{\sqrt{N}}{\hat{\sigma}_{\rm ratio}}\left(\frac{\overline{X}_N}{\overline{Y}_N}-\frac{\mu_X}{\mu_Y}\right)=\sqrt{N}\frac{{\sigma}_{\rm ratio}}{\hat{\sigma}_{\rm ratio}}\frac{\frac{\overline{X}_N}{\overline{Y}_N}-\frac{\mu_X}{\mu_Y}}{{\sigma}_{\rm ratio}}\xrightarrow{d} \mathcal{N}(0,1),
    \end{equation}
\end{proof}

\section{Proof of Theorem~\ref{thm:1}}\label{apdx:prf_thm1}

\begin{proof}[Proof of Theorem~\ref{thm:1}]
    We first show that for any operator $A$ on S, the resulting operator after $\Gamma$ can be written as 
    \begin{align}\label{apdx:prf_thm1_temp1}
        \Gamma[\Pi_{\rm A}\otimes \ket{+}\bra{+}_{\rm B}\otimes  A] &:=\sum_{k,k'=1}^{G} q_kq_{k'} L^{\rm (c)}_{kk'}\left(\Pi_{\rm A}\otimes \ket{+}\bra{+}_{\rm B}\otimes  A\right) (L^{\rm (c)}_{kk'})^\dagger\notag \\
        &=\Pi_{\rm A}\otimes \left[\frac{I_{\rm B}}{2}\otimes \left(\sum_k q_k K_k AK_k^\dagger\right)+\frac{X_{\rm B}}{2}\otimes \tilde{\Lambda}_{p,\mathcal{U}}(A)\right]+\sum_{\bm{i},\bm{j}\neq (\bm{0},\bm{0})}\ket{\bm{i}}\bra{\bm{j}}_{\rm A}\otimes \tilde{\sigma}_{{\rm BS},\bm{ij}},
    \end{align}
    where $\ket{\bm{i}}$ denotes the computational basis on A and 
    $\Pi_{\rm A}=\ket{\bm{0}}\bra{\bm{0}}$.
    $K_k$ is defined as 
    $K_k:=\sum_{i\in S_k}\frac{p_i}{q_k} U_i$, and $\tilde{\sigma}_{{\rm BS},\bm{ij}}$ is some operator depending on $\bm{i,j}$ and the input operator $A$.
    We calculate the action of ${L}^{(\rm c)}_{kk'}$ to $\ket{+}\bra{+}_{\rm B}\otimes \Pi_{\rm A}\otimes A$ as follows:
    \begin{align}
        &{L}^{(\rm c)}_{kk'}\left(\ket{+}\bra{+}\otimes \Pi_{\rm A}\otimes A\right) \left({L}^{(\rm c)}_{kk'}\right)^\dagger \notag\\
        &= {L}^{\rm (c)}_{kk'}\cdot \frac{1}{2}\left(|0\rangle\langle 0|\otimes \Pi\otimes A +|0\rangle\langle 1|\otimes \Pi\otimes A+|1\rangle\langle 0|\otimes \Pi\otimes A +|1\rangle\langle 1|\otimes \Pi\otimes A\right)\left({L}^{\rm (c)}_{kk'}\right)^\dagger\notag\\
        &= \frac{1}{2}\left(|0\rangle\langle 0|\otimes {L}_{k'}\cdot \Pi\otimes A\cdot {L}_{k'}^\dagger +|0\rangle\langle 1|\otimes {L}_{k'}\cdot \Pi\otimes A\cdot{L}_{k}^\dagger+|1\rangle\langle 0|\otimes {L}_{k}\cdot \Pi\otimes A\cdot{L}_{{k'}}^\dagger+|1\rangle\langle 1|\otimes {L}_{{k}}\cdot \Pi\otimes A\cdot{L}_{{k}}^\dagger \right).
    \end{align}
    From the definition of $L_k$, we have
    \begin{align}
        L_{k'}\cdot \Pi_{\rm A}\otimes A\cdot {L}_{k'}^\dagger
        &=\Pi_{\rm A}\otimes \bm{1} \left\{{L}_{k'}\cdot \Pi_{\rm A}\otimes A\cdot {L}_{k'}^\dagger \right\}\Pi_{\rm A}\otimes \bm{1}
        +\left(\bm{1}-\Pi_{\rm A}\right)\otimes \bm{1} \left\{{L}_{k'}\cdot \Pi_{\rm A}\otimes A\cdot {L}_{k'}^\dagger \right\}\Pi_{\rm A}\otimes \bm{1}\notag\\[6pt]
        &~~~+ \Pi_{\rm A}\otimes \bm{1} \left\{{L}_{k'}\cdot \Pi_{\rm A}\otimes A\cdot {L}_{k'}^\dagger\right\} \left(\bm{1}-\Pi_{\rm A}\right)\otimes \bm{1}+\left(\bm{1}-\Pi_{\rm A}\right)\otimes \bm{1} \left\{{L}_{k'}\cdot \Pi_{\rm A}\otimes A\cdot {L}_{k'}^\dagger\right\} \left(\bm{1}-\Pi_{\rm A}\right)\otimes \bm{1}\notag\\[6pt]
        &=\Pi_{\rm A}\otimes K_{k'} A K_{k'}^\dagger 
        +\sum_{\bm{i,j}\neq (\bm{0},\bm{0})} \ket{\bm{i}}\bra{\bm{i}}\otimes \bm{1}\cdot 
        {L}_{k'}\cdot \Pi_{\rm A}\otimes A\cdot {L}_{k'}^\dagger \cdot \ket{\bm{j}}\bra{\bm{j}}\otimes \bm{1}
    \end{align}
    Thus, the same calculation yields 
    \begin{align}
        &\sum_{k,k'=1}^{G} q_kq_{k'}{L}^{(c)}_{kk'}\left(\ket{+}\bra{+}\otimes \Pi_{\rm A}\otimes A\right) \left({L}^{(\rm c)}_{kk'}\right)^\dagger \notag\\
        &=\sum_{k,k'=1}^{G} q_kq_{k'}\frac{1}{2}\Pi_{\rm A}\otimes \left(|0\rangle\langle 0|\otimes  {K}_{k'} A {K}_{k'}^\dagger+|0\rangle\langle 1|\otimes  {K}_{k'} A {K}_{k}^\dagger+|1\rangle\langle 0|\otimes  {K}_{k} A {K}_{k'}^\dagger+|1\rangle\langle 1|\otimes  {K}_{k} A {K}_{k}^\dagger\right)\notag\\
        &~~~+ \sum_{\bm{i},\bm{j}\neq (\bm{0},\bm{0})}\ket{\bm{i}}\bra{\bm{j}}_{\rm A}\otimes \tilde{\sigma}_{{\rm BS},\bm{ij}}\notag\\
        &=\frac{1}{2}\Pi_{\rm A}\otimes \left(|0\rangle\langle 0|\otimes \left\{\sum_{{k'}=1}^{G} q_{k'} {K}_{{k'}} A {K}_{{k'}}^\dagger\right\}
        +|0\rangle\langle 1|\otimes  \tilde{\Lambda}_{p,\mathcal{U}}(A)
        +|1\rangle\langle 0|\otimes  \tilde{\Lambda}_{p,\mathcal{U}}(A)
        +|1\rangle\langle 1|\otimes \left\{\sum_{k=1}^{G} q_k {K}_{k} A {K}_{k}^\dagger\right\}\right)\notag\\
        &~~~+ \sum_{\bm{i},\bm{j}\neq (\bm{0},\bm{0})}\ket{\bm{i}}\bra{\bm{j}}_{\rm A}\otimes \tilde{\sigma}_{{\rm BS},\bm{ij}}.
    \end{align}
    In third line, we used the fact that $\sum_k q_k=\sum_{k=1}^{G}\sum_{i\in S_k}p_i = 1$ and $\sum_k q_k K_k=\sum_i p_i U_i$.
    Also, in the second line, we defined $\tilde{\sigma}_{{\rm BS},\bm{ij}}$ as
    \begin{align}
        \tilde{\sigma}_{{\rm BS},\bm{ij}}=\sum_{k,k'=1}^{G} \frac{q_kq_{k'}}{2}&[\left(I_{\rm B}\otimes \bra{\bm{i}}_{\rm A}\otimes \bm{1}\right) |0\rangle\langle 0|\otimes {L}_{k'}\cdot \Pi_{\rm A}\otimes A\cdot {L}_{k'}^\dagger
        \left(I_{\rm B}\otimes \ket{\bm{j}}_{\rm A}\otimes \bm{1}\right)\notag\\
        &~~~+\left(I_{\rm B}\otimes \bra{\bm{i}}_{\rm A}\otimes \bm{1}\right) |0\rangle\langle 1|\otimes {L}_{k'}\cdot \Pi_{\rm A}\otimes A\cdot {L}_{k}^\dagger
        \left(I_{\rm B}\otimes \ket{\bm{j}}_{\rm A}\otimes \bm{1}\right)\notag\\
        &~~~+\left(I_{\rm B}\otimes \bra{\bm{i}}_{\rm A}\otimes \bm{1}\right) |1\rangle\langle 0|\otimes {L}_{k}\cdot \Pi_{\rm A}\otimes A\cdot {L}_{k'}^\dagger
        \left(I_{\rm B}\otimes \ket{\bm{j}}_{\rm A}\otimes \bm{1}\right)\notag\\
        &~~~+\left(I_{\rm B}\otimes \bra{\bm{i}}_{\rm A}\otimes \bm{1}\right) |1\rangle\langle 1|\otimes {L}_{k}\cdot \Pi_{\rm A}\otimes A\cdot {L}_{k}^\dagger
        \left(I_{\rm B}\otimes \ket{\bm{j}}_{\rm A}\otimes \bm{1}\right)].
    \end{align}
    This completes the proof of Eq.~\eqref{apdx:prf_thm1_temp1}.

    Multiplying the observable $\Pi_{\rm A}\otimes X_{\rm B}\otimes\bm{1}$ to the operator Eq.~\eqref{apdx:prf_thm1_temp1} followed by the partial trace over the ancilla systems A and B, we obtain
    \begin{align}
        &{\rm tr}_{{\rm AB}}\left[\left(\Pi_{\rm A}\otimes X_{\rm B}\otimes \bm{1}\right) \cdot \Gamma[\Pi_{\rm A}\otimes \ket{+}\bra{+}_{\rm B}\otimes  A] \right]\notag\\
        &={\rm tr}_{{\rm AB}}\left[\left(\Pi_{\rm A}\otimes X_{\rm B}\otimes \bm{1}\right) \cdot\left(\Pi_{\rm A}\otimes \left[\frac{I_{\rm B}}{2}\otimes \left(\sum_k q_k K_k AK_k^\dagger\right)+\frac{X_{\rm B}}{2}\otimes \tilde{\Lambda}_{p,\mathcal{U}}(A)\right]\right)\right]\notag\\
        &=\tilde{\Lambda}_{p,\mathcal{U}}(A).
    \end{align}
    Note that the second equality holds because ${\rm tr}_{\rm AB}[\Pi_{\rm A}\sum_{\bm{i},\bm{j}\neq (\bm{0},\bm{0})}\ket{\bm{i}}\bra{\bm{j}}_{\rm A}\otimes \tilde{\sigma}_{{\rm BS},\bm{ij}}]$ is zero.
    Since we can take the operator $A$ arbitrarily, we establish Theorem~\ref{thm:1}.
\end{proof}

We here mention two special cases of Theorem~\ref{thm:1}.
First, we consider the following partition
$$
S_1:=[m]=\{1,2,...,m\},
$$
and this partition yields
\begin{align}
    \tilde{\Lambda}_{p,\mathcal{U}}(\rho)&={\rm tr}_{{\rm AB}}\left[\left( \Pi_{\rm A}\otimes X_{\rm B}\otimes \bm{1}\right) \cdot 
    q_1q_{1} L^{\rm (c)}_{11}\left(\Pi_{\rm A}\otimes \ket{+}\bra{+}_{\rm B}\otimes \rho\right) (L^{\rm (c)}_{11})^\dagger
    \right]\notag\\
    &={\rm tr}_{{\rm A}}\left[\left( \Pi_{\rm A}\otimes \bm{1}\right) \cdot 
    L_1\left(\Pi_{\rm A}\otimes \rho\right)L_1^\dagger \right],
\end{align}
which corresponds to the original LCU method in Eq.~\eqref{eq:physical_pic_lcumap}.
Here, we note that $q_1 = \sum_{i\in S_1=[m]} p_i=1$.
Next, in the case of 
$$S_k:=\{k\}~~\mbox{for~all}~~k=1,2,...,m,$$ 
we can take $L_k$ as $\bm{1}_{\rm A}\otimes U_k$ (which acts on only S) without loss of generality.
Thus, we have
\begin{align}
    \tilde{\Lambda}_{p,\mathcal{U}}(\rho)& = {\rm tr}_{{\rm AB}}\left[\left( \Pi_{\rm A}\otimes X_{\rm B}\otimes \bm{1}\right) \cdot \sum_{k,k'=1}^{m} p_kp_{k'}L^{(\rm c)}_{kk'}[\Pi_{\rm A}\otimes \ket{+}\bra{+}_{\rm B}\otimes \rho](L^{(\rm c)}_{kk'})^\dagger\right]\notag\\
    &=\sum_{k,k'=1}^{m} p_kp_{k'} {\rm tr}_{{\rm B}}\left[\left( X_{\rm B}\otimes \bm{1}\right) \cdot L^{(\rm c)}_{kk'}[\ket{+}\bra{+}_{\rm B}\otimes \rho](L^{(\rm c)}_{kk'})^\dagger\right].
\end{align}
If we take the expectation value on an observable $O$ on S, then this choice of the partition recovers the full virtual LCU method:
\begin{align}
    {\rm tr}[O\tilde{\Lambda}_{p,\mathcal{U}}(\rho)]&=
    \sum_{k,k'=1}^{m} p_kp_{k'} {\rm tr}\left[\left( X_{\rm B}\otimes O\right) \cdot L^{(\rm c)}_{kk'}[\ket{+}\bra{+}_{\rm B}\otimes \rho](L^{(\rm c)}_{kk'})^\dagger\right]=\sum_{k,k'=1}^{m} p_kp_{k'} {\rm Re}\left({\rm tr}\left[OU_k\rho U_{k'}^\dagger\right]\right).
\end{align}

\section{Proof of Theorem~\ref{thm:2}}\label{apdx:prf_thm2}
In the following, we provide the proof of a more general version of Theorem~\ref{thm:2}.
We recall that the estimation scheme provided in Section~\ref{sec:main_exp_val_est} gives the estimator ${g}_O:=(-1)^b\chi_0(z)o_j$ for ${\rm tr}[O\tilde{\Lambda}(\rho)]$, and its variance can be written as
\begin{equation}
    {\rm Var}[{g}_O]=R^{O}\left[\{S_k\};\rho\right]-{\rm tr}[O\tilde{\Lambda}(\rho)]^2,
\end{equation}
where $R^{O}$ is the second moment of $g_O$ defined as
\begin{equation}
    R^{O}\left[\{S_k\};\rho\right]:=\mathbb{E}[g^2_O]=\sum_{k=1}^{G} q_k{\rm tr}
    \left[ O^2 K_k \rho K_k^\dagger
    \right].
\end{equation}
The relation between $R$ (defined by Eq.~\eqref{eq:def_R_factor}) and $R^{O}$ is given as
\begin{equation}
    R^{O}\left[\{S_k\};\rho\right]\leq R\left[\{S_k\};\rho\right]\|O\|^2.
\end{equation}
In fact, the following theorem for $R^{O}$ holds; this is a generalized version of Theorem~\ref{thm:2} for $R$.

\begin{thm}\label{apdx:general_thm2}
    Let us consider the decomposition methods for $\tilde{\Lambda}(\bullet)=(\sum_{i=1}^m p_i U_i)\bullet (\sum_{i=1}^m p_i U_i)^\dagger$ in Theorem~\ref{thm:1} and any observable $O$.
    Suppose we have two partitions $\{S_k\}_{k=1}^G$ and $\{S'_{l}\}_{l=1}^{G'}$ of $[m]:=\{1,2,...,m\}$ satisfying the condition Eq.~\eqref{eq:condition4partition} in Theorem~\ref{thm:1}.
    If for any $S'_l$ ($l=1,2,...,G'$), there exists a single set $S_k$ that contains $S_l'$, then the corresponding reduction factors $R^{O}$ satisfy the following inequalities for any input quantum state $\rho$:
    \begin{equation}\label{apdx:inequalities4R}
        {\rm tr}\left[O^2\tilde{\Lambda}[\rho]\right] \leq R^{O}\left[\{S_k\}_{k=1}^G;\rho\right] \leq R^{O}\left[\{S'_{l}\}_{l=1}^{G'};\rho\right]\leq {\rm tr}\left[O^2\cdot \sum_{i=1}^m p_i U_i\rho U^\dagger_i\right].
    \end{equation}
    The upper and lower bound can be saturated by the two extreme cases: the complete fragmentation Eq.~\eqref{eq:complete_fragment} and no partition Eq.~\eqref{eq:no_partition}, respectively.
\end{thm}
\noindent
By taking the observable $O$ satisfying $O^2=\bm{1}$, Theorem~\ref{thm:2} can be obtained from Theorem~\ref{apdx:general_thm2} immediately.

\begin{proof}
    [Proof of Theorem~\ref{apdx:general_thm2}]
    To prove the middle inequality of Eq.~\eqref{apdx:inequalities4R}, it suffices to show that $R^{O}$ for $\{S_k\}_{k=1}^G$ is less than or equal to $R^{O}$ for $\{S_k\}_{k=1}^{G-1}\cup \{S_A,S_B\}$, where $\{S_A,S_B\}$ is a partition of $S_G$.
    From the definition of $q_k$ in Eq.~\eqref{eq:weight_qk_def_Kk} and the relation of $S_A$, $S_B$, and $S_G$, we have
    \begin{equation}
        q_{G}=\sum_{i\in S_G} p_i = q_{A}+q_{B},~~q_{A}=\sum_{i\in S_A} p_i,~~q_{B}=\sum_{i\in S_B} p_i.
    \end{equation}
    \purple{Also, we define $K_G,K_A$ and $K_B$ via}
    \begin{equation}
        q_G K_{G}=\sum_{i\in S_G} {p_i} U_i,~~q_A K_{A}=\sum_{i\in S_A}{p_i}U_i,~~q_B K_{B}=\sum_{i\in S_B}{p_i} U_i.
    \end{equation}
    Using these relations, we can show that
    \begin{align*}
        &q_G\left(R^{O}\left[\{S_k\}_{k=1}^{G-1}\cup \{S_A,S_B\};\rho\right] -R^{O}\left[\{S_k\}_{k=1}^{G};\rho\right] \right)\notag\\[6pt]
        &=
        q_Gq_A{\rm tr}
        \left[ K_A^\dagger O^2 K_A \rho \right]
        +q_Gq_B{\rm tr}\left[ K_B^\dagger O^2 K_B \rho \right]
        -q^2_G{\rm tr}
        \left[ K_G^\dagger O^2 K_G \rho \right]\notag\\[6pt]
        &=
        q_Gq_A{\rm tr}
        \left[ (O K_A)^\dagger O K_A \rho \right]
        +q_Gq_B{\rm tr}\left[ (O K_B)^\dagger O K_B \rho \right]
        -q^2_G{\rm tr}
        \left[ (O K_G)^\dagger O K_G \rho \right]\notag\\[6pt]
        &= \frac{q_G}{q_A}{\rm tr}\left[\left(q_A O K_A\right)^\dagger\left(q_A O K_A\right) \rho\right]+\frac{q_G}{q_B} {\rm tr}\left[\left(q_B O K_B\right)^\dagger\left(q_B O K_B\right)\rho\right]\notag\\[6pt]
        &~~~~~-{\rm tr}\left[ \left(q_A O K_A+q_B O K_B\right)^\dagger\left(q_A O K_A+q_B O K_B\right) \rho\right]\notag\\[6pt]
        &= \frac{q_B}{q_A}{\rm tr}\left[\left(q_A O K_A\right)^\dagger\left(q_A O K_A\right) \rho\right]+\frac{q_A}{q_B} {\rm tr}\left[\left(q_B O K_B\right)^\dagger\left(q_B O K_B\right)\rho\right]\notag\\[6pt]
        &~~~~~+{\rm tr}\left[\left(q_A O K_A\right)^\dagger\left(q_A O K_A\right) \rho\right]+{\rm tr}\left[\left(q_B O K_B\right)^\dagger\left(q_B O K_B\right)\rho\right]-{\rm tr}\left[ \left(q_A O K_A+q_B O K_B\right)^\dagger\left(q_A O K_A+q_B O K_B\right) \rho\right]\notag\\[6pt]
        &= {q_Aq_B}{\rm tr}\left[(O K_A)^\dagger(O K_A) \rho+ (O K_B)^\dagger(O K_B)\rho-(O K_A)^\dagger(O K_B)\rho- (O K_B)^\dagger(O K_A) \rho\right]\notag\\[6pt]
        &= q_Aq_B{\rm tr}\left[\left(O K_A-O K_B\right)^\dagger\left(O K_A-O K_B\right) \rho\right].
    \end{align*}
    Therefore, we conclude that
    \begin{equation}\label{eq:RO_factor_difference_exact}
        R^{O}\left[\{S_k\}_{k=1}^{G-1}\cup \{S_A,S_B\};\rho\right] -R^{O}\left[\{S_k\}_{k=1}^{G};\rho\right]=
        \frac{q_A q_B}{q_{A}+q_B}{\rm tr}\left[\left(O K_A-O K_B\right)^\dagger\left(O K_A-O K_B\right) \rho\right]\geq 0.
    \end{equation}
    From the above observation, it is obvious that the following inequalities hold for any partition $\{S_k\}_{k=1}^G$ of $[m]$:
    \begin{equation}
        R^{O}\left[\{[m]\};\rho\right]\leq R^{O}\left[\{S_k\}_{k=1}^G;\rho\right]\leq R^{O}\left[\{\{i\}\}_{i=1}^m;\rho\right].
    \end{equation}
    Thus, the direct calculation of $R^{O}$ for these two partitions completes the proof of the inequalities Eq.~\eqref{apdx:inequalities4R}.
\end{proof}

\section{Further properties of the reduction factor $R$}\label{sec:apdxD}

\begin{lemma}\label{lem:R_factpr_ineq_prop}
    Suppose that the same assumption of Theorem~\ref{thm:1} and $|S_{G}|\geq 2$ hold. Then, the followings hold:
    \begin{itemize}
        \item (Split a subset) If we split $S_{G}$ into two subsets $S_{A},S_{B}$ such that $\{S_A,S_B\}$ is a partition of $S_{G}$, then for any $\rho$
        \begin{equation}
            R^{O}\left[\{S_k\}_{k=1}^G; \rho\right] \leq R^{O}\left[\{S_k\}_{k=1}^{G-1}\cup \{S_A,S_B\}; \rho\right] \leq R^{O}\left[\{S_k\}_{k=1}^G; \rho\right]+2\|O^2\|H\left(q_A,q_B\right)
        \end{equation}
        holds. Here, $H(a,b):=2ab/(a+b)$ denotes the harmonic mean, and $q_k=\sum_{i\in S_k} p_i$ for $k=A,B$.
        \item (Fully fragment a subset) If we fully fragment $S_{G}$ into $\{\{i\}:i\in S_G\}$, then for any $\rho$
        \begin{equation}
            R^{O}\left[\{S_k\}_{k=1}^G; \rho\right] \leq R^{O}\left[\{S_k\}_{k=1}^{G-1}\cup\{\{i\}:i\in S_G\}; \rho\right] \leq R^{O}\left[\{S_k\}_{k=1}^G; \rho\right]+\|O^2\|q_G,
        \end{equation}
        where $q_G :=\sum_{i\in S_G} p_i$.
    \end{itemize}
\end{lemma}

\begin{proof}
    From Eq.~\eqref{eq:RO_factor_difference_exact}, we have
    \begin{align}
        R^{O}\left[\{S_k\}_{k=1}^{G-1}\cup \{S_A,S_B\};\rho\right] -R^{O}\left[\{S_k\}_{k=1}^{G};\rho\right]&=
        2\frac{2q_A q_B}{q_{A}+q_B}{\rm tr}\left[ O^2 \left(\frac{K_A-K_B}{2}\right) \rho \left(\frac{K_A- K_B}{2}\right)^\dagger\right]\notag\\
        &\leq 2\|O^2\|H(q_A,q_B).
    \end{align}
    \purple{Here, the map $\rho\mapsto ((K_A-K_B)/2)\rho ((K_A-K_B)/2)^\dagger$ is completely positive and trace-non-increasing.}
    Also, from the definition, 
    \begin{align}
        R^{O}\left[\{S_k\}_{k=1}^{G-1}\cup\{\{i\}:i\in S_G\}; \rho\right] &= \sum_{k=1}^{G-1} q_k{\rm tr}
    \left[ O^2 K_k \rho K_k^\dagger
    \right] +\sum_{i\in S_G} p_i {\rm tr}
    \left[ O^2 U_i \rho U_i^\dagger
    \right]\notag\\
    &= R^{O}\left[\{S_k\}_{k=1}^{G};\rho\right]+\sum_{i\in S_G} p_i {\rm tr}
    \left[ O^2 U_i \rho U_i^\dagger
    \right]-q_G {\rm tr}
    \left[ O^2 K_G \rho K_G^\dagger
    \right]\notag\\
    &\leq R^{O}\left[\{S_k\}_{k=1}^{G};\rho\right]+\|O^2\|q_G,
    \end{align}
    \purple{where we used $q_G {\rm tr}\left[ O^2 K_G \rho K_G^\dagger\right]\geq 0$.}
\end{proof}
By substituting $O^2=\bm{1}$, we obtain Lemma~\ref{lem:R_factpr_ineq_prop_main}.
From the direct use of the Lemma~\ref{lem:R_factpr_ineq_prop}, we obtain the following useful result.
\begin{cor}
    Suppose that the same assumption of Theorem~\ref{thm:1}. Then, we have
    \begin{equation}
        R^{O}\left[\{S_k\}_{k=1}^{G-1}\cup \{S_A\}\cup \{\{i\}:i \in S_B\}; \rho\right]-R^{O}\left[\{S_k\}_{k=1}^{G-1}\cup \{S_G\}; \rho\right] \leq \|O^2\|\times \min\{2H(q_A,q_B)+q_B,{q_A+q_B}\}.
    \end{equation}
\end{cor}
\noindent
We note that $q_A > 2H(q_A,q_B)$ holds when $q_B$ is small.

\section{Sample complexity of the proposed LCU methods}\label{apdx:sample_complexity}

We consider the sample complexity to estimate the target expectation values: 
$$\mbox{(i)}~{\rm tr}[OK\rho K^\dagger]~~~\mbox{and}~~~\mbox{(ii)}~{\rm tr}[OK\rho K^\dagger]/{{\rm tr}[K\rho K^\dagger]}$$
for $K=\sum_{i=1}^m c_i U_i=\|c\|_1 \sum_{i=1}^m p_i U_i$, via the decomposition of Theorem~\ref{thm:1} given by 
\begin{equation}
    \tilde{\Lambda}_{p,\mathcal{U}}(\bullet):=\left(\sum_{i=1}^m p_i U_i\right)\bullet\left(\sum_{i=1}^m p_i U_i\right)^\dagger={\rm tr}_{{\rm AB}}\left[\left( \Pi_{\rm A}\otimes X_{\rm B}\otimes \bm{1}\right) \cdot \Gamma[\Pi_{\rm A}\otimes \ket{+}\bra{+}_{\rm B}\otimes (\bullet)]\right].
\end{equation}
Here, we recall that $\Gamma$ is a mixed unitary channel that is determined by a partition $\{S_k\}_{k=1}^G$ of the index set $[m]:=\{1,2,...,m\}$.
As mentioned in Section~\ref{sec:main_exp_val_est}, running the quantum circuit preparing $\Gamma[\Pi_{\rm A}\otimes \ket{+}\bra{+}_{\rm B}\otimes \rho]$ and measuring the final state by the POVM $\{\Pi_{\rm A}^{(z)}\otimes \ket{b}\bra{b}_{\rm B}\otimes \ket{j}\bra{j}\}$, we obtain a single sample from the random variable $g_{O}:=(-1)^b \chi_0(z)o_j$ satisfying $\mathbb{E}[g_{O}]={\rm tr}[O\tilde{\Lambda}_{p,\mathcal{U}}(\rho)]$.
Thus, we can define an unbiased estimator $G_{O}:=\|c\|^2_1 g_O$ for the target value ${\rm tr}[OK\rho K^\dagger]$, i.e., 
\begin{equation}
    \mathbb{E}[G_{O}]=\|c\|^2_1\mathbb{E}[g_{O}]=\|c\|_1^2{\rm tr}[O\tilde{\Lambda}_{p,\mathcal{U}}(\rho)]={\rm tr}[OK\rho K^\dagger].
\end{equation}
The variance of $G_{O}$ is upper bounded by $R$ as well as in Eq.~\eqref{eq:var_upb_general}:
\begin{equation}
    {\rm Var}[{G}_O]=\|c\|_1^4 {\rm Var}[{g}_O]\leq \|c\|_1^4 \|O\|^2 R\left[\{S_k\};\rho\right].
\end{equation}
Also, we note that $G_{O}$ is upper bounded by
\begin{equation}
    \left|G_{O}\right|=\|c\|_1^2\left|g_O\right|\leq \|c\|_1^2\times \max_j \left|o_j\right| = \|c\|_1^2\|O\|.
\end{equation}

\subsection{Expectation value estimation for unnormalized states}

First, we focus on the estimation of ${\rm tr}[OK\rho K^\dagger]$.
Measuring $N$ copies of the quantum state $\Gamma[\Pi_{\rm A}\otimes \ket{+}\bra{+}_{\rm B}\otimes \rho]$ by the observable $\Pi_{\rm A}\otimes X_{\rm B}\otimes O$ followed by the multiplication of $\|c\|_1^2$, we obtain $N$ independent samples $\{G_O^{(i)}=\|c\|_1^2 g_{O}^{(i)}\}_{i=1}^N$ from the random variable $G_O$.
Thus, using the Bernstein's inequality in Eq.~\eqref{eq:Bernstein}, we obtain
\begin{equation}
    {\rm Pr}\left(\left|\overline{G}_{O,N}-{\rm tr}[OK\rho K^\dagger]\right|\geq \epsilon\right)\leq 2\exp\left(-\frac{N\epsilon^2}{2\|c\|_1^4\|O\|^2 R\left[\{S_k\};\rho\right] +4\|c\|_1^2\|O\|\epsilon/3}\right)
\end{equation}
for any $\epsilon>0$.
Here, $\overline{G}_{O,N}$ denotes the sample mean of $\{G_O^{(i)}\}_{i=1}^N$.
To assure that we obtain an $\epsilon$-close estimate for ${\rm tr}[OK\rho K^\dagger]$ with at least $1-\delta$ probability:
\begin{equation}
    {\rm Pr}\left(\left|\overline{G}_{O,N}-{\rm tr}[OK\rho K^\dagger]\right|\leq \epsilon\right)\geq 1-\delta,
\end{equation} 
it suffices to set
\begin{align}
    N \geq 2\ln(\frac{2}{\delta}) \times \left\{{R\left[\{S_k\};\rho\right]\frac{\|c\|_1^4\|O\|^2}{\epsilon^2} +\frac{2}{3}\frac{\|c\|_1^2\|O\|}{\epsilon}}\right\}.
\end{align}
This completes the proof of the first part of Theorem~\ref{thm:2_samplecomplexity}.

When we have no knowledge of $R$, we will construct an (asymptotic) confidence interval by
simultaneously estimating the variance of $G_O$.
In this case, we can use the asymptotic theory provided in Appendix~\ref{apdx:est_ratio_samplenum_B}.
Specifically, we can show that 
\begin{equation}
    \sqrt{N}\frac{\overline{G}_{O,N}-{\rm tr}[OK\rho K^\dagger]}{\hat{\sigma}_{g_O}}\xrightarrow{d} \mathcal{N}(0,\|c\|_1^4).
\end{equation}
holds as well as the proof of Lemma~\ref{lem:asymnormal_ratio}.
Here, $\hat{\sigma}_{g_O}:=\sqrt{\hat{\sigma}_{g_O}^2}$ and $\hat{\sigma}_{g_O}^2$ is the sample variance defined as
\begin{equation}\label{eq:go_sample_var}
    \hat{\sigma}^2_{g_O} := {\frac{1}{N}\sum_{i=1}^N \left(g_O^{(i)}-\overline{g}_{O,N}\right)^2},
\end{equation}
where $\overline{g}_{O,N}=\overline{G}_{O,N}/\|c\|_1^2$.
Therefore, we obtain 
\begin{align}
    {\rm Pr}\left(\left|\overline{G}_{O,N}-{\rm tr}[OK\rho K^\dagger]\right|\leq \sqrt{\frac{\hat{\sigma}_{g_O}^2 \|c\|_1^4 z^2_{\delta/2}}{N}}\right)\to 1-\delta~~\mbox{as}~~N\to\infty.
\end{align}
This asymptotic $1-\delta$ confidence interval is tighter than the interval obtained from the $R$ because $\hat{\sigma}_{g_O}^2$ is upper bounded by the $R$ (more precisely, $R^O\left[\{S_k\};\rho\right]$) when $N$ is sufficiently large:
\begin{align}
    &\left[\overline{G}_{O,N}-\sqrt{\frac{\hat{\sigma}_{g_O}^2 \|c\|_1^4 z^2_{\delta/2}}{N}},~\overline{G}_{O,N}+\sqrt{\frac{\hat{\sigma}_{g_O}^2 \|c\|_1^4 z^2_{\delta/2}}{N}}\right]\notag\\
    &\subseteq \left[\overline{G}_{O,N}-\sqrt{\frac{\|c\|_1^4 {R}^O\left[\{S_k\};\rho\right]  z^2_{\delta/2}}{N}},~\overline{G}_{O,N}+\sqrt{\frac{ \|c\|_1^4 {R}^O\left[\{S_k\};\rho\right] z^2_{\delta/2}}{N}}\right].
\end{align}
Also, we remark that the $R^O$ can be estimated by $\hat{\sigma}_{g_O}^2+ \overline{g}_{O,N}^2$, which converges to $R^O$ in probability. That is, 
\begin{equation}
    R^O:=\mathbb{E}[{g}^2_O]\approx \hat{\sigma}_{g_O}^2+\overline{g}_{O,N}^2
\end{equation} 
holds when $N$ is large.

\subsection{Expectation value estimation for normalized states}

Next, we move to the estimation of ${\rm tr}[OK\rho K^\dagger]/{{\rm tr}[K\rho K^\dagger]}$.
Since we can generate $N$ independent copies of $g_O$ and $g_{\bm{1}}$, Lemma~\ref{lem:ratio_confidenceint} says that for a small $\epsilon>0$ and failure probability $\delta>0$,
\begin{equation}
    N\geq 32\ln \left(\frac{4}{\delta}\right)\times 
        \left\{\frac{R\left[\{S_k\};\rho\right]}{{\rm tr}[\tilde{\Lambda}_{p,\mathcal{U}}(\rho)]^2}\frac{\|O\|^2}{ \epsilon^2 }+\frac{1}{6{\rm tr}[\tilde{\Lambda}_{p,\mathcal{U}}(\rho)]}\frac{\max\{\|O\|^2,\|O\|\}}{\epsilon}\right\}
\end{equation}
suffices to obtain an estimate $\overline{g}_{O,N}/\overline{g}_{\bm{1},N}$ that is $\epsilon$-close to the target value with at least $1-\delta$ probability, that is,
\begin{equation}
    {\rm Pr}\left(\left|\frac{\overline{g}_{O,N}}{\overline{g}_{\bm{1},N}}-\frac{{\rm tr}[OK\rho K^\dagger]}{{{\rm tr}[K\rho K^\dagger]}}\right|\leq \epsilon\right)={\rm Pr}\left(\left|\frac{\overline{g}_{O,N}}{\overline{g}_{\bm{1},N}}-\frac{{\rm tr}[O\tilde{\Lambda}_{p,\mathcal{U}}(\rho)]}{{{\rm tr}[\tilde{\Lambda}_{p,\mathcal{U}}(\rho)]}}\right|\leq \epsilon\right)\geq 1-\delta.
\end{equation}
This completes the proof of Theorem~\ref{thm:2_samplecomplexity}.
Also, in the asymptotic case as $N\to\infty$, Lemma~\ref{lem:asymnormal_ratio} yields
\begin{equation}
    \frac{\sqrt{N}}{\hat{\sigma}_{\rm ratio}}\left(\frac{\overline{g}_{O,N}}{\overline{g}_{\bm{1},N}}-\frac{{\rm tr}[OK\rho K^\dagger]}{{{\rm tr}[K\rho K^\dagger]}}\right)\xrightarrow{d} \mathcal{N}(0,1),
\end{equation}
for $\hat{\sigma}_{\rm ratio}$ defined as
\begin{equation}
    \hat{\sigma}_{\rm ratio}:= \sqrt{\frac{\hat{\sigma}^2_{g_O}}{(\overline{g}_{\bm{1},N})^2} + \frac{(\overline{g}_{O,N} \hat{\sigma}_{g_{\bm{1}}})^2}{(\overline{g}_{\bm{1},N})^4}}.
\end{equation}
Here, $\hat{\sigma}^2_{g_{O}}$ and $\hat{\sigma}^2_{g_{\bm{1}}}$ are the sample variance of $\{g^{(i)}_{O}\}$ and $\{g^{(i)}_{\bm{1}}\}$ defined as in Eq.~\eqref{eq:go_sample_var}.
Consequently, we have the following asymptotic $1-\delta$ confidence interval as $N\to \infty$:
\begin{equation}
    {\rm Pr}\left(\left|\frac{\overline{g}_{O,N}}{\overline{g}_{\bm{1},N}}-\frac{{\rm tr}[OK\rho K^\dagger]}{{{\rm tr}[K\rho K^\dagger]}}\right|\leq \sqrt{\frac{\hat{\sigma}_{\rm ratio}^2 z^2_{\delta/2}}{N}}\right)\to 1-\delta,
\end{equation}
which can be calculated from the obtained samples.
Note that from Eq.~\eqref{eq:hat_ratio_approx}, the following approximation for $\hat{\sigma}_{\rm ratio}$ also holds when $N$ is large:
\begin{equation}
    \hat{\sigma}_{\rm ratio}^2\approx \frac{{\rm Var}[g_O]}{{\rm tr}[\tilde{\Lambda}_{p,\mathcal{U}}(\rho)]^2} +\frac{{\rm tr}[O\tilde{\Lambda}_{p,\mathcal{U}}(\rho)]^2}{{\rm tr}[\tilde{\Lambda}_{p,\mathcal{U}}(\rho)]^4}{\rm Var}[g_{\bm{1}}]\leq 
    \frac{2{R}\left[\{S_k\};\rho\right] \|O\|^2}{{\rm tr}[\tilde{\Lambda}_{p,\mathcal{U}}(\rho)]^2}.
\end{equation}

\section{Detailed derivations for the results for the linear system solver}\label{apdx:QLSS}
We first discuss the projection probability $\mathcal{P}$ for the conventional LCU algorithm for the normalized LCU operator
\begin{equation}
A_{\rm norm}=\frac{1}{\|c\|_1} \frac{i}{\sqrt{2\pi}} \sum_{j=0}^{J-1} \Delta_y \sum_{k=-K}^{K}\Delta_z z_k e^{-z_k^2/2} e^{-i M y_j z_k}.
\label{Eq: Anorm}
\end{equation}
For the conventional LCU method with the input state $\ket{b}$, the projection probability reads
\begin{equation}
\begin{aligned}
\mathcal{P} &= \bra{b} A_{\rm norm}^\dag A_{\rm norm} \ket{b} = \|c\|_1^{-2} (\bra{b}(M^{-1})^2 \ket{b} + \mathcal{O}(\varepsilon)) = \mathcal{O}(\mathrm{log}^{-1} (\kappa/\varepsilon)),
\end{aligned}
\end{equation}
where we used $\bra{b}(M^{-1})^2 \ket{b} = \mathcal{O}(\kappa^{2})$ and ignored the term with the scaling $\mathcal{O}(\varepsilon \kappa^{-2} \mathrm{log}^{-1}(\kappa/\varepsilon))$. Next, we proceed to our intermediate LCU implementation for solving the linear system of equations.
We define $q_j = \beta\Delta_y/(\|c\|_1\sqrt{2\pi})$ and $K_j= i\beta^{-1}\sum_{k=-K}^{K}\Delta_z z_k e^{-z_k^2/2} e^{-i M y_j z_k}$ for the intermediate LCU decomposition $A_{\rm norm}= \sum_j q_j K_j$, where $\beta$ is the normalization factor of $K_j$ given by
\begin{equation}
    \beta=\sum_{k=-K}^K \Delta_z |z_k|e^{-z_k^2/2}\approx 2.
\end{equation}
{\color{black}
We first show that
\begin{equation}\label{apdx_der_K_j}
\begin{aligned}
K_j&= i\beta^{-1}\sum_{k=-K}^{K}\Delta_z z_k e^{-z_k^2/2} e^{-i M y_j z_k}\approx \sqrt{2\pi}\beta^{-1}M y_j e^{-(My_j)^2/2}.
\end{aligned}
\end{equation}
For the eigenvalue $\lambda$ of $M$, we consider the following function $F_K(\lambda)$
\begin{equation}
    F_K(\lambda)=\sum_{k=-K}^K \Delta_z z_ke^{-z_k^2/2}e^{-i\lambda y_jz_k}.
\end{equation}
The difference between $F_{\infty}$ and $F_{K}$ is uniformly upper bounded as
\begin{align}
    |F_{\infty}(\lambda)-F_{K}(\lambda)|\leq \sum_{|k|>K} \Delta_z |z_k|e^{-z_k^2/2}=2\sum_{k=K+1}\Delta_z(\Delta_z k)e^{-(\Delta_z k)^2/2}\leq 2\int_{K\Delta_z}^{\infty} te^{-t^2/2}dt=2e^{-(K\Delta_z)^2/2}.
\end{align}
Due to $K\Delta_z=\Theta(\sqrt{\log(\kappa/\epsilon)})$, we can ignore the truncation error.
Then, the $F_{\infty}$ can be written as
\begin{align}
    F_{\infty}(\lambda)&=\sum_{k=-\infty}^\infty \Delta_z z_ke^{-z_k^2/2}e^{-i\lambda y_jz_k}\notag\\
    &=\sum_{k=-\infty}^{\infty}\int_{-\infty}^{\infty} \Delta_z (\Delta_z x')e^{-(\Delta_z x')^2/2}e^{-i\lambda y_j\Delta_z x'} e^{-2\pi ikx'}dx'\notag\\
    &= \sum_{k=-\infty}^{\infty}\int_{-\infty}^{\infty} ue^{-u^2/2}e^{-i\lambda y_ju} e^{-2\pi i(k/\Delta_z)u}du\notag\\
    &=\sum_{k=-\infty}^{\infty} \hat{a}(\lambda y_j+2\pi k/\Delta_z),
\end{align}
where the second equality follows from the Poisson sum formula, and in the final equality, we defined
\begin{equation}
    \hat{a}(w):=\int dx~ xe^{-x^2/2}e^{-iwx}=-i\sqrt{2\pi}w e^{-w^2/2}.
\end{equation}
We focus on the argument of $\hat{a}(w)$, that is, $w=\lambda y_j+2\pi k/\Delta_z$.
From the assumption, we have
\begin{equation}
    \lambda y_j+2\pi k/\Delta_z = \Theta\left(k\kappa\sqrt{\log(\kappa/\epsilon)}\right),
\end{equation}
and thus, $\hat{a}(\lambda y_j+2\pi k/\Delta_z)$ scales as (ignoring constants)
\begin{equation}
    k\kappa\sqrt{\log(\kappa/\epsilon)}e^{-k^2\kappa^2{\log(\kappa/\epsilon)}}= k\kappa\sqrt{\log(\kappa/\epsilon)}\left(\frac{\varepsilon}{\kappa}\right)^{k^2\kappa^2}.
\end{equation}
Thus, $F_{\infty}(\lambda)$ is well approximated as
\begin{equation}
    F_{\infty}(\lambda)\approx \hat{a}(\lambda y_j)=-i\sqrt{2\pi} \lambda y_j e^{-(\lambda y_j)^2/2},
\end{equation}
and we obtain Eq.~\eqref{apdx_der_K_j} due to the approximation $F_{K}\approx F_{\infty}$ for any $y_j$.}
\purple{Using this result}, the reduction factor reads
\begin{equation}\label{apdx_der_R_int}
\begin{aligned}
R_{\rm int}&= \sum_j q_j \bra{b}K_j^\dagger K_j \ket{b}\approx \sum_j q_j (\sqrt{2\pi}\beta^{-1})^2\bra{b}
(M y_j)^2 e^{-(My_j)^2}\ket{b}\\
&= \frac{\sqrt{2\pi}}{\beta}\frac{1}{\|c\|_1}\left[ \int_0^{\infty}dy~\bra{b} (My)^2 e^{-(My)^2} \ket{b}+\tilde{\mathcal{O}}({\varepsilon})\right]\\
&\approx \frac{\pi\sqrt{2}}{4\beta \|c\|_1} \bra{b} |M|^{-1} \ket{b},
\end{aligned}
\end{equation}
where we used $\int_0^{\infty} dt~(xt)^2e^{-(xt)^2}= \sqrt{\pi}/(4 |x|)$. Here, $|A| = \sqrt{A^\dag A}$ for an operator $A$. 
{\color{black}We can derive the second equality of Eq.~\eqref{apdx_der_R_int} as follows.
Let $G(y)=y^2M^2e^{-y^2M^2}$.
\begin{align}
    &\left\|\sum_{j=0}^{J-1}\Delta_y G(y_j)-\int_{0}^{y_J}dy~ G(y)\right\|=\left\|\sum_{j=0}^{J-1}\Delta_y G(y_j)-\sum_{j=0}^{J-1}\int_{y_j}^{y_j+\Delta_y}dy~G(y)\right\|\notag\\
    &\leq\sum_{j=0}^{J-1} \left\|\int_{0}^{\Delta_y}du~(\Delta_y-u) G'(y_j+u)\right\|\leq \Delta_y \int_{0}^{y_J} du~\|G'(u)\| = \tilde{\mathcal{O}}(\varepsilon).
\end{align}
The truncation error of the integral is given by $\tilde{\mathcal{O}}(\varepsilon)$.}
Because of $\bra{b} |M|^{-1} \ket{b}=\mathcal{O}(\kappa)$, we finally obtain
\begin{equation}
R_{\rm int}= \mathcal{O}(\mathrm{log}^{-1/2} (\kappa/\varepsilon)).
\end{equation}

{\color{black}
Now, we mention the gate complexity. 
Due to Lemma~8 in Ref.~\cite{childs2017quantum}, for implementing the SELECT operation $U_{\rm sel}=\sum_{k=0}^K \ket{k}\bra{k} \otimes V^k$ for a unitary operator $V$, the gate complexity of $U_{\rm sel}$ is $\mathcal{O}(KC)$ with $C$ being the gate complexity of $V$. 
For our case, for simulating Eq.~\eqref{Eq: Anorm} with our intermediate LCU method, we uniformly generate $y_j$ and use the select unitary with $V=e^{-itM}$ for the time $|t|=y_j\Delta_z$ to apply $K_j$. 
Since $\langle y_j \rangle:=\sum_j q_j y_j = \Theta(\kappa \sqrt{\mathrm{log}(\kappa/\varepsilon)})$, the maximal time of the Hamiltonian simulation for our protocol reads $\langle y_j \rangle \Delta_z K =\Theta(\kappa \mathrm{log}(\kappa/\varepsilon))$ in average.
Note that the resulting gate complexity is almost the same as the conventional and random LCU implementations.
}

\end{document}